\documentclass[11pt]{article}
\usepackage{amsmath,amssymb,amsfonts,amsthm,epsfig}
\usepackage[usenames,dvipsnames]{xcolor}

\usepackage{bm,xspace}

\usepackage{cancel}
\usepackage{fullpage}
\usepackage{liyang}
\usepackage{framed}
\usepackage{verbatim}
\usepackage{enumitem}
\usepackage{array}
\usepackage{multirow}
\usepackage{afterpage}
\usepackage{mathrsfs}

\usepackage[normalem]{ulem}

\newcommand{\lnote}[1]{\footnote{{\bf \color{blue}Li-Yang}: {#1}}}
\newcommand{\rnote}[1]{\footnote{{\bf \color{red}Rocco}: {#1}}}
\newcommand{\onote}[1]{\footnote{{\bf \color{green}Ryan}: {#1}}}

\usepackage{caption}

\usepackage{todonotes}

\makeatletter
\newtheorem*{rep@theorem}{\rep@title}
\newcommand{\newreptheorem}[2]{
\newenvironment{rep#1}[1]{
 \def\rep@title{#2 \ref{##1}}
 \begin{rep@theorem}\itshape}
 {\end{rep@theorem}}}
\makeatother
\theoremstyle{plain}



\newcommand{\ignore}[1]{}

\def\colorful{0}

\ifnum\colorful=1
\newcommand{\violet}[1]{{\color{violet}{#1}}}

\newcommand{\blue}[1]{{{\color{blue}#1}}}
\newcommand{\red}[1]{{\color{red} {#1}}}
\newcommand{\green}[1]{{\color{green} {#1}}}

\newcommand{\gray}[1]{{\color{gray}{#1}}}

\fi
\ifnum\colorful=0
\newcommand{\violet}[1]{{{#1}}}

\newcommand{\blue}[1]{{{#1}}}
\newcommand{\red}[1]{{{#1}}}
\newcommand{\green}[1]{{{#1}}}
\newcommand{\gray}[1]{{{#1}}}

\fi

\usepackage{boxedminipage}

\newreptheorem{theorem}{Theorem}
\newtheorem*{theorem*}{Theorem}
\newreptheorem{lemma}{Lemma}
\newreptheorem{proposition}{Proposition}
\newtheorem*{noclaim*}{Claim}

\newcommand{\pparagraph}[1]{\medskip \noindent {\bf {#1}}}

\newcommand{\normal}{{\mathcal{N}(0,1)}}

\newcommand{\hash}{\mathrm{hash}}

\newcommand{\bucket}{\mathrm{bucket}}

\newcommand{\Ind}{\mathds{1}}
\newcommand{\tInd}{\wt{\mathds{1}}}

\newcommand{\CNF}{\mathrm{CNF}}

\newcommand{\inner}{{\mathrm{in}}}
\newcommand{\outter}{{\mathrm{out}}}

\usepackage{chngpage}

\usepackage{dsfont}

\renewcommand{\N}{\mathds{N}}
\renewcommand{\R}{\mathds{R}}
\renewcommand{\Z}{\mathds{N}}

\newcommand{\err}{\mathrm{err}}

\newcommand{\head}{\textsc{Head}}
\newcommand{\tail}{\textsc{Tail}}

\newcommand{\Surface}{\mathrm{Surface}}
\newcommand{\Inside}{\mathrm{Inside}}

\newcommand{\standardized}{-standardized\xspace}

\newcommand{\tay}{d}
\newcommand{\bucks}{L}

\newcommand{\AInf}{{\mathcal E}}   \newcommand{\Ainf}{\AInf}
\newcommand{\AOInf}{{\vec{\mathcal E}}}   \newcommand{\AOinf}{\AOInf}

\newcommand{\BC}{\mathrm{BC}}
\newcommand{\EC}{\mathrm{EC}}
\newcommand{\CE}{\mathrm{CE}}

\newcommand{\myfig}[4]{\begin{figure}[H] \centering \includegraphics[width=#1\textwidth]{#2} \caption{#3} \label{#4} \end{figure}}

\newcommand{\ol}[1]{\overline{#1}}
\newcommand{\bdry}{\partial}

\begin{document}

\title{Fooling Polytopes\vspace*{10pt}}

\author{\hspace{-35pt}Ryan O'Donnell \\
\hspace{-35pt} \small{\sl Carnegie Mellon University} \and Rocco A. Servedio\\ \small{\sl Columbia University} \and Li-Yang Tan   \\ \small{\sl Stanford University}\vspace*{10pt}
}

\date{\today}

\maketitle

\begin{abstract}
We give a pseudorandom generator  that fools $m$-facet polytopes over $\zo^n$ with seed length $\polylog(m) \cdot \log n$. % (equivalently, intersections of $m$ halfspaces over $\{0,1\}^n$).
The previous best seed length had superlinear dependence on~$m$.  An immediate consequence is a deterministic quasipolynomial time algorithm for approximating the number of solutions to any $\zo$-integer program.

\end{abstract}

\thispagestyle{empty}

\newpage
\setcounter{page}{1}

%!TEX root = main.tex

\section{Introduction}

Unconditional derandomization has been a major focus of research in computational complexity theory for more than thirty years.  A significant line of work in this area has been on developing unconditional pseudorandom generators (PRGs) for various types of Boolean functions.  Early seminal results in this vein focused on Boolean circuits~\cite{AW89,Nis91,NW94} and branching programs \cite{Nisan:92,INW94,NZ96}, but over the past decade or so a new strand of research has emerged in which the goal is to construct PRGs against \emph{halfspaces} 
%(i.e.~functions of the form $F(x) = \Ind[w \cdot x \le \theta]$) 
and various generalizations of halfspaces.  This work has included a sequence of successively more efficient PRGs against halfspaces \cite{DGJ+10:bifh,KRS12,MZ13,Kane14-subpoly,KothariMeka15,GKM15},  low-degree polynomial threshold functions \cite{DKN10,Kane11focs,Kane11ccc,MZ13,Kane14-subpoly,KR18}, and, most relevant to this paper,
%(various restricted types of)
\emph{intersections of halfspaces} \cite{GOWZ10,HKM12,ST17,CDS18}.

Since intersections of $m$ halfspaces correspond to $m$-facet polytopes, and also to $\zo$-integer programs with $m$ constraints, these objects are of fundamental interest in high-dimensional geometry, optimization, and a range of other areas. A pseudorandom generator that $\delta$-fools intersections of $m$ halfspaces can equivalently be viewed as an explicit \emph{discrepancy set} for $m$-facet polytopes: a small subset of $\{0,1\}^n$ that $\delta$-approximates the $\{0,1\}^n$-volume of every $m$-facet polytope.  (Discrepancy sets are stricter versions of \emph{hitting sets}, which are only required to intersect every polytope of volume at least $\delta$.)  The problem of constructing a PRG for intersections of $m$ halfspaces is also a stricter version of the algorithmic problem of deterministically approximating the number of solutions of a $\zo$-integer  program with $m$ constraints. It is stricter because a PRG yields an input-oblivious algorithm: the range of a PRG is a single fixed set of points which gives approximately the right answer for \emph{every} $\zo$-integer program.  Beyond pseudorandomness, intersections of halfspaces also play a significant role in other fields such as concrete complexity theory \cite{MinskyPapert:68,BRS:95,OS10:combinatorica,Sherstov13sicomp,Sherstov13combinatorica,Kane14intersection} and computational learning theory \cite{BlumKannan:97,KOS:04,KlivansSherstov:06,KOS:08,Vempala:10,KhotSaket:11jcss,GKM12,ST17itcs}.

The main result of this paper is a new PRG for intersections of $m$ halfspaces.  Its seed length grows polylogarithmically with $m$, which is an exponential improvement of the previous best PRG for this class.
Before giving the precise statement of our result, we briefly describe the prior state of the art for this problem.

\subsection{Prior work on PRGs for intersections of halfspaces}

A halfspace $F(x) = \Ind[w \cdot x \leq \theta]$ is said to be \emph{$\tau$-regular} if
%\onote{changing this to the classic definition}
$|w_j| \leq \tau \|w\|_2$ for all $j \in [n]$; %$\sum_{{j}=1}^n w_j^4 \leq \tau^2 \big(\sum_{{j}=1}^n w_j^2\big)^2$;
intuitively, a $\tau$-regular halfspace is one in which no coefficient $w_j$ is too large relative to the overall scale of all the coefficients. Harsha, Klivans, and Meka~\cite{HKM12} gave a PRG which
% \ignore{
%%$\tilde{O}((\log m)^{8/5} \tau^{1/5})$-fools any intersection of $m$ many $\tau$-regular halfspaces; their PRG has seed length $O((\log n \log m)/\tau)$.\rnote{Let's rephrase maybe - this superficially makes them look like they're linear in $\log m$.}
%}
$\delta$-fools any intersection of $m$ many $\tau$-regular halfspaces with seed length $\poly(\log m,1/\delta)\cdot \log n$, where $\tau$ has to be sufficiently small relative to $m$ and $\delta$ (specifically, $\tau \le \text{some } \poly(\frac{\delta}{\log m})$ is required).
%\lnote{used to be ``any intersection of $m$ many $\delta^5/(\log^{8.1} m)(\log(1/\delta))$-regular halfspaces with seed length $\log (n) \cdot   \log^{9.1} (m) \cdot \log(1/\delta)/\delta^5$"}
%we will discuss the techniques of \cite{HKM12} in considerable detail below).
While this seed length has the desirable property of being polylogarithmic in $m$, due to the regularity requirement this result cannot be used to fool intersections of even two general halfspaces.  We note that there are very basic halfspaces, such as $F(x) = \Ind[x_1 \leq 1/2]$, that are highly irregular.

Recently,~\cite{ST17} built on the work of~\cite{HKM12} to give a PRG that fools a different subclass of intersections of halfspaces.  They give a PRG that $\delta$-fools any intersection of $m$ many \emph{weight-$W$} halfspaces with seed length $\poly(\log m, W, 1/\delta)\cdot \polylog\,n$; a halfspace has weight $W$ if it can be expressed as $\Ind[w \cdot x \le \theta]$ where each coefficient $w_j$ is an integer of magnitude at most $W$.  Unfortunately, many $n$-variable halfspaces require weight polynomially or even exponentially large in $n$; in fact, a counting argument shows that almost all halfspaces require exponentially large weight.
%\rnote{I'm not so crazy about the proposed second half of this sentence, \violet{and note that \cite{ST17}'s seed length depends polynomially on $W$} --- to my ear it's awkward for the first half of the sentence to be declarative (``many $n$-variable halfspaces\dots'') and the second half to be addressing the reader directly (``note that\dots''), and also the immediately previous sentence just said that it's polynomial in $W$ so I don't think we need to repeat this.)}
Therefore, the \cite{ST17} result also cannot be used to fool even two general halfspaces.

In \cite{GOWZ10}, Gopalan, O'Donnell, Wu, and Zuckerman gave a PRG that \emph{can} fool intersections of $m$ general halfspaces.  However, various aspects of their approach
% \gray{(a generalization of structural results for single halfspaces to $m$-tuples of halfspaces, and read-once branching program based techniques applied to $m$-tuples of halfspaces)}\lnote{How about just saying "various aspects of their approach" / ``a couple of aspects", without going into any detail?  Not sure which way is better.}
each necessitate a seed length which is at least linear in $m$, and indeed their overall seed length is $O((m \log(m/\delta) + \log n) \cdot \log(m/\delta))$.\footnote{Their seed length improves to $O(m \log(m/\delta) + \log n)$ if $m/\delta$ is bounded by any $\polylog(n)$.}
%\rnote{\gray{This sentence is arguably giving [GOWZ] somewhat short shrift, since it doesn't mention that for $m/\delta \leq$ any $\polylog(n)$, their seed length is $O(\log(n) + m \log(m/\delta)).$  Should we mention this, and add it to the table, perhaps with a second line?  The table could say
%
%``$O(\log(n) + m \log(m/\delta)),$ if $m/\delta \leq$ any $\polylog(n)$''
%
%(maybe we should add the word ``some'' before ``$\poly(\delta/\log m)$'' in the [HKM12] line?)
%
%(I love the table, by the way!)}}
 So while this PRG is notable for being able to handle intersections of general halfspaces, its seed length becomes trivial (greater than $n$) for intersections of $m \geq n$ many halfspaces.  (Indeed, this PRG of~\cite{GOWZ10}  fools \emph{arbitrary monotone functions} of $m$ general halfspaces, with intersections (i.e.~{\sc And}s) being a special case.  Due to the generality of this class---which of  course includes every monotone function over $\zo^m$---it can be shown that any PRG has to have at least linear seed length dependence on $m$.)

\subsubsection{PRGs over Gaussian space}

\noindent There has also been work on PRGs for functions over $\R^n$ endowed with the $n$-dimensional Gaussian distribution.  Analyses in this setting are often facilitated by the continuous nature of $\R^n$ and rotational invariance of the Gaussian distribution, useful technical properties not afforded by the standard  setting of Boolean space.  For halfspaces and polytopes, PRGs over Gaussian space can be viewed as a first step towards PRGs over Boolean space; as we describe below, Boolean PRGs even for restricted subclasses of halfspaces and polytopes yield Gaussian PRGs for general halfspaces and polytopes, but the converse does not hold.  We also note that the correspondence between polytopes and $\zo$-integer programs is specific to Boolean space, and in particular, Gaussian PRGs do not yield algorithms for counting solutions to these programs.

For halfspaces, Meka and Zuckerman~\cite{MZ13} showed that any PRG for the subclass of $O(\frac1{\sqrt{n}})$-regular halfspaces  over Boolean space yields a PRG for all halfspaces over Gaussian space.
 Note that $O(\frac1{\sqrt{n}})$-regular halfspaces are ``the most regular" ones; every halfspace is $\tau$-regular for some $\tau \in [\frac1{\sqrt{n}},1]$.  \cite{HKM12} generalized this connection to polytopes: they showed that any PRG for intersections of $m$ many $O((\log m)/\sqrt{n})$-regular halfspaces over Boolean space yields a PRG for intersections of $m$ many arbitrary halfspaces over Gaussian space. Combining this with their Boolean PRG for intersections of regular halfspaces discussed above,~\cite{HKM12} obtained a Gaussian PRG for intersections of $m$ halfspaces with seed length $\poly(\log m,1/\delta)\cdot \log n$.  Recent work of~\cite{CDS18} gives a different Gaussian PRG with seed length $\poly(\log m,1/\delta) + O(\log n)$.

 The focus of the current work is on PRGs over Boolean space, and the rest of the paper addresses this (more challenging) setting.

\subsection{This work:  A PRG for intersections of general halfspaces}
\label{sec:this-work}

\begin{table}[t]
\begin{adjustwidth}{-2em}{}
\renewcommand{\arraystretch}{1.7}
\centering
\begin{tabular}{|c|c|l|}
\hline
  Reference  & Function class  & \multicolumn{1}{c|}{Seed length of PRG}   \\ \hline
 \cite{GOWZ10}   & Monotone functions of $m$ halfspaces &
 \begin{tabular}{@{}l@{}}
$O((m \log(m/\delta) + \log n) \cdot \log(m/\delta))$  \vspace{-6pt} \\
$O(m \log(m/\delta) + \log n),$ \text{if $m/\delta \le \text{any } \polylog(n)$}
 \end{tabular}
 \\  [1.2em] \hline
\cite{HKM12} & Intersections of $m$ $\tau$-regular halfspaces &
$\poly(\log m,1/\delta)\cdot \log n,$ \text{if $\tau \le \text{some }\poly(\frac{\delta}{\log m})$}
 \\ [.2em] \hline
 \cite{ST17} &  Intersections of $m$ weight-$W$ halfspaces  &
 $\poly(\log m, W, 1/\delta)\cdot \polylog\,n$ \\ [.2em] \hline \hline
  {\bf This work} & Intersections of $m$ halfspaces & $\poly(\log m, 1/\delta)\cdot \log n$  \\ [.2em] \hline
\end{tabular}
\caption{PRGs for intersections of halfspaces over $\zo^n$}
\label{prior}
\end{adjustwidth}
\end{table}

Summarizing the prior state of the art on PRGs over Boolean space, there were no PRGs that could fool intersections of $m = n$ many general halfspaces, and relatedly, the best PRG for intersections of $m\le n$ general halfspaces had a superlinear seed length dependence on $m$.  The PRGs that could fool intersections of $m\ge n$ halfspaces imposed technical restrictions on the halfspaces: either regularity (hence excluding simple halfspaces such as $\Ind[x_1 \le 1/2]$), or small weights (hence excluding almost all halfspaces). Please refer to \Cref{prior}.

%In fact, there were no non-trivial  algorithms (running in time $< 2^n$) for deterministically approximating the number of satisfying assignments of an intersection of $m=n$ general halfspaces.

  The main result of this paper is a PRG which fools intersections of $m$ general halfspaces with a polylogarithmic seed length dependence on $m$:

\begin{theorem}[PRG for polytopes]
\label{thm:main}
For all $n,m\in \N$ and $\delta \in (0,1)$, there is an explicit pseudorandom generator with seed length $\poly(\log m,1/\delta)\cdot \log n$ that $\delta$-fools the class of intersections of~$m$ halfspaces over $\zo^n$.
\end{theorem}

In particular, this PRG fools intersections of $\quasipoly(n)$ many halfspaces with seed length $\polylog(n)$, and its seed length remains non-trivial for intersections of exponentially many halfspaces ($\exp(n^c)$ where $c > 0$ is an absolute constant). %See~\Cref{fact:ourseedlength} and~\Cref{eq:sl} for a detailed statement of our seed length.

An immediate consequence of~\Cref{thm:main} is a deterministic algorithm that runs in time $n^{\polylog(m)}$ and additively approximates the number of solutions to any $n$-variable $\zo$-integer program with $m$ constraints.   Prior to our result, no non-trivial deterministic algorithm (running in time $< 2^n$) was known even for general $\zo$-integer programs with $m=n$ constraints.
\Cref{thm:main} also yields PRGs with comparable seed lengths for intersections of halfspaces over a range of other domains, such as the $n$-dimensional hypergrid $\{ 0,1,\ldots, N\}^n$ and the solid cube $[0,1]^n$ (details are left to the interested reader).

Our proof of \Cref{thm:main} involves several novel extensions of the central technique driving this line of work, namely Lindeberg-style proofs of probabilistic invariance principles and derandomizations thereof.  We develop these  extensions to overcome challenges which arise due to the generality of our setting; specifically, the fact that we are dealing with intersections of  arbitrary halfspaces, with no restrictions whatsoever on their structure.  One of the key new ingredients in our analysis, which we believe is of independent interest,  is a sharp high-dimensional generalization of the classic Littlewood--Offord anticoncentration inequality~\cite{LO:43,Erd:45} that we establish. We now describe our proof and the new ideas underlying it in detail.

\section{Overview of our proof}
\label{sec:overview}

\subsection{Background: the~\cite{HKM12} PRG for regular polytopes}

We begin by recalling the arguments of Harsha, Klivans, and Meka~\cite{HKM12} for fooling regular polytopes. At a high level,~\cite{HKM12} builds on the work of Meka and Zuckerman~\cite{MZ13}, which gave a versatile and powerful framework for constructing pseudorandom generators from probabilistic \emph{invariance principles}; the main technical ingredient underlying the~\cite{HKM12} PRG for regular polytopes is a new invariance principle for such polytopes, which we now describe.

\pparagraph{\cite{HKM12}'s invariance principle and the Lindeberg method.}   At a high level, the~\cite{HKM12} invariance principle for regular polytopes is as follows: given an $m$-tuple of regular linear forms over $n$ input variables $x=(x_1,\dots,x_n)$ (denoted by $Ax$, where $A$ is an $m$-by-$n$ matrix), the distribution (over $\R^m$) of $A \bu$, where $\bu \sim \bn$ is uniform random, is very close to the distribution of  $A \bg$, where $\bg \sim \normal^n$ is distributed according to a standard $n$-dimensional Gaussian.  Here closeness is measured by \emph{multidimensional CDF distance}; we observe that multidimensional CDF distance corresponds to test functions of the form $\Ind[Ax \leq b]$ where $b \in \R^m$, which synchs up precisely with an intersection of $m$ halfspaces $\Ind[A_1 x \leq b_1]  \wedge \cdots \wedge \Ind[A_m x \leq b_m].$    To prove this invariance principle,~\cite{HKM12} employs the well-known Lindeberg method (see e.g.~Chapter \S11 of~\cite{ODbook} and~\cite{Tao254A}) and proceeds in two main conceptual steps.  The first step establishes a version of the result for \emph{smooth test functions}, proxies for the actual ``hard threshold" test functions $\Ind[Ax \le b]$, and the second step relates distance with respect to these smooth test functions to multidimensional CDF distance via \emph{Gaussian anticoncentration}. We outline each of these two steps below.

The first step is to prove an invariance principle for \emph{smooth test functions}.   Here instead of measuring the distance between $A\bu$ and $A\bg$ using test functions that are orthant indicators $\calO_b(v_1,\ldots,v_m) = \Ind[v \le b]$ (corresponding to multidimensional CDF distance), distance is measured using a sufficiently smooth \emph{mollifier} $\wt{\calO}_b : \R^m \to [0,1]$ of~$\calO_b$.  Such mollifiers, with useful properties that we now discuss, were proposed and analyzed by Bentkus~\cite{Bentkus:90}.   In more detail, \cite{HKM12} prove that the difference between the expectations of $\wt{\calO}_b(A\bu)$ and $\wt{\calO}_b(A\bg)$ is bounded by a certain function involving $\wt{\calO}_b$'s derivatives. In fact, as in standard in Lindeberg-style proofs of invariance principles,~\cite{HKM12} actually bounds this difference with respect to \emph{any} smooth test function $\Upsilon : \R^m \to \R$ in terms of $\Upsilon$'s derivatives; the only specific property of Bentkus's mollifier $\wt{\calO}_b$ that is used is that its derivatives are appropriately small.  At a high level, the proof of this smooth invariance principle proceeds by hybridizing from $\Upsilon(A\bu)$ to $\Upsilon(A\bg)$, using the multidimensional Taylor expansion of $\Upsilon$ to bound the error incurred in each step. (The regularity of the linear forms is used in a crucial way to control the approximation error that results from truncating the Taylor expansion at a certain fixed degree.)
%\rnote{I like Li-Yang's rephrasing of this sentence; as a cheap hack, I wonder if we should stress/mention that regularity is crucial, to swell the applause when we amazingly sidestep it later.  Feel free to delete the violet if you don't like it.}

The second step is to establish the desired bound on multidimensional CDF distance using the aforedescribed smooth invariance principle applied to Bentkus's mollifier.  This step relies on a second key property of Bentkus's mollifier: $\wt{\calO}_b$ agrees with the orthant indicator $\calO_b$ except on a small error region near the orthant boundary.  With this property in hand, a fairly simple and standard argument shows that it suffices to bound the \emph{anticoncentration} of the \emph{Gaussian} random variable $A \bg$;  intuitively, such anticoncentration establishes that $A\bg$ does not place too much probability weight on the error region where $\wt{\calO}_b$ disagrees with $\calO_b$.  In \cite{HKM12}, the required anticoncentration for $A \bg$ follows immediately from a result of Nazarov~\cite{Nazarov:03,KOS:08} on the Gaussian surface area of $m$-facet polytopes.

\pparagraph{The~\cite{HKM12} PRG via  a derandomized invariance principle.}  Having proved this invariance principle for regular polytopes,~\cite{HKM12} then establish a \emph{pseudorandom} version by derandomizing its proof.   That is, they argue that their proof in fact establishes multidimensional-CDF-closeness between $A \bz$ and $A\bg$, where $\bg \sim \calN(0,1)^n$ is distributed according to a standard Gaussian as before, but $\bz \sim\bn$ is the output of a suitable pseudorandom suitable generator $\mathscr{G} : \bits^r \to \bn$ (rather than uniform random).   Combining the ``full-randomness" invariance principle (establishing closeness between $A\bu$ and $A\bg$) with this pseudorandom version (establishing closeness between $A\bz$ and $A\bg$), it follows from the triangle inequality that $A\bz$ and $A\bu$ are close.  Recalling that multidimensional CDF distance corresponds to test functions of the form $\Ind[Ax \le b] = \Ind[A_1 x \leq b_1] \wedge \cdots \wedge \Ind[A_m x \leq b_m]$, this is precisely equivalent to the claim that $\mathscr{G}$ fools the intersection of $m$ halfspaces with weight matrix $A \in \R^{m\times n}$ (and an arbitrary vector of thresholds $b \in \R^m$).

For later reference, we close this section with an informal description of the~\cite{HKM12} generator (for fooling  intersections of $m$ many $\tau$-regular halfspaces):
%\rnote{Ryan had pointed out that the description below is not quite precise because [HKM] use the 4-norm notion of regularity while we've defined regularity in the $\infty$-norm way.  I remember pooh-poohing this concern, saying that we say "informal description".  I'm not very consistent, though; reading this now, the description comes across as completely formal, and indeed is slightly inaccurate because of this issue.
%
%What do you think about replacing ``$L := 1/\tau$'' buckets" with ``$L := \poly(1/\tau)$'' buckets?  That is accurate (right?), and is somewhat more informal.
%Or we could just leave it as it is?}
\begin{enumerate}
\item Pseudorandomly hash the $n$ variables into $\bucks \coloneqq {\poly(1/\tau)}$ buckets using an $(r_\hash \coloneqq 2\log m)$-wise uniform hash function $\bh : [n] \to [\bucks]$.
\item Independently across buckets, assign values to the variables within each bucket using an $(r_\bucket \coloneqq 4\log m)$-wise uniform distribution.
\end{enumerate}
We remark that this is the structure of the Meka--Zuckerman generator~\cite{MZ13} for fooling a single regular halfspace, the only difference being that the relevant parameters $\bucks,r_\hash,$ and $r_\bucket$ are larger in~\cite{HKM12} than in~\cite{MZ13} (naturally so, given that the~\cite{HKM12} generator fools intersections of $m$ regular halfspaces instead of a single one).

Our analysis in this paper can be used to show that the~\cite{MZ13} generator, instantiated with suitable choices of $\bucks, r_\hash$, and $r_\bucket$, fools intersections of $m$ general halfspaces.  However, for technical reasons (that are not essential for this high-level discussion), this results in a seed length that is $\poly(\log m,1/\delta, \log n)$.  To achieve our seed length of $\poly(\log m,1/\delta)\cdot \log n$, we slightly extend the~\cite{MZ13} generator in two ways. First, within each bucket the variables are assigned using an $r_\bucket$-wise uniform distribution {\sc Xor}-ed with an independent draw from a generator that fools small-width CNF formulas~\cite{GMR13}.  Second, we {\sc Xor} the entire resulting $n$-bit string with an independent draw from a $k$-wise independent generator.
%\lnote{Should we say something about the global XOR, or no?}\rnote{Probably we should make our description of the generator here synch up with the actual generator --- we don't want people referring back and forth between this description and the one in Section 4 and getting stressed/confused because of an inconsistency -- so I'd say yes, let's briefly mention the global XOR.  How about something like
%
%``
%\violet{To achieve our seed length of $\poly(\log m,1/\delta)\cdot \log n$, we slightly extend the~\cite{MZ13} generator in two ways. First, within each bucket the variables are assigned using an $r_\bucket$-wise uniform distribution {\sc Xor}-ed with an independent draw from a generator that fools small-width CNF formulas~\cite{GMR13}.  Second, we {\sc Xor} the entire resulting $n$-bit string with an independent draw from a $k$-wise independent generator.}
%''}
(See~\Cref{sec:theprg} for a detailed description of our PRG.)

%%\bigskip
%%
%%\bigskip
%%
%%
%%\gray{
%%%Finally, once \cite{HKM12} have established the invariance principle, it's not too difficult to derandomize their proof and show that the \cite{MZ13} generator suffices.
%%
%%
%%%Explain the basic \cite{MZ13} generator (at a high level).\rnote{Do we need to explain the \cite{MZ13} generator here for this section to make sense?  I think maybe yes...?}
%%
%%
%%Dear reader, the actual generator of \cite{HKM12} and \cite{ST17} are both just the same as the \cite{MZ13} generator with ramped-up parameters; increasingly sophisticated analyses are employed in these later works.
%%
%%For our results, we could also use the \cite{MZ13} generator with ramped-up parameters, but this would give a seed length $\poly(\log n, \log m, 1/\delta)$.  In order to achieve a seed length of only $\log(n) \cdot \poly(\log m, 1/\delta)$, we use a slight variant of the \cite{MZ13} generator.  But for the rest of this section the reader is encouraged just to think about the \cite{MZ13} generator.
%%}

\subsection{Some key new ingredients in our analysis}
\label{sec:new-ingredients}

A fundamental challenge in extending the~\cite{HKM12} PRG result from regular to general polytopes stems from the fact that  an invariance principle \emph{simply does not hold} for general polytopes $Ax \le b$.  Without the regularity requirement on $A$, it is not true that $A\bu$ and $A\bg$ are close in CDF distance; indeed, even a single non-regular linear form such as $x_1$ is distributed very differently under $\bu \sim \bn$ versus $\bg \sim \calN(0,1)^n$.  This therefore necessitates a significant conceptual departure from the Meka--Zuckerman framework for constructing pseudorandom generators from invariance principles: rather than establishing closeness between $A\bu$ and $A\bz$  (where $\bz\sim\bn$ is the output of a suitable pseudorandom generator) through $A\bg$ by means of an invariance principle,  one has to establish closeness between $A\bu$ and $A\bz$ ``directly" without using invariance.

Somewhat surprisingly, even though an invariance principle does not hold in our setting of general polytopes, our proof nonetheless proceeds via the Lindeberg method for proving invariance principles.   Following the two main conceptual steps of the  method (as outlined in the previous section), we first prove that $A\bu$ and $A\bz$ are close with respect to Bentkus's smooth mollifiers $\wt{\calO}_b $ for the orthant indicators $\calO_b$, and then use this to establish closeness in multidimensional CDF distance.  However, the fact that we are dealing with matrices $A \in \R^{m\times n}$ whose rows are \emph{arbitrary} linear forms (corresponding to the facets of general $m$-facet polytopes) instead of regular linear forms poses significant challenges in both steps of the Lindeberg method.     We discuss some of these challenges, and the new ideas that we employ to overcome them, next.  For concreteness we will discuss these challenges and new ingredients by contrasting our proof with that of~\cite{HKM12}, but we remark here that these are in fact qualitative differences between our approach and the Lindeberg method in general.

\pparagraph{Step 1: Fooling Bentkus's mollifier.} Recall that~\cite{HKM12} first proves a general invariance principle establishing closeness in expectation (with a quantitative bound that depends on $\Upsilon$'s derivatives) between $\Upsilon(A\bu)$ and $\Upsilon(A\bg)$ for \emph{any} smooth test function $\Upsilon$.  They then apply this general invariance principle with Bentkus's orthant mollifier $\wt{\calO}_b$ being the test function, using the bounds on $\wt{\calO}_b$'s derivatives established in~\cite{Bentkus:90} but no other properties of $\wt{\calO}_b$.  %\rnote{maybe replace ``as a black box'' with ``but no other properties of $\tilde{\calO}_b$''?}.

In contrast, we do not prove closeness between $A\bu$ and $A\bz$ for all smooth test functions; our argument is carefully tailored to Bentkus's specific mollifier.  In addition to bounds on $\wt{\calO}_b$'s derivatives, we crucially rely on the specific structure of $\wt{\calO}_b$, in particular, the fact that it is the product of $m$ univariate functions, one for each coordinate (i.e.~$\wt{\calO}_b(v) = \prod_{i=1}^m \psi_{b_i}(v_i)$, where each $\psi_{b_i}$ maps $\R$  to $[0,1]$).
%\rnote{We've now moved away from $\Psi_{b_i}$ to $\tilde{\Ind}_{b_i,\lambda}$, but I think it's better to just write ``$\psi_{b_i}$'' here since we haven't introduced the $\tilde{\Ind}$ notation yet and it would be a digression to introduce it here.}
A high-level intuition for why such product structure is useful is as follows.  By doing some structural analysis of halfspaces (see \Cref{sec:reduction}), we can decompose each of our $m$ halfspaces into a small ``head'' portion, consisting of at most $k$ variables, and a remaining ``tail'' portion which is regular.   From this point of view, the difference between regular and general polytopes is therefore the presence of these size-at-most-$k$ head portions in each of the $m$ halfspaces.  Very roughly speaking, the product structure of~$\wt{\calO}_b$ allows us to handle these head portions using pseudorandom generators for \emph{small-width CNF formulas}~\cite{GMR13}. (To see the relevance of CNF formulas in this context, at least at a conceptual level, observe that a product of $\{0,1\}$-valued $k$-juntas is a width-$k$ CNF formula.)

Our proof incorporates these PRGs for CNFs into~\cite{HKM12}'s analysis of the regular tail portions.  We highlight one interesting aspect of our analysis: In all previous instantiations of the Lindeberg method that we are aware of,   expressions like $|\E[\Upsilon(\bv + \bDelta)]  -\E[\Upsilon(\bv + \bDelta')]|$ are bounded by considering two Taylor expansions of $\Upsilon$, both taken around the ``common point'' $\bv$.  Lindeberg method arguments analyze the difference of these Taylor expansions using moment-matching properties of $\bDelta$ and $\bDelta'$
%(typically using the fact that each of these random variables is independent of $\bv$)
and the fact that they are ``small" in a certain technical sense, which is directly related to the regularity assumptions that underlie these invariance principles.  In contrast, in our setting, since we are dealing with arbitrary linear forms rather than regular ones, we end up having to bound expressions like $|\E[\Upsilon(\bv + \bDelta)] - \E[\Upsilon(\bv' + \bDelta')]|$.  Note that this involves considering the Taylor expansions of $\Upsilon$ around two \emph{distinct} points $\bv$ and $\bv'$, which may be far from each other --- indeed, \emph{a priori} it is not even clear that $|\E[\Upsilon(\bv)]-\E[\Upsilon(\bv')]|$ will be small.
%(It also turns out that in our setting $\bv$ and $\bDelta$ are not independent, and neither are $\bv'$ and $\bDelta'$.) \rnote{Is this too much of a fast one, to say this but not point out that there was non-independence in \cite{ST17}?}
  Because of these differences from the standard Lindeberg scenario, moment-matching properties of $\bDelta$ and $\bDelta'$ and their ``smallness" no longer suffice to ensure that the overall expected difference is small.  Instead, as alluded to above, our analysis additionally exploits the product structure of Bentkus's mollifier via PRGs for CNFs to bound $|\E[\Upsilon(\bv + \bDelta)] - \E[\Upsilon(\bv' + \bDelta')]|$ (see \Cref{sec:fool-bentkus}).\ignore{\lnote{Rocco makes a good point about how $\bv$ and $\bDelta$ are also correlated in~\cite{ST17}, and so we cannot claim that this challenge is new to us.  I personally prefer to not mention this at all then.  It sort of distracts from our main point here, Tayloring around distinct points, anyway. \red{Rocco:  I'm happy with the treatment of this given above.}}}

\ignore{
%\gray{
%in all previous instantiations of the Lindeberg method that we are aware of, one argues that $|\E[\psi(\bx)]  -\E[\psi(\bx + \Delta(\bx))]|$ is small by applying (multidimensional) Taylor's theorem to $\psi$ at the point $\bx$ and at  a ``slight perturbation'' of it, $\bx + \Delta(\bx)$, and cancelling out various terms to show that the resulting difference of expectations has small magnitude.  In contrast, in our analysis because of the head portions of the linear forms we must apply multidimensional Taylor's theorem to two distinct points which may be quite far from each other (see the left-hand size of Equation~(\ref{eq:single-swap}).  In such a setting it is \emph{a priori} unclear that any useful cancellation can take place and that the expected difference of $\psi$ at these two distinct points will be small; as alluded to above, results on fooling CNF formulas play a crucial role in enabling us to argue that this is indeed the case.
%}
}

\pparagraph{Step 2: Anticoncentration.} The next step is to pass from closeness of $\wt{\calO}_b(A\bu)$ and $\wt{\calO}_b(A\bz)$  in expectation, to closeness of $A\bu$ and $A\bz$ in multidimensional CDF distance.   We recall that in the analogous step in~\cite{HKM12}'s proof, the starting point was  closeness in expectation of $\wt{\calO}_b(A\bu)$ and $\wt{\calO}_b(A\bg)$, where $\bg \sim \calN(0,1)^n$ is a standard Gaussian (instead of $\wt{\calO}_b(A\bz)$ where $\bz\sim\bn$ is pseudorandom).  For this reason, it sufficed for~\cite{HKM12} to bound the Gaussian anticoncentration of $A\bg$, and as mentioned, such a bound is an immediate consequence of Nazarov's bound on the Gaussian surface area of $m$-facet polytopes.

In contrast, since the Gaussian distribution does not enter into our arguments at all (by necessity, as explained above), we instead have to bound the \emph{Boolean} anticoncentration of $A\bu$ where $\bu \sim \bn$ is uniform random.   This task, which is carried out in~\Cref{sec:anticonc}, requires significantly more work; indeed, Boolean anticoncentration formally contains Gaussian anticoncentration as a special case.  At the heart of our arguments for this step is a new Littlewood--Offord-type anticoncentration inequality for \mbox{$m$-facet} polytopes, a high-dimensional generalization of the classic Littlewood--Offord theorem~\cite{LO:43,Erd:45}.  We discuss this new theorem, which we believe is of independent interest, next.

\subsubsection{A Littlewood--Offord theorem for polytopes}
\label{sec:LO}

 We first recall the classic Littlewood--Offord anticoncentration inequality.
% \ignore{\onote{{I would hesitate to call this the Littlewood--Offord theorem per se.  LO/Erdos show this result not for equalling some $\theta$.  Rather, they assume all weights are at least $1$ in abs-value and then bound the probability of hitting an interval of width~$2$ (stated as ``width~$1$'' when you're considering $0$-$1$ sums rather than $\pm 1$ sums).  Formally, the latter statement easily implies the former, but the reverse looks doubtful to me. See, e.g., Chapter~4 of Bollobas's ``Combinatorics''.  Speaking of which, we should probably compare to Kleitman's 1970 high-dimensional Littlewood--Offord, which if I understand correctly, says that in our setting, if you assume every column of~$A$ has $2$-norm at least~$1$, then the probability of falling into any open ball of radius~$2$ is at most $O(1/\sqrt{n})$ (in fact $2^{-n} \binom{n}{n/2}$ again). (Note: any detail I wrote in this paragraph has a 50\% chance of being inaccurate.)} \blue{Li-Yang: I've incorporated this.}}
% }

 \begin{theorem}[Littlewood--Offord]
 \label{thm:LO}
 For all $\theta \in \R$ and $w \in \R^n$ such that $|w_j| \ge 1$ for all $j\in [n]$,
 \[ \Pr\bracks*{w\cdot \bu \in (\theta-2,\theta\,]} = O\parens*{\frac1{\sqrt{n}}},  \]
 where $\bu\sim\bn$ is uniformly random.
 \end{theorem}

Littlewood and Offord~\cite{LO:43} first proved a bound of $O((\log n)/\sqrt{n})$; Erd\"os~\cite{Erd:45} subsequently sharpened this to  $O(1/\sqrt{n})$, which is optimal by considering $w = 1^n$ and $\theta = 0$. (We observe that the question trivializes without the assumption on the magnitudes of $w$'s coordinates; for instance, the relevant probability is $1/2$ for $w = (1,0,\ldots,0)$ and $\theta = 1$.)

\Cref{thm:LO} has the following natural geometric interpretation: the maximum fraction of hypercube points  that can fall within the ``width-$2$ boundary" of a halfspace $\Ind[w\cdot x \le \theta]$ where $|w_j|\ge 1$ for all $j$ is $O(1/\sqrt{n})$.  Given this geometric interpretation, it is natural to seek a generalization from single halfspaces (i.e.~$1$-facet polytopes) to $m$-facet polytopes:
\begin{quote}
\emph{What is the maximum fraction of hypercube points $u \in \bn$ that can lie within the ``width-$2$ boundary" of an $m$-facet polytope $A x \le b$ where\ignore{
%\rnote{The notation was ``$A_j^i$'' but $A_{ij}$ is more standard, no?  \blue{Li-Yang: Yeah, definitely.  But $A^i_j$ is slightly more consistent with our $A^B$ notation.  If we go with the more standard $A_{ij}$ let's make sure to change the highlighted equation after~\Cref{thm:fool-bentkus}, the definition of $H$ and $T$.} \red{Rocco: I changed those.  I think readers will have imbibed the $A_{ij}$ notation with their mothers' milk, and at this point we haven't introduced our $A^B$ notation yet, so let's go with $A_{ij}$.}
} $|A_{ij}| \geq 1$ for all $i$ and $j$? }
\end{quote}
In more detail, we say that $u$ lies within the ``width-$2$ boundary" of the polytope $Ax \le b$ provided $Au \le b$ and $A_i \cdot  u > b_i - 2$ for some $i\in [m]$; equivalently, $u$ lies in the difference of the two polytopes $Ax \le b$ and $Ax \le b - 2\cdot \Ind_m$, where $\Ind_m$ denotes the all-1's vector in $\R^m$.
The Littlewood--Offord theorem (\Cref{thm:LO}), along with a naive union bound, implies a bound of $O(m/\sqrt{n})$; we are not aware of any improvement of this naive bound prior to our work.
\ignore{
%\rnote{Can we add one or two sentences giving the main point of those high-dimensional generalizations? \blue{Li-Yang: I've incorporated this; see previous footnote above}}\onote{{I'm mildly nervous about highlighting the following precise form.  Didn't we have some kind of cheesy reduction from Kane that did something for the problem of landing \emph{exactly} on the boundary?  Maybe if you weren't too concerned about losing probability from strings $x$ that a hugely skewed number of $\pm 1$'s?  I vaguely remember this from the cafe where Witmer got free pastries. I think I'd be more confident in the novelty and also in the parallel with the actual Littlewood--Offord if we instead stated our ``clean'' theorem as: If all the weights are at least 1 in abs value, then the probability of following within a ``width-1 strip [you know what I mean]'' is at most $O(\sqrt{\log m}/\sqrt{n})$.} \blue{Li-Yang: I've incorporated this.}}
}

%\rnote{removed ``In this paper'' starting the sentence, since to my ear that's what one says when one is introducing the main result of a paper}
We give an essentially complete answer to this question, with upper and lower bounds that match up to constant factors.  In~\Cref{sec:anticonc}  we prove the following ``Littlewood--Offord theorem for polytopes'':

\begin{theorem}[Littlewood--Offord theorem for polytopes]
\label{thm:LO-for-polytopes}
There is a universal constant $C$ ($C = 5\sqrt{2}$ suffices) such that the following holds.  For all $b \in \R^m$ and $A \in \R^{m\times n}$ with
$|A_{ij}| \geq 1$ for all $i \in [m]$ and $j \in [n]$,
\[ \Pr\bracks*{A\bu \le b\ \& \   A_i \cdot \bu > b_i-2 \text{ for some $i\in [m]$}} \leq \frac{C\sqrt{\ln m}}{\sqrt{n}},\]
where $\bu \sim\bn$ is uniformly random.
\end{theorem}

Our proof of~\Cref{thm:LO-for-polytopes} draws on and extends techniques from Kane's bound on the Boolean average sensitivity of $m$-facet polytopes~\cite{Kane14intersection}.
We complement~\Cref{thm:LO-for-polytopes} with a matching lower bound, which establishes the existence of an $m$-facet polytope with an $\Omega(\sqrt{\ln m}/\sqrt{n})$-fraction of hypercube points lying within its width-$2$ boundary. (In fact, our lower bound is slightly stronger: it establishes the existence of a polytope with an $\Omega(\sqrt{\ln m}/{\sqrt{n}})$-fraction of hypercube points lying on its \emph{surface}, corresponding to its width-0 boundary.)

\Cref{thm:LO-for-polytopes} does not suffice for the purpose of passing from closeness with respect to Bentkus's orthant mollifier $\wt{\calO}_b$ to closeness in multidimensional CDF distance (i.e.~Step 2 in~\Cref{sec:new-ingredients}):  while the assumption on the magnitudes of $A$'s entries is essential to~\Cref{thm:LO-for-polytopes} (just as the analogous assumption on $w$'s coordinates is essential to the Littlewood--Offord theorem), the weight matrix of a general $m$-facet polytope need not have this property.   In~\Cref{sec:anticonc} we establish various technical extensions of~\Cref{thm:LO-for-polytopes} that are required to handle this issue.

\begin{remark}
Our generalization of the Littlewood--Offord theorem (\Cref{thm:LO-for-polytopes}) is, to our knowledge, incomparable to other high-dimensional generalizations that have been studied in the literature. In particular, the papers~\cite{Kleitman70,FF88,TV12} (see also the references therein) study the probability that $A\bu$ falls within a ball of fixed radius in $\R^m$, where $A \in \R^{m\times n}$ is a matrix whose columns have $2$-norm at least $1$ (i.e.~$A\bu$ is the random $\pm 1$ sum of $n$ many $m$-dimensional vectors of length at least~$1$).
\end{remark}

\subsection{Relation to \cite{ST17}}  We close this section with a discussion of the connection between our techniques and those of the recent work~\cite{ST17}.  Recall that the main result of \cite{ST17} is a PRG for $\delta$-fooling intersections of $m$ weight-$W$ halfspaces using seed length $\poly(\log m, W,1/\delta)\cdot \polylog\, n$ (whereas our main result, which is strictly stronger, is a PRG for $\delta$-fooling intersections of $m$ general halfspaces using seed length $\poly(\log m, 1/\delta)\cdot \log n$, with no dependence on the weights of the halfspaces).

A key structural observation driving \cite{ST17} is that
%\gray{every low-weight halfspace must be either sparse (having only a small number of non-zero coefficients) or regular.  Since each sparse halfspace depends only on few variables (i.e.~is a junta), and an intersection of juntas is a small-width CNF, it follows that}
every intersection of $m$ low-weight halfspaces can be expressed as $H \wedge G$, where $H$ is an intersection of $m$ regular halfspaces and $G$ is a small-width CNF.    (The width of $G$ grows polynomially with the weights of the halfspaces, and this polynomial growth is responsible for the polynomial dependence on $W$ in the seed length of the~\cite{ST17} PRG.)  From this starting point, it suffices for~\cite{ST17} to bound the multidimensional CDF distance between the $(\R^{m} \times \{\pm 1\})$-valued random variables $(A\bu, G(\bu))$ and $(A\bz, G(\bz))$, where $A \in \R^{m\times n}$ is the weight matrix of $H$, $\bu$ is uniform random, and $\bz$ is the output of the~\cite{ST17} PRG (which is a slight variant of~\cite{HKM12}'s pseudorandom generator).  Since $H$ is an intersection of regular halfspaces, the fact that $A\bu$ and $A\bz$ are close in multidimensional CDF distance is precisely the main result of~\cite{HKM12}; the crux of the work in~\cite{ST17} therefore lies in dealing with the additional distinguished $(m+1)^{\text{st}}$ coordinate corresponding to the CNF $G$.   Very roughly speaking,~\cite{ST17} employs a careful coupling $(\hat{\bu},\hat{\bz})$ (whose existence is a consequence of the fact that bounded independence fools CNFs~\cite{Baz09,Raz09}) to ensure that  $G(\hat{\bu})$ and $G(\hat{\bz})$ almost always agree, and hence these $(m+1)^{\text{st}}$ coordinates ``have a negligible effect" throughout~\cite{HKM12}'s Lindeberg-based proof of the regular case establishing closeness between $A\bu$ and $A\bz$.

Because of the aforementioned structural fact (that an $m$-tuple of low-weight halfspaces is equivalent to ``an $m$-tuple of regular halfspaces plus a CNF"), the low-weight case analyzed in \cite{ST17} did not require as significant a departure from~\cite{HKM12}'s approach, and from  the Lindeberg method as a whole, as the general case which is the subject of this paper.  In particular, the new ideas discussed in~\Cref{sec:new-ingredients} that are central to our proof were not present in \cite{ST17}'s analysis for the low-weight case.  To elaborate on this,
\begin{itemize}[leftmargin=0.5cm]
\item[$\circ$]~\cite{ST17} did not have to exploit the product structure of Bentkus's orthant mollifier $\wt{\calO}_b$ in order to fool it.  Like~\cite{HKM12}, the arguments of~\cite{ST17} establish closeness in expectation between $\Upsilon(A \bu, G(\bu))$ and $\Upsilon(A \violet{\bz}, G(\violet{\bz}))$ for \emph{all} smooth test functions $\Upsilon$, and the only properties of Bentkus's mollifier that are used are the bounds on its derivatives given in~\cite{Bentkus:90} (which are used in a black box way).  The simpler setting of~\cite{ST17} also did not necessitate comparing the Taylor expansions of $\Upsilon$ around distinct points, as discussed in~\Cref{sec:new-ingredients}.

\item[$\circ$] \cite{ST17} did not have to reason about Boolean anticoncentration, which as discussed above requires significant novel conceptual and technical work, including our new Littlewood--Offord theorem for polytopes.  Like~\cite{HKM12},\cite{ST17} were able to apply Nazarov's Gaussian anticoncentration bounds  as a black box to pass from fooling Bentkus's mollifier to closeness in multidimensional CDF distance.
 \end{itemize}

\section{Preliminaries}  \label{sec:prelims}

For convenience, in the rest of the paper we view halfspaces as having the domain $\bn$ rather than $\zo^n$. 
%\lnote{Do we have to move this sentence about the $\zo \to \bits$ switch up to the start of \Cref{sec:overview}? 
%\red{Rocco:}  We could, but I don't think its absence at the start of~\Cref{sec:overview} is notable there, and if any poor reader gets lost as to whether we're over $\bn$ or $\zo^n$, probably the start of the preliminaries rather than the start of~\Cref{sec:overview} is where (s)he would go to look; so I think it's okay here} 
We remind the reader that a halfspace $F: \bn \to \zo$ is a function of the form $F(x) = \Ind[w \cdot x \leq \theta]$ for some $w \in \R^n$, $\theta \in \R$.
%\lnote{We also use `$w$' to denote the width of a CNF.  I guess there's zero risk of confusion, but do we want to change our weight vectors from $w$ to $a$?  I can do this at the very end. \red{Rocco:}  Ah, good point.  I think we don't absolutely have to do this, since we only use ``$w$'' in kind of a fleeting way as when we're discussing a generic LTF --- it isn't part of our construction when we use it for LTFs, whereas the ``$w$'' of CNF width is part of our construction.  But switching would be fine too of course modulo the hassle and potential for mistake :P }

For an $n$-dimensional vector $y$ and subset $B \subseteq [n]$, we write \violet{$y_B$}\ignore{$y_S$} to denote the $|B|$-dimensional vector obtained by restricting $y$ to the coordinates in $B$.  For an $m \times n$ matrix $A$ and subset $B \sse [n]$, we write $A^B$ to denote the $m\times |B|$ matrix obtained by restricting $A$ to the columns in $B$.  For indices $i \in [m]$ and $j \in [n]$, we write $A_i$ to denote the $n$-dimensional  vector corresponding to the $i$-th row of $A$, and $A^j$ to denote the $m$-dimensional  vector corresponding to the $j$-column of~$A$.

\subsection{Regularity, orthants, and Taylor's theorem} \label{sec:ROT}

%\noindent {\bf Regularity.}

\begin{definition}[$(k,\tau)$-regular vectors and matrices]
\label{def:regularity}
We say that a vector $w \in \R^n$ is \emph{$\tau$-regular} if
%\blue{$\| w \|_2 = 1$}
%\rnote{This is a little non-standard, no? I think of regularity as something that should be scale-invariant. But I guess the HKM definition $\sum_{j=1}^n w_j^4 \leq \tau^2 \sum_{j=1}^n w_j^2$ isn't scale-invariant either.  And I can see how terminologically some later statements like Lemma 5.1 are cleaner this way.
%
%Maybe we could use the term ``normalized'' to indicate that the two-norm is 1 or something?}
%and
$|w_j| \leq \tau \|w\|_2$ for all $j \in [n]$.
%\onote{changed this back, but haven't gotten around to updating the rest of this section, because I got caught up the question which Li-Yang's March 15 email fixed\dots}
%$\sum_{j=1}^n w_j^4 \le \tau^2 \cdot (\sum_{j=1}^n w_j^2)^2$.\ignore{
%\onote{I guess we might stick to using either $\|w\|_2^2$ or $\sum_j w_j^2$ in this definition but not both? {\bf Rocco:} How about minimizing notation and just writing it the way it is?}}
More generally, we say that $w$ is \emph{$(k,\tau)$-regular} if there is a partition $[n] = \textsc{Head} \sqcup \textsc{Tail}$ where $|\textsc{Head}|\le k$ and the subvector $w_{\textsc{Tail}}$ is $\tau$-regular.     We say that $w$ is \emph{$(k,\tau)$-standardized} if $w$ is $(k,\tau)$-regular and $\sum_{j \in \tail} w_j^2 = 1$.
We say that a matrix $A \in \R^{m\times n}$ is $\tau$-regular (respectively: $(k,\tau)$-regular,  $(k,\tau)$-standardized) if all its rows are $\tau$-regular (respectively: $(k,\tau)$-regular, $(k,\tau)$-standardized). We also use this terminology to refer to polytopes $Ax \leq b$.
\end{definition}

\paragraph{Translated orthants and their boundaries.} For $b \in \R^m$, we write $\calO_b \subset \R^m$ to denote the \emph{translated orthant}
\[ \calO_b = \{ v \in \R^m \colon v_i \le b_i \text{ for all $i\in [m]$}\}. \]
We will overload notation and also write ``$\calO_b$" to denote the indicator $\R^m \to \zo$ of the orthant~$\calO_b$ (i.e.,~$\calO_b(v) = \Ind[v \le b]$).
We write $\Game \calO_{b} \subset \calO_b$ to denote $\calO_b$'s \emph{surface},
\[ \Game \calO_b = \{ v \in \calO_b \colon v_i = b_i \text{ for some $i\in [m]$}\}. \]
For $\Lambda > 0$,
%%%\onote{It kind of stresses me that we use a capital letter for a small quantity\dots well, I guess I can deal\dots \blue{Li-Yang:} I'll fix this when things are more stable. \red{Rocco:}  Stating the obvious, this was chosen b/c of the affinity with $\lambda$ which it upper bounds, right?  It'd be nice to retain that connection so if we do switch it let's go with ``$\lambda'$'' or ``$\lambda_1$ or something similar.} 
we write $\Game_{-\Lambda}\calO_b$ and $\Game_{+\Lambda} \calO_b$ to denote the \emph{inner and outer $\Lambda$-boundaries of $\calO_b$},
\begin{equation} \label{eq:boundaries}
\Game_{-\Lambda}\calO_b = \calO_b \setminus \calO_{b - (\Lambda, \dots, \Lambda)},
\quad \quad \quad
\Game_{+\Lambda}\calO_b = \calO_{b + (\Lambda, \dots, \Lambda)} \setminus \calO_{b},
\end{equation}
\ignore{
\begin{align*}
\Game_{-\Lambda}\calO_b &= \calO_b \setminus \calO_{b - (\Lambda, \dots, \Lambda)}, \\
\Game_{+\Lambda}\calO_b &= \calO_{b + (\Lambda, \dots, \Lambda)} \setminus \calO_{b},
\end{align*}
}
and $\Game_{\pm\Lambda}\calO_b$ to denote the disjoint union $\Game_{\pm\Lambda}\calO_b = \Game_{+\Lambda}\calO_b \sqcup \Game_{-\Lambda}\calO_b$.

\paragraph{Derivatives and multidimensional Taylor expansion.}
We write $\psi^{(\tay)}$ to denote the $\tay$-th derivative of a $\calC^d$ function $\psi : \R \to \R.$
For an $m$-dimensional  multi-index $\alpha = (\alpha_1,\ldots,\alpha_m) \in \N^m$, we write $|\alpha|$ to denote $\alpha_1 + \cdots + \alpha_m$, and $\alpha!$ to denote $\alpha_1! \alpha_2! \cdots \alpha_m!$.  Given a vector $\Delta \in \R^m$, the expression $\Delta^\alpha$ denotes $\prod_{i=1}^m \Delta_i^{\alpha_i}$.  Given a function $\Upsilon : \R^m \to \R$, the expression $\partial_\alpha \Upsilon$ denotes the mixed partial derivative taken $\alpha_i$ times in the $i$-th coordinate.
%\gray{For $c \in \N$, we write
%\[ \| \Upsilon^{(c)} \|_1 \quad \text{to denote} \quad \sup_{v \in \R^m}\Bigg\{ \sum_{|\alpha|=c} |\partial_\alpha \Upsilon(v)|\Bigg\}. \]}

The following is a straightforward consequence of the multidimensional Taylor theorem, upper-bounding the error term by the $L_1$-norm of the derivatives times the $L_\infty$-norm of the offset-powers\ignore{\rnote{This phrase reads a little weirdly to me, but I guess it's okay.}}:
\begin{fact}[Multidimensional Taylor approximation]% \cite{Konigsberger02}\rnote{Let's cite something (doesn't have to be to this specific $L_1$/$L_\infty$ form, and doesn't have to be this book, but something other than Wikipedia :)}\onote{Could we not cite a whole book?  Actually, I propose not citing anything.  I mean, it's Taylor's theorem for goodness sake; it's not like we provide citations for the Cauchy--Schwarz inequality or the union bound :)}]
\label{fact:taylor}
    Let $\tay \in \N$ and let $\Upsilon : \R^m \to \R$ be a $\calC^\tay$ function. Then for all $v ,\Delta \in \R^m$,
%    \onote{Someone want to double-check this?  Old version seemed to be missing the $\alpha!$, for instance.  {\bf {\red{Rocco}:}}  Looks good to me now\dots}
    \[
        \Upsilon(v + \Delta) = \sum_{0 \le |\alpha| \le \tay-1} \frac{\partial_\alpha \Upsilon(v)}{\alpha!} \Delta^{\alpha} + \err(v,\Delta),
    \]
%    \[
%        \Upsilon(v + \Delta) = \sum_{0 \le |\alpha| \le \tay-1} \partial_\alpha \Upsilon(v) \Delta^{\alpha} + \err(v,\Delta),
%    \]
where
    \[
        \abs*{\err(v,\Delta)} \le
%        \sup_{v^* \in [v,v+\Delta]}\Bigg\{ \sum_{|\alpha|=\tay} \frac{|\partial_\alpha \Upsilon(v^*)|}{\alpha!}\Bigg\} \cdot \| \Delta\|_\infty^\tay \leq
        \sup_{v^* \in \R^m}\Bigg\{ \sum_{|\alpha|=\tay} |\partial_\alpha \Upsilon(v^*)|\Bigg\} \cdot \| \Delta\|_\infty^\tay.
    \]
%\gray{$|\err(v,\Delta)| \le \|\Upsilon^{(3)}\|_1 \cdot \| \Delta\|^3_\infty$.}\onote{I think the notation $\|\Upsilon^{(3)}\|_1$ is pretty misleading, no?  I mean, I don't actually know what $\Upsilon^{(3)}$ even notates, but it looks strongly to me like a\dots function?  (If $m = 1$ then the notation makes sense and it would be a function. I think for Bentkus it's a ``Frechet derivative'' but I don't know what that means.)  But $\|f\|_1$ generally denotes something like $\int_x |f(x)|$, whereas we're actually sup-ing over the domain of the function.  I'm not sure there \emph{is} amazing notation for the quantity we care about, but I think we can't really use $\|\Upsilon^{(3)}\|_1$. {\bf Rocco:} I agree that this notation is terrible; if we can't think of anything better, though, we should explicitly say that we are following \cite{HKM12}.}
\end{fact}

\subsection{Pseudorandomness preliminaries}

Throughout this work we use \textbf{boldface} for random variables and random vectors.  If $\calD$ is a probability distribution, we write $\bx \sim \calD$ to denote that $\bx$ is drawn from that distribution.  For example, $\calN(0,1)$ will denote the standard normal distribution, so $\bg \sim \calN(0,1)$ means~$\bg$ is a standard Gaussian random variable.  In case~$S$ is a finite set, the notation $\bx \sim S$ will mean that $\bx$ is chosen uniformly at random from~$S$.  The most common case for this will be  $\bu \sim \{-1,1\}^n$, meaning that $\bu$ is chosen uniformly from $\{-1,1\}^n$. We will  reserve~$\bu$ for this specific random vector.

%When $S$ is a set the notations $\Prx_{\bx \sim S}[\cdot], \Ex_{\bx \sim S}[\cdot]$ indicate that the relevant probability or expectation is over a uniform draw of $\bX$ from set $S$.\lnote{I'm not sure we want to say this since we do write ``Let $\by\sim \bn$ be pseudorandom". \blue{Rocco: I think this is okay, we are defining here the notations $\Prx_{\bx \sim S}[\cdot], \Ex_{\bx \sim S}[\cdot]$ where the ``$\bx \sim S$'' is in the subscript of a $\Pr$ or an $\Ex$}}\onote{In fact, do we even use this notation ever?  I couldn't find a spot.}   Throughout the paper we use bold fonts such as $\bx,\bu,\bh,$ etc. to indicate random variables.
%Throughout the paper we write ``$\bu$'' to denote a uniform random string drawn from $\bn$.

We recall the definition of a pseudorandom generator:

\begin{definition}[Pseudorandom generator] \label{def:PRG}
    A function $\mathscr{G} : \bits^r \to \bits^n$ is said to \emph{$\delta$-fool a function $F: \bits^n \to \R$ with seed length $r$} if
    \[
        \Big|
        \Ex_{\bs \sim \bits^r}\big[F(\mathscr{G}(\bs))\big] -
        \Ex_{\bu \sim \bits^n}\big[F(\bu)\big]
        \Big| \leq \delta.
    \]
    Such a function $\mathscr{G}$ is said to be a \emph{explicit pseudorandom generator (PRG) that $\delta$-fools a class $\calF$ of $n$-variable functions} if $\mathscr{G}$ is computable by a deterministic uniform $\poly(n)$-time algorithm and $\mathscr{G}$ $\delta$-fools every function $F \in {\cal F}$.  We will also use the notation $\bz \sim \mathscr{G}$ to mean that $\bz = \mathscr{G}(\bs)$ for $\bs \sim \bits^r$.
\end{definition}

\pparagraph{Bounded independence and hash families.}
A sequence of random variables $\bx_1, \dots, \bx_n$ is said to be \emph{$r$-wise independent} if any collection of~$r$ of them is independent.  In case the $\bx_i$'s are uniformly distributed on their range, we say the sequence is \emph{$r$-wise uniform}.  We will also use this terminology for distributions $\calD$ on~$\{-1,1\}^n$.  An obvious but useful fact about $r$-wise uniform PRGs $\mathscr{G}$ is that they $0$-fool the class of degree-$r$ polynomials $\{-1,1\}^n \to \R$.

A distribution $\calH$ on functions $[n] \to [\bucks]$ is said to be an \emph{$r$-wise uniform hash family} if, for $\bh \sim \calH$, the sequence $(\bh(1), \dots, \bh(n))$ is $r$-wise uniform. Such a distribution also has the property that for any $\ell \in [\bucks]$, the sequence $(\Ind_{\bh(1) = \ell}, \dots, \Ind_{\bh(n) = \ell})$ is $r$-wise independent on $\{0,1\}^n$, with each individual random variable being Bernoulli$(1/\bucks)$.  Well-known constructions (see e.g. Section~3.5.5 of \cite{Vadhan12}) give that for every $n,\bucks$ and $r$, there is an $r$-wise uniform hash family $\calH$ of functions $[n] \to [\bucks]$ such that choosing a random function from $\calH$ takes $O(r  \log(nL))$ random bits (and evaluating a function from $\calH$ takes time $\poly(r, \log n, \log \bucks)$), and consequently there are known efficient constructions of $r$-wise uniform distributions over $\zo^n$ with seed length $O(r  \log n).$

\pparagraph{Fooling CNFs.}  Gopalan, Meka, and Reingold \cite{GMR13} have given an efficient explicit PRG that fools the class of small-width CNFs:

\begin{theorem} [PRG for small-width CNFs]\label{thm:GMR}  There is an explicit PRG $\mathscr{G}_{\mathrm{GMR}} = \mathscr{G}_{\mathrm{GMR}}(w,\delta_\CNF)$ that $\delta_\CNF$-fools the class of all width-$w$ CNF formulas over $\bn$ and has seed length
\[
O(w^2 \log^2(w \log(1/\delta_\CNF)) + w \log(w)\log(1/\delta_\CNF) + \log \log n).
\]
\end{theorem}

%\pparagraph{Bounded-\red{uniformity} distributions and hash families.}
%A distribution $\calD$ over $\bits^n$ is \emph{$r$-wise \red{uniform}} if for every $1 \leq j_1 < \cdots < j_r \leq n$ and every $(b_1,\dots,b_r) \in \bits^r$, we have
%\[
%\Prx_{\bX \sim \calD}\big[\bX_{j_1} = b_1 \text{~and~} \cdots \text{~and~}\bX_{j_r} = b_r\big] = 2^{-r}.
%\]
%A family $\calH$ of functions from $[n]$ to $[\bucks]$ is said to be an \emph{$r$-wise \red{uniform} hash family} if for every $1 \leq j_1 < \cdots < j_r \leq n$ and $(\ell_1,\dots,\ell_r) \in [\bucks]^r$, we have
%\[
%\Prx_{\bh \sim \calH}\big[\bh(j_1) = \ell_1 \text{~and~} \cdots \text{~and~} \bh(j_r)=\ell_r\big]=\bucks^{-r}.
%\]

%!TEX root = main.tex

\section{Our PRG}  \label{sec:theprg}

\noindent {\bf The Meka--Zuckerman generator.}  As stated earlier the PRG which we will analyze is a slight variant of a PRG first proposed by Meka and Zuckerman for fooling a single halfspace~\cite{MZ13}.  We begin by recalling the Meka--Zuckerman PRG.

\begin{definition}[Meka--Zuckerman generator]
\label{def:MZ}
The \emph{Meka--Zuckerman generator with parameters
\ignore{\rnote{Nit-picky: why do we include ``$n$'' as a parameter to $\mathscr{G}_{\mathrm{GMR}}$ but not to $\mathscr{G}_{\mathrm{MZ}}$? {\blue{Li-Yang:} Okay I removed $n$ as a parameter to GMR.}}}$\bucks, r_\hash,r_\bucket \in [n]$}, denoted $\mathscr{G}_{\mathrm{MZ}}$, is defined as follows.  Let $\bh : [n] \to [\bucks]$ be an $r_\hash$-wise uniform hash function.  Let $\by^{1},\ldots,\by^{\bucks} \sim \bn$ be independent random variables, each $r_\bucket$-wise uniform.  A draw from $\mathscr{G}_{\mathrm{MZ}} = \mathscr{G}_{\mathrm{MZ}}(\bucks,r_\hash,r_\bucket)$ is $\bz \sim \bn$ where
\begin{equation*}
  \bz_{\bh^{-1}(\ell)} = \by^{\ell}_{\bh^{-1}(\ell)} \quad \text{for all $\ell \in [\bucks]$.}
  \end{equation*}
\end{definition}

In words, an \ignore{$r_\bucket$}$r_\hash$-wise uniform hash $\bh$ is used to partition the variables $x_1,\dots,x_n$ into $\bucks$ ``buckets,'' and then independently across buckets, the variables in each bucket are assigned according to an $r_\bucket$-wise uniform distribution.

We note in passing that the generators of~\cite{HKM12,ST17} also have this structure (though the choice of parameters $L, r_\bucket$, and $r_\hash$ are different than those in~\cite{MZ13}).
\medskip

%\begin{fact}[Seed length of $\mathscr{G}_{\mathrm{MZ}}$]  \label{fact:MZseedlength}
%The seed length of the Meka--Zuckerman generator with parameters $\bucks, r_\hash$, and $r_\bucket$ is $O(r_\hash  \log(nL) + \bucks\cdot r_\bucket\log n).$
%\end{fact}

\noindent {\bf Our generator.} Now we are ready to describe our generator and bound its seed length.  Roughly speaking, our generator extends the Meka--Zuckerman generator by (i) additionally {\sc Xor}-ing each bucket with an independent pseudorandom variable that fools CNF formulas; and (ii) globally {\sc Xor}-ing the entire resulting $n$-bit string with an independent draw from a $2k$-wise uniform distribution.

%%\rnote{Maybe an expository sentence before the formal definition? (I have had the experience of checking out of a PRG paper when confronted with a long technical definition of a generator\dots) Perhaps something like ``Roughly speaking our generator extends the [MZ] generator by (i) XORing each bucket with an independent pseudorandom variable that fools small-width CNF formulas, and (ii) globally XORing the entire resulting $n$-bit string with an independent draw from a $2k$-wise uniform distribution over $\zo^n$.''}

\begin{definition} [Our generator] \label{def:ourgenerator}  Our generator, denoted $\mathscr{G}$, is parameterized by values $\bucks, r_\hash, r_\bucket,$ $k,w \in [n]$, $\delta_\CNF \in (0,1)$ and is defined as follows.  Let:
\begin{itemize}
\itemsep -.5pt
\item[$\circ$] $\bh, \by^1,\dots,\by^\bucks$ be defined as in the Meka--Zuckerman generator with parameters $L$, $r_\hash$, and $r_\bucket$.
\item[$\circ$] $\tilde{\by}^1,\dots,\tilde{\by}^\bucks \sim \bn$ be independent draws from $\mathscr{G}_{\mathrm{GMR}}(w,\delta_\CNF)$.
\item[$\circ$] $\by^\star \sim \bn$ be $2k$-wise uniform.
\end{itemize}
Define the random variable $\breve{\by} \sim \bn$ by
\[  \breve{\by}_{\bh^{-1}(\ell)} = (\by^{\ell} \oplus \tilde{\by}^\ell)_{\bh^{-1}(\ell)}   \qquad \text{for all $\ell \in [\bucks]$},\]
where $\oplus$ denotes bitwise {\sc Xor}.
%In words, $\breve{\by}$ differs from the Meka--Zuckerman generator (cf.~(\ref{eq:MZ})) in the following way: within each bucket, the $r_\bucket$-wise uniform assignment to the variables is additionally {\sc Xor}-ed with a draw from the \cite{GMR13} generator for small-width CNF formulas (a fresh independent draw for each bucket).
A draw from our generator $\mathscr{G} = \mathscr{G}(\bucks,r_\hash,r_\bucket,k,w,\delta_\CNF)$ is $\bz \sim \bn$ where $\bz = \breve{\by} \oplus \by^\star$.\end{definition}

%\red{TO BE UPDATED}
%In words, our generator $\mathscr{G}$ differs from $\mathscr{G}_{\mathrm{MZ}}$ in the following way: within each bucket, the variables in the bucket are additionally {\sc Xor}-ed with a draw from the \cite{GMR13} generator (a fresh independent draw for each bucket).
%
%By Fact~\ref{fact:MZseedlength} and Theorem~\ref{thm:GMR}, using $\bucks \leq n$, we have the following seed length bound for our PRG:

Recalling the standard constructions of $r$-wise uniform hash functions and random variables described at the end of~\Cref{sec:prelims}, we have the following:

\begin{fact}[Seed length]\label{fact:ourseedlength}
The seed length of our PRG $\mathscr{G}$ with parameters $\bucks, r_\hash,r_\bucket,k,w,\delta_\CNF$~is
\begin{align*}
\precsim &\ r_\hash \cdot  \log(nL) + \bucks\cdot r_\bucket \cdot \log n  \tag*{(Seed length for $\mathscr{G}_\mathrm{MZ}$)}\\
& +
L \cdot(w^2 \log^2(w \log(1/\delta_\CNF)) + w \log(w)\log(1/\delta_\CNF) + \log \log n) \tag*{($L$ copies of $\mathscr{G}_{\mathrm{GMR}}$)}\\
& + k\log n.\tag*{($2k$-wise uniform string)}
\end{align*}
\end{fact}

\subsection{Setting of parameters}
\label{sec:params}

We close this section with the parameter settings for fooling intersections of $m$ halfspaces over $\bn$.  Fix $\eps \in (0,1)$ to be an arbitrarily small \violet{absolute} constant; the parameters we now specify will be for fooling to accuracy $O_\eps(\delta) = O(\delta)$.    We first define a few auxiliary parameters:
\begin{align*}
\lambda &= \frac{\delta}{\sqrt{\log(m/\delta) \log m}}  \tag*{\small{(Dictated by \Cref{eq:our-lambda})}} \\
\tau &= \frac{\delta^{1+\eps}}{(\log m)^{2.5+\eps}}  \tag*{\small{(Dictated by \Cref{eq:params})}} \\
%\frac{\delta^{1+\eps}}{(\log m)^{1+\eps} \cdot  (\log(m/\delta))^{3/2}}  \\
d &= \text{constant depending only on $\eps$.}  \tag*{\small{(Dictated by \Cref{eq:params})}}
\end{align*}
The precise value of $d =d(\eps)$ will be specified in the proof of \Cref{thm:fool-bentkus}.    We will instantiate our generator $\mathscr{G} = \mathscr{G}(\bucks,r_\hash,r_\bucket,k,w,\delta_\CNF)$ with parameters:
{\setlength{\jot}{-3pt}
 \begin{align*}
 L &= \frac{(\log m)^5}{\delta^{2+\eps}}  \tag*{\small{(Constrained by \Cref{eq:params},}} \\
& \tag*{chosen to optimize seed length)}
\end{align*}}
\vspace*{-25pt}
\begin{align*}
r_\hash &= C_1 \log(Lm/\delta)  \tag*{\small{(Dictated by \Cref{prop:hash})}} \\
r_\bucket &= \log(m/\delta) \tag*{\small{(Dictated by \Cref{lem:hash})}}  \\
 k &= \frac{C_2 \log(m/\delta)\log\log(m/\delta)}{\tau^2} \tag*{\small{(Dictated by \Cref{thm:regularize})}}   \\
w &= \frac{2k}{L} \tag*{\small{(Dictated by \Cref{prop:hash})}} \\
 \delta_\CNF &= \frac{\delta}{L}\cdot  \parens*{\frac{\lambda}{m\sqrt{n}}}^{d-1}, \tag*{\small{(Dictated by \Cref{eq:params})}}
 \end{align*}
where $C_1$ and $C_2$ are absolute constants specified in the proofs of \Cref{prop:hash} and \Cref{thm:regularize} respectively.

\paragraph{Our seed length:} By \Cref{fact:ourseedlength}, our overall seed length is
\begin{equation}
\tilde{O}\parens*{\frac{(\log m)^{6+\eps}}{\delta^{2+\eps}}} \cdot \log n + \tilde{O}\parens*{\frac{(\log m)^{7+\eps}}{\delta^{2+\eps}}} = \poly(\log m,1/\delta) \cdot \log n
\label{eq:sl}
\end{equation}
for any absolute constant $\eps \in (0,1)$.

\begin{remark}
\label{rem:MZ-itself}
As alluded to in the introduction, our techniques can also be used to show that the Meka--Zuckerman generator itself fools the class of intersections of $m$ halfspaces over $\bn$.   However, this would require setting the parameters $L, r_\hash$, and $r_\bucket$ to be somewhat larger than the values used above, and would result in a slightly worse seed length of $\poly(\log m,1/\delta,\log n)$ than our $\poly(\log m,1/\delta)\cdot \log n$.  Briefly,
such an analysis would use the fact that bounded-uniformity distributions fool CNF formulas~\cite{Baz09,Raz09}; our analysis instead uses the (more efficient) \cite{GMR13} generator for this purpose.
\end{remark}

\section{Reduction to standardized polytopes} \label{sec:reduction}

\subsection{A reduction from fooling polytopes to fooling standardized polytopes}

In this section we reduce from the problem of fooling general $m$-facet polytopes to the problem of fooling $m$-facet $(k,\tau)$-standardized polytopes (\Cref{def:regularity}).  The main technical result we prove in this section is the following:

\begin{lemma}
[Approximating arbitrary polytopes by $(k,\tau)$\standardized polytopes under bounded-uniformity distributions]
\label{thm:regularize}
There is a universal constant $C_2 > 0$ such that the following holds.  Fix $m \geq 1$ and $0 < \delta, \tau < 1/2$ such that the right-hand side of \Cref{eq:k-value} below is at most~$n/2$.\ignore{\onote{Do we really need this thing about $n/2$? {\bf \red{Rocco}:}  I guess the definition of $(k,\tau)$-standardized still formally makes sense if $k > n$, but it would be nice to spare the reader from having to think about that wacky regime, and I think it comes in handy at the very end of the proof of \Cref{thm:regularize}. Is there a downside to keeping this thing about $n/2$?}}  Let
    \begin{equation} \label{eq:k-value}
        k = \frac{C_2 \log(m/\delta)\log\log(m/\delta)}{\tau^2}.%, \qquad r = O(\log(m/\delta))
    \end{equation}
For every $m$-facet polytope $Ax \leq b$ in $\R^n$, there is an $m$-facet $(k,\tau)$\standardized polytope $A'x \leq b'$ in $\R^n$ such that if $\by \sim \{-1,1\}^n$ is $2k$-wise uniform,
then
    \begin{equation}    \label{eqn:regularize-me}
        \Pr\bracks*{\Ind[A\by \leq b] \neq \Ind[A' \by \leq b']} \leq \delta.
    \end{equation}
\end{lemma}

\begin{remark} \label{rem:hedge}
Had we been content in this theorem with the worse value of $k = O\parens*{\log^2(m/\delta)/\tau^2}$, then the result would essentially be implicit in~\cite[Theorem~5.4]{DGJ+10:bifh} (and \cite[Theorem~7.4]{GOWZ10}), using only $(k+2)$-wise uniformity.  To save essentially a $\log(m/\delta)$ factor, we give a modified proof in \Cref{sec:proof-thm-regularize}.
\end{remark}

We stress that \Cref{thm:regularize} establishes that $\Ind[Ax \leq b]$ is well-approximated by $\Ind[A'x \leq b']$ under \emph{both} the uniform distribution and the pseudorandom distribution constructed by our generator, since both of these distributions are $2k$-wise uniform.  (Note that a draw $\bz = \breve{\by} \oplus \by^\star$ from our generator is indeed $2k$-wise uniform, since $\by^\star$ is;
indeed,~\Cref{thm:regularize} is the motivation for why our construction includes a bitwise-{\sc Xor} with $\by^\star$.)
%\onote{We should carefully verify that our PRG is indeed $2k$-wise uniform.  This requires staring deeply at our parameters, so I don't know if we can just say this sentence casually. {\bf \red{Rocco}:}  Agreed; I think this and other aspects of this section's exposition will change when we implement LYT's excellent suggestion of having the PRG do a global XOR with an $O(k)$-wise independent distribution.}
This is crucial: in general, given a function $F$ and an approximator $F'$ that is close to $F$ only under the uniform distribution (i.e.~$\Pr[F(\bu) \neq F'(\bu)]$ is small), fooling~$F'$ does not suffice to fool $F$ itself.

Given \Cref{thm:regularize}, in order to prove \Cref{thm:main} it is sufficient to prove the following:

\begin{theorem}[Fooling $(k,\tau)$-standardized polytopes]
\label{thm:fool-k-tau-regular}
Let $\mathscr{G}$ be our generator with parameters as set in \Cref{sec:params}.   For all $m$-facet $(k,\tau)$-standardized polytopes $A'x \leq b'$,
\[ \bigg| \Prx_{\bu \sim \bn}\big[A'\bu \in \calO_{b'} \big] - \Prx_{\bz \sim \mathscr{G}}\big[ A'\bz \in \calO_{b'}\big] \bigg| =  O(\delta). \]
\end{theorem}

\begin{proof}[Proof of \Cref{thm:main} assuming \Cref{thm:fool-k-tau-regular} and \Cref{thm:regularize}]
Let $Ax \leq b$ be any $m$-facet polytope in $\R^n$. Given $\delta>0$, we recall that $\tau = \red{\delta^{1+\eps}/(\log m)^{2.5+\eps}}.$  If the quantity~\eqref{eq:k-value} is greater than $n/2$ then the claimed seed length from~\Cref{fact:ourseedlength} \ignore{$\red{FILLINSL = \log n \cdot \poly(\log m, 1/\delta)}$} is greater than $n$ and the conclusion of \Cref{thm:main} trivially holds, so we suppose that~\eqref{eq:k-value} is less than $n/2.$  Let $A'x \leq b'$ be the $m$-facet $(k,\tau)$\standardized polytope given by \Cref{thm:regularize}.  We have
\begin{align*}
\Prx_{\bu \sim \bn}[A\bu \in \calO_b]
&= \Prx_{\bu \sim \bn}[A'\bu \in \calO_{{b'}}] \pm \delta \tag{\Cref{thm:regularize} applied to $\bu$}\\
&= \Prx_{\bz \sim \mathcal{G}}[A'\bz \in \calO_{{b'}}] \pm \delta \pm \delta \tag{\Cref{thm:fool-k-tau-regular}}\\
&= \Prx_{\bz \sim \mathcal{G}}[A\bz \in \calO_b] \pm \delta \pm \delta \pm \delta \tag{\Cref{thm:regularize} applied to $\bz$}
\end{align*}
and \Cref{thm:main} is proved.
\end{proof}

The rest of the paper following this section is devoted to proving \Cref{thm:fool-k-tau-regular}.  In the remainder of this section we prove \Cref{thm:regularize}.

\subsection{Proof of \Cref{thm:regularize}} \label{sec:proof-thm-regularize}

The proof uses the ``critical index'' theory for Boolean halfspaces, introduced in~\cite{Servedio:07cc} and used in several subsequent works on halfspaces.% various subsequent works such as~\cite{DGJ+:10,OS11:chow,GOWZ10}.
\begin{definition}[Critical index]\label{def:crit-ind}
    Let $w \in \R^n$ and assume for notational simplicity that $|w_1| \geq |w_2| \geq \cdots \geq |w_n|.$  The \emph{$\tau$-critical index} of~$w$ is the least~$j$ such that the ``tail'' $(w_j, w_{j+1}, \dots, w_n)$ is $\tau$-regular, or $\infty$ if no such~$j$ exists.
\end{definition}
Given $A$ as in \Cref{thm:regularize}, the rows that are already $(k,\tau)$-regular pose no difficulty as a simple rescaling of any such row (and the corresponding entry of $b$) makes it $(k,\tau)$-standardized.  The remaining rows $A_i$ have $\tau$-critical index exceeding~$k$.  The critical index theory~\cite{Servedio:07cc,OS11:chow} says that such halfspaces $\Ind[A_i x \leq b_i]$ are very close to $k$-juntas, and in fact~\cite{DGJ+10:bifh} shows that this is true even under $(k+2)$-wise uniform distributions (for a slightly larger choice of $k$ as alluded to in \Cref{rem:hedge}).  We tweak the quantitative aspects of these arguments below to work for the choice of $k$ given in~\eqref{eq:k-value}.  It will be convenient to follow the treatment in~\cite{GOWZ10}.

The first lemma below says that if the ``head'' variables are set uniformly, the resulting random variable has good anticoncentration at the scale of the two-norm of the tail:

\begin{lemma}                                       \label{lem:crit1}
    Let $\tau \in (0,1)$, $\eps \in (0,1/2)$, $s > 1$.  Then for a certain  $\ell = O(\log(s) \log(1/\eps)/\tau^2)$ the following holds: If $w \in \R^n$ as in \Cref{def:crit-ind} has $\tau$-critical index at least~$\ell$, then for all $\theta \in \R$,
    \[
        \Prx_{\substack{\bu\sim \{-1,1\}^\ell \\ \textnormal{uniform}}}
        \bracks*{\abs{w_1\bu_1 + \cdots + w_\ell\bu_\ell - \theta} \leq s \cdot \sigma} \leq \eps + O(\log(1/\eps)\exp(-s^2/2)),
    \]
    where $\sigma \coloneqq \sqrt{w_{\ell+1}^2 + \cdots + w_n^2}$.
\end{lemma}
\begin{proof}
    We refer directly to the proof of the almost identical~\cite[Theorem~5.3]{GOWZ10} in the full version of that paper.  In that proof we may take ``$\delta$'' to be $\tau^2$, and ``$\eta$'' to be $1/\sqrt{3}$ since we work with uniform $\pm 1$ bits (see Fact~3.3.5 therein). The only change needed in the proof occurs before ``inequality~(10)''.  That inequality uses the fact that a certain random variable $\bz$ satisfies the tail bound $\Pr[|\bz| \geq s \rho] \leq O(1/s^4)$ when $\rho$ is at most the standard deviation of $\bz$. But in our current setting, the random variable $\bz$ equals $w_1\bu_1 + \cdots + w_\ell\bu_\ell$, i.e.~it is a weighted sum of independent uniform~$\pm 1$ bits, and so we have the improved tail bound $2\exp(-s^2/2)$ using Hoeffding.  Carrying through the remainder of the proof with this change yields the conclusion of \Cref{lem:crit1}.
\end{proof}

\begin{lemma}                                       \label{lem:crit2}
    Let $\tau \in (0,1)$ and let $\eps \in (0,1/2)$. Then for a certain $k = O(\log(1/\eps)\log\log(1/\eps)/\tau^2)$ and $r = O(\log(1/\eps))$, the following holds for every $w \in \R^n$ that is \emph{not} $(k,\tau)$-regular:

    Let $H \subseteq [n]$ be the set of $k$ coordinates $i$ for which $|w_i|$ is largest and let $T = [n] \setminus H$.  Assume $w' \in {\R^n}$ has $w'_H = w_H$ and $\|w'_T\|_2 \leq \|w_T\|_2$.  Then for any $\ignore{\theta'}{\theta} \in \R$,
    \[
        \Prx_{\by}\bracks*{\Ind[w \cdot \by \leq \theta] \neq \Ind[w' \cdot \by \leq \theta]} = O(\eps)
    \]
    provided $\by \sim \{-1,1\}^n$ is $(k + r)$-wise uniform.
\end{lemma}
\begin{proof}
    Suppose $w$ is not $(k,\tau)$-regular.     By reordering coordinates we may assume that $H = [k]$; then the non-$(k,\tau)$-regularity of~$w$ means the $\tau$-critical index of~$w$ exceeds~$k$.  We may therefore apply \Cref{lem:crit1} with $s = O(\sqrt{\log(1/\eps)})$.  Using the fact that $\by_H$ is fully uniform we get
    \begin{equation}
        \Pr\bracks*{\abs{w_H \cdot \by_H - \theta} \leq s \cdot \|w_{\red{T}}\|_2} = O(\eps) \qquad \text{(and note that } w'_H \cdot \by_H = w_H \cdot \by_H\text{).} \label{eq:aaa}
    \end{equation}
    Conditioned on any outcome of $\by_H$, the distribution of $\by_T$ remains $r$-wise uniform.  We claim that it  remains to show the following:
    \begin{equation}    \label{eqn:crit2-crit}
        \Pr[|w'_T \cdot \by_T| \geq s \cdot \|w_{\red{T}}\|_2] = O(\eps).
    \end{equation}
    To see that this suffices, observe that by \eqref{eq:aaa} we have that $\abs{w_H \cdot \by_H - \theta} = \abs{w'_H \cdot \by_H - \theta} > s \cdot \|w_{\red{T}}\|_2$ except with probability $O(\eps)$. Also, by applying~\eqref{eqn:crit2-crit} with $w'$ and with $w' = w$, we get both $|w_T \cdot \by_T|, |w'_T \cdot \by_T| \leq s \cdot \|w_{\red{T}}\|_2$ except with another probability at most~$O(\eps)$.
    %\onote{Is it okay now?} 
    When all of these events occur, $\Ind[w \cdot \by \leq \theta]$ and $\Ind[w' \cdot \by \leq \theta]$ agree.

    Finally, we can establish \Cref{eqn:crit2-crit} by appealing to, e.g., \cite[Theorem~9.23]{ODbook}.  That theorem (with $k = 1$) shows that for $t \geq \sqrt{2e}$, any linear form $f(\bx)$ in uniform $\pm 1$ random variables~$\bx$ has $\Pr[|f(\bx)| \geq t \|f\|_2] \leq \exp(-O(t^2))$.  If we could directly apply this to the linear form $w'_T \cdot \by_T$, we would be done by taking $t = s$ and using $\|w'_{\red{T}}\|_2 \leq \|w_{\red{T}}\|_2$.  We cannot directly apply this theorem because the bits $\by_T$ are not uniformly random. However, inspecting the proof of \cite[Theorem~9.23]{ODbook} shows that it suffices for those bits to be $O(t^2)$-wise uniform, which they are provided that $r = O(\log(1/\eps)) = O(s^2) = O(t^2)$.  The reason that this suffices is because the proof only uses $(2,q,1/\sqrt{q-1})$-hypercontractivity of $f(\bx)$ for $q = O(t^2)$, and (for even integer~$q$) this condition only involves the first $q$ moments of~$f(\bx)$, which don't change if $\bx$ is assumed to be merely $q$-wise uniform rather than truly uniform.\ignore{\onote{I think we're going to appeal to literally exactly this fact again later, like around Proposition~2 of hkm-reproof5.  So we might want to box it up. \red{Rocco:} Is this still true?}}
\end{proof}

We can now prove \Cref{thm:regularize}:

\begin{proof}
    We will use \Cref{lem:crit2} with $\eps = c \delta/m$ for small constant $c > 0$. This leads to the choice of~$k$ in the statement of \Cref{thm:regularize}; also, $r \ll k$ and so $2k \geq r + k.$

    Given $A \in \R^{m \times n}$, as noted earlier the rows that are $(k,\tau)$-regular are not a problem, so we  consider all rows $A_i$ that are \emph{not} $(k,\tau)$-regular.  For these rows we apply \Cref{lem:crit2}, taking $A'_i$ to agree with $A_i$ on the appropriate ``head'' coordinates~$H_i$, and taking $A'_i$ to simply be $0$ on the remaining ``tail'' coordinates.  Note that $A'_i$ is now trivially $(k,\tau)$-regular.  By \Cref{lem:crit2} we have that
    \[
        \Prx_{\by}\bracks*{\Ind[A_i \cdot \by \leq b_i] \neq \Ind[A'_i \cdot \by \leq b_i]} \leq \delta/m.
    \]
    Taking $b'_i = b_i$ for these $i$'s, and union-bounding over the at most~$m$ of them, we are almost at the point of establishing \Cref{eqn:regularize-me} from the theorem statement.  We now have that \emph{all} $A'_i$ are $(k,\tau)$-regular; the only deficiency is that the ``tail'' of each row need not have $2$-norm~$1$ as required.

    Whenever the ``tail'' of $A'_i$  has nonzero $2$-norm, we can simply scale $A'_i$ and $b'$ by the same positive factor so as to make the tail of $A'_i$ have $2$-norm~$1$; this scaling does not change the Boolean function $\Ind[A'_i \cdot x \leq b'_i]$ at all.  The only (very minor) difficulty now remaining is that some of the rows $A'_i$ may have tail with $2$-norm zero.  It is well known, however, that one can always slightly perturb the coefficients and threshold in a halfspace without changing it as a Boolean function.\footnote{Given a halfspace $\Ind[w \cdot x \leq \theta]$, there is a smallest value $\theta' > \theta$ achievable as $w \cdot x$ for $x \in \{-1,1\}^n$; first perturb $\theta$ upward to $(\theta + \theta')/2$.  Now no input $x$ achieves $w\cdot x = \theta$ exactly, so we can perturb the coefficients of~$w$ by sufficiently small amounts.}  We can perturb in such a way that the tail coefficients all become equal to some sufficiently small $\eta > 0$.  After this perturbation, the row $A'_i$ \emph{is} $(k,\tau)$-regular
(this holds, recalling that $k \leq n/2$, since $n-k \geq k \geq 1/\tau^2$)
and its tail has positive 2-norm.  Now we can scale up $(A'_i,b'_i)$ as before to make the tail have $2$-norm~$1$.
\end{proof}

\section{Bentkus's mollifier and its properties} \label{sec:mollifier}

In this section we introduce and analyze Bentkus's orthant mollifier $\wt{\calO}_b: \R^m \to (0,1)$, which is a smoothed version of the translated orthant indicator function $\calO_b: \R^m \to \zo$ from \Cref{sec:ROT}.

%\gray{ We write $\psi_{b_i} : \R \to \zo$ to denote the indicator of the left-infinite interval $(-\infty,b_i]$.
%\onote{Mildly weird to me to use a subscript on what is essentially a dummy variable.}  Note that $\prod_{i \in [m]} \psi_{b_i}(v_i) = \Ind[v \in \calO_b]$ for all $b,v \in \R^m$.
%}

\begin{definition}[Gaussian-mollified halfline]
For $\theta \in \R$ and $\lambda > 0$, we define the ${\calC^\infty}$ function ${\tInd_{\theta,\lambda}}: \R \to (0,1)$,
\[ {\tInd_{\theta,\lambda}}(t) = \Ex_{\bg\sim N(0,1)}\big[\Ind[t + \lambda\bg \le \theta]\big].  \]
\end{definition}

\begin{definition}[Bentkus's orthant mollifier]
\label{def:bentkus}
For $b \in \R^m$ and $\lambda >0$, the \emph{Bentkus $\lambda$-mollifier} for $\calO_b$ is defined to be the $C^\infty$ function $\wt{\calO}_{b,\lambda}: \R^m \to (0,1)$,
%\onote{It's getting a little notation-heavy, I'm thinking; might tax a poor reader's brain to memorize $\psi$ and $\Psi$.  Could $\psi_b$ just be $\Ind_{(-\infty,b]}$?  Could $\Psi_{b,\lambda}$ be, I dunno, $\wt{\calO}_b^\lambda$ or something?}\onote{If we want to emphasize the idea behind B's mollifier, maybe put $\E[\calO_b(v+\lambda \bg)]$ below?  If instead we want to emphasize the product structure, maybe intro a name for the Gaussian-mollified halfline (which I think we'll need to do) and then write $\Psi$ is the product of these mollified halflines.}
\[ \wt{\calO}_{b,\lambda}(v) = \Ex_{\bg \sim N(0,1)^m} \big[ \calO_b(v + \lambda\bg)\big]. \]
\end{definition}
Since $\calO_b(v) = \prod_{i=1}^m \Ind[v_i \le b_i]$ {and $\calN(0,1)^m$ is a product distribution}, the mollifier $\wt{\calO}_{b,\lambda}$ can be equivalently defined as follows:
\begin{equation} \wt{\calO}_{b,\lambda}(v) = \prod_{i=1}^m {\tInd_{b_i,\lambda}}(v_i). \label{eq:bentkus-product}
\end{equation}
This product structure of Bentkus's mollifier will be crucially important for us in the analysis that we carry out in~\Cref{sec:singleswap}.
We note the following translation property of Bentkus's mollifier:
\begin{fact}
%\onote{Seems hardly worth even making a proposition out of this, but if we do, I always feel a ``Fact'' is more trivial than a ``Proposition'', yet compare this guy with the next guy\dots}
\label{fact:shift}
For all $b,v,\Delta \in \R^m$ and $\lambda >0$, we have $\wt{\calO}_{b,\lambda}(v+\Delta) = \wt{\calO}_{b-v,\lambda}(\Delta)$.
\ignore{\rnote{Of course it's equivalent, but is it more natural to have the RHS be $\tilde{\calO}_{b-v,\lambda}(\Delta)$
or to have it be
$\tilde{\calO}_{b-\Delta,\lambda}(v)$
?}}
\end{fact}

%\lnote{Say something about the smoothness of $\wt{\Ind}$ and $\wt{\calO}$, since we'll now be taking derivatives? \red{Rocco:} I added a few ``$C^\infty$'''s in the definitions}
In \Cref{sec:singleswap} we will also use the following global bound on the magnitude of the derivatives of the Gaussian-mollified halfline:

\begin{fact}[Standard; see Exercise 11.41 in \cite{ODbook}]
\label{fact:psi-derivatives}
%\rnote{I guess we should prove, cite, or bully the reader into accepting this?}
For all $\theta \in \R$, $\lambda > 0$, and  integer $\tay \ge 1$,
%\lnote{If we can be sure that we are going to take $d$ to be constant, then let's just drop the subscripts on the $O_d$'s.  \red{Rocco:} (nitpicking, apologies) I'm not crazy about writing ``constant $d$'' --- what is the possible parameter with respect to which $d$ might be viewed as non-constant?  And my preference is to be more precise and write $O_d$\dots but I'm okay with it either way.}\onote{concur with Rocco on this},
%\onote{Feel free to never speak to me again, but I guess if you write O and then \textbackslash left and then stuff, LaTeX puts too much space between the O and the paren.  I believe it's literally a bug in LaTeX.  I tend to change these to a parens macro.}
%\onote{I'm confused about what's written below; isn't $\psi_{b_i}$ not even differentiable?  I think perhaps someone intended that $\psi$ denote some mollified version of the halfline throughout?  By the way, although this fact was probably proved 200 years ago, I guess we could cite Ex11.41 in my book for it.}
\[ \norm*{ {\tInd_{\theta,\lambda}}^{(\tay)}}_\infty = O_d\parens*{\frac{1}{\lambda}}^d.\]
\end{fact}

The following result, from Bentkus~\cite[Theorem~3(ii)]{Bentkus:90}, can be viewed as a multidimensional generalization of \Cref{fact:psi-derivatives}.
(Strictly speaking \cite{Bentkus:90} only considers $b$'s of the form $(\theta, \theta, \dots, \theta)$, but by translation-invariance the bound holds for all $b \in \R^m.$)
\begin{theorem}[Bounded sum of derivatives]
\label{thm:bentkus}
For all $m \geq 2$, $b \in \R^m$, $\lambda > 0$,
and integer $\tay \ge 1$, %\rnote{I'm not clear: are we gonna end up taking $c$ super-constant?  If so we will have to keep an eye on what exactly is the $O_c$'s below I think.}
\[
\sup_{v \in \R^m}\Bigg\{ \sum_{|\alpha|=\tay} |\partial_\alpha \wt{\calO}_{b,\lambda}(v)|\Bigg\}  = O_d\parens*{\frac{\sqrt{\log m}}{\lambda}}^d. \]
\end{theorem}

\ignore{Glossary:

``\emph{$(\Lambda,\delta)$-quasi-lower sandwich'' $\Longrightarrow$ `` \emph{$(\Lambda,\delta)$-inner approximator}}''

``\emph{$(\Lambda,\delta)$-quasi-upper sandwich'' $\Longrightarrow$ `` \emph{$(\Lambda,\delta)$-outer approximator}}''

}

Recall from (\ref{eq:boundaries}) that
$\Game_{-\Lambda}\calO_b = \calO_b \setminus \calO_{b - (\Lambda, \dots, \Lambda)}$
and
$\Game_{+\Lambda}\calO_b = \calO_{b + (\Lambda, \dots, \Lambda)} \setminus \calO_{b}.$
We will use the following notions of approximation for translated orthants:

\begin{definition}[Inner and outer approximators for orthants] We say that $\Upsilon : \R^m \to [0,1]$ is a
\emph{$(\Lambda,\delta)$-inner approximator for $\calO_b$} if
\[ |\Upsilon(v) - \calO_b(v)| \le \delta \quad \text{for all $v \notin \Game_{-\Lambda}\calO_b$}. \]
Similarly, we say that $\Upsilon$ is a \emph{$(\Lambda,\delta)$-outer approximator for $\calO_b$} if
\[ |\Upsilon(v) - \calO_b(v)| \le \delta \quad \text{for all $v \notin \Game_{+\Lambda}\calO_b$}.  \]
\end{definition}

The connection between Bentkus's mollifier and these notions of approximation is established in the following claim.

\begin{lemma}[Bentkus's mollifier, appropriately translated, yields inner and outer approximators for translated orthants]
\label{lemma:bentkus-sandwich}
For all $b \in \R^m$ and $\lambda,\delta \in (0,1)$, there are $b^{\inner},b^{\outter} \in \R^m$ such that $\wt{\calO}_{b^{\inner},\lambda},\wt{\calO}_{b^{\outter},\lambda}$ are $(\Lambda,\delta)$-inner and -outer approximators for $\calO_b$ respectively,
where $\Lambda = \Theta(\lambda \sqrt{\log(m/\delta)}).$\ignore{ \rnote{Note that  the earlier definition of $\Lambda$ was $\Lambda = \Theta(\lambda \sqrt{\log(1/\delta)})$.  Is there a way to not pay this $\sqrt{\log m}$ factor? Not seeing one right now.}}
\end{lemma}

\begin{proof}
Let $b^{\inner} = b - \beta \Ind_m$ where $\beta = \Theta(\lambda \sqrt{\log(m/\delta)}) < \Lambda$ will be specified in more detail later. We show below that $\wt{\calO}_{b^{\inner},\lambda}$ is an $(\Lambda,\delta)$-inner approximator for $\calO_b$; an analogous argument in which the $v \in \calO_b$ and $v \notin \calO_b$ cases switch roles
shows that $\wt{\calO}_{b^{\outter},\lambda}$ is a $(\Lambda,\delta)$-outer approximator for $\calO_b$, where $b^{\outter} = b + \beta\Ind_m.$
\ignore{
%\rnote{The ``entirely analogous argument'' is why we take $\beta$ to be $\Theta(\lambda \sqrt{\log(\red{m}/\delta)})$; for the argument given below, analyzing just $\wt{\calO}_{b^{\inner},\lambda}$, it would be enough to have $\beta = \Theta(\lambda \sqrt{\log(1/\delta)}).$ (Taking $\beta = \Theta(\lambda \sqrt{\log(1/\delta)})$ wouldn't result in any savings in the theorem's overall bound, so it really doesn't matter.) Basically the $v \in \calO_b$ and $v \notin \calO_b$ cases switch roles for the outer approximator I think.
%
%We can add more exposition to the proof detailing some of this if people want?  I think it would be overkill to do the whole $\wt{\calO}_{b^{\outter},\lambda}$ analysis though.}
}

Fix $v \notin \Game_{-\Lambda}\calO_b$.  There are two possibilities:  either $v \in \calO_b$, or $v \notin \calO_b.$
We first consider the case in which $v$ lies in $\calO_b$.  Since $v \notin \Game_{-\Lambda}\calO_b$, we have $v_i \leq b_i - \Lambda$ for all $i \in [m].$ Since $\calO_b(v)=1$, we must show that $\wt{\calO}_{b^{\inner},\lambda}(v) \geq 1-\delta$. Recalling \Cref{eq:bentkus-product} and the fact that the function ${\tInd_{\theta,\lambda}}: \R \to (0,1)$ is monotone decreasing for all $\theta \in \R$ and $\lambda>0$, it suffices to show that $\wt{\calO}_{b^{\inner},\lambda}(b - \Lambda \Ind_m) \geq 1-\delta.$ Again by \Cref{eq:bentkus-product} this holds if and only if
\[
\prod_{i=1}^m {\tInd_{b_i - \beta,\lambda}}(b_i - \Lambda) \geq 1-\delta,
\]
which is equivalent to
\[
{\parens*{\Prx_{\bg \sim \calN(0,1)}\bracks*{\bg \leq (\Lambda - \beta)/\lambda}}}^m \geq 1-\delta,
\]
which holds if
\begin{equation} \label{eq:goal1}
\Prx_{\bg \sim \calN(0,1)}[\bg \leq (\Lambda - \beta)/\lambda] \geq 1-\delta/m.
\end{equation}
By the well-known Gaussian tail bound $\Pr[ \bg  \geq t] \leq 1 - \frac1{t\sqrt{2\pi}} e^{-t^2/2}$ for $t > 0$ (see e.g. \cite{Feller}, Section 7.1), we see that to achieve \Cref{eq:goal1} it suffices to have $\Lambda - \beta \geq C \lambda \sqrt{\ln(m/\delta)}$ for an absolute constant $C>0$, and hence $\Lambda = \Theta(\lambda \sqrt{\log({m}/\delta)})$ suffices.

Now we turn to the case in which $v \notin \calO_b$, and hence for some $i \in [m]$ we have $v_i > b_i$; without loss of generality we suppose that $v_1 > b_1.$  Since $\calO_b(v)=0$ in this case, we must show that $\wt{\calO}_{b^{\inner},\lambda}(v) \leq \delta$.
By \Cref{eq:bentkus-product} this holds if and only if
\[
\prod_{i=1}^m {\tInd_{b_i - \beta,\lambda}}(v_i) \leq \delta,
\]
which holds if
\[
{\tInd_{b_1 - \beta,\lambda}}(v_1) \leq \delta,
\]
which is equivalent to
\[
\Prx_{\bg \sim \calN(0,1)}[v_1 + \lambda \bg \leq b_1 - \beta] \leq \delta.
\]
Recalling that $v_1 > b_1$, it suffices to have
\[
\Prx_{\bg \sim \calN(0,1)}[\bg \leq - \beta/\lambda] \leq \delta,
\]
which holds (with room to spare) for our choice of $\beta$ by the standard Gaussian tail bound.
\end{proof}

\subsection{The connection between inner/outer approximators and CDF distance}

The following elementary properties of inner/outer approximators  will be useful for us:

\begin{fact}
\label{fact:two-properties}
 Fix $b \in \R^m$ and let $\Upsilon^{\inner}, \Upsilon^{\outter}$ be $(\Lambda,\delta)$-inner and -outer approximators for $\calO_b$.  Then
\begin{enumerate}
\item
$\Upsilon^{\inner}(v) - \delta \le \calO_b(v) \le \Upsilon^{\outter}(v) + \delta$ for all $v \in \R^m$.
\item  $\Upsilon^{\inner}$ is a $(\Lambda,\delta)$-\emph{outer} approximator for $\calO_{b-\Lambda\Ind_m}$, and similarly $\Upsilon^{\outter}$ is a $(\Lambda,\delta)$-\emph{inner} approximator for $\calO_{b + \Lambda\Ind_m}$.
\end{enumerate}
\end{fact}

The next lemma is straightforward but very useful for us.   Intuitively, it says that in order for an $\R^m$-valued random variable $\tilde{\bv}$ to fool a translated orthant $\calO_b$ relative to another $\R^m$-valued random variable $\bv$, it suffices to (i) have $\tilde{\bv}$ fool both inner and outer approximators for $\calO_b$, and (ii) establish anticoncentration of the original random variable $\bv$ at the inner and outer boundaries of $\calO_b$. We explain in detail how we will use this lemma after giving its proof below.

\begin{lemma}
\label{lem:soft-to-hard}
Let $\Upsilon^\inner, \Upsilon^\outter : \R^m \to [0,1]$ be $(\Lambda,\delta)$-inner and -outer approximators for $\calO_b$.  Let $\bv$ and $\tilde{\bv}$ be $\R^m$-valued random variables satisfying:
\begin{equation} \big| \E\big[\Upsilon(\bv)\big] -  \E\big[\Upsilon(\tilde{\bv})\big]\big| \le \gamma \label{eq:fool-Upsilon}
\end{equation}
for both $\Upsilon \in \{ \Upsilon^{\outter},\Upsilon^{\inner}\}$.  Then
\[ \big| \Pr\big[ \bv \in \calO_b\big] - \Pr\big[\tilde{\bv}\in \calO_b\big] \big| \le \gamma + 2\delta + \Pr\big[ \bv \in \Game_{\pm\Lambda}\calO_b\big]. \]
\end{lemma}

\begin{proof} The proof follows similar lines to the arguments used to prove Lemma~3.3 in \cite{HKM12}.
We first note that
\begin{align*}
\Pr\big[\tilde{\bv} \in \calO_{b}\big] & \le \E\big[\Upsilon^{\outter}(\tilde{\bv})\big] + \delta \tag*{(Item 1 of \Cref{fact:two-properties})} \\
&\le (\E\big[\Upsilon^{\outter}(\bv)\big] + \gamma)  + \delta \tag*{(\Cref{eq:fool-Upsilon} with $\Upsilon = \Upsilon^{\outter}$)} \\
&\le \Pr\big[\bv \in \calO_{b+\Lambda\Ind_m}\big] + \gamma + 2\delta. \tag*{(Item 2 of \Cref{fact:two-properties})}
\end{align*}
Combining this with a symmetric argument for the lower bound, we have:
\begin{equation}  \Pr\big[\bv \in \calO_{b - \Lambda\Ind_m}\big] - \gamma -2\delta \le
\Pr\big[\tilde{\bv} \in \calO_b \big] \le  \Pr\big[\bv \in \calO_{b+ \Lambda\Ind_m}\big] + \gamma + 2\delta.\label{eq:levy}
\end{equation}
To convert this type of closeness into CDF closeness, we observe that
\begin{align*}
\Pr\big[\bv \in \calO_{b+ \Lambda\Ind_m}\big]  &= \Pr\big[\bv \in \calO_b\big] + \Pr\big[\bv \in \Game_{+\Lambda}\calO_b\big]  \\
\Pr\big[\bv \in \calO_{b- \Lambda\Ind_m}\big]  &= \Pr\big[\bv \in \calO_{b}\big] -\Pr\big[\bv \in \Game_{-\Lambda}\calO_b \big].
\end{align*}
Plugging these identities into \Cref{eq:levy}, we conclude that
\begin{align*}
\Pr\big[\tilde{\bv} \in \calO_b \big] &=  \Pr\big[\bv \in \calO_b \big] \pm \big(\gamma+ 2\delta + \Pr\big[\bv \in \Game_{+\Lambda}\calO_b\big] +\Pr\big[\bv \in \Game_{-\Lambda}\calO_b \big] \big) \\
&= \Pr\big[\bv \in \calO_b \big] \pm \big(\gamma + 2\delta + \Pr\big[\bv \in \Game_{\pm\Lambda}\calO_b\big] \big),
\end{align*}
thus completing the proof of \Cref{lem:soft-to-hard}.\end{proof}

\subsubsection{Applying \Cref{lem:soft-to-hard} in the context of \Cref{thm:fool-k-tau-regular}, and the organization of the rest of this paper.}
Applying \Cref{lem:soft-to-hard} with $\bv$ and $\tilde{\bv}$ being $A\bu$ and $A \bz$ respectively, the task of bounding
\[ \Big| \Prx_{\bu \sim \bn}\big[A\bu \in \calO_b \big] - \Prx_{\bz \sim \mathscr{G}_{\mathrm{MZ}}}\big[ A\bz \in \calO_b\big] \Big| \]
reduces to the following two-step program:
\begin{enumerate}
\item Establishing anticoncentration within orthant boundaries: bounding $\Pr\big[ A\bu \in \Game_{\pm \Lambda}\calO_b\big]$; and,
\item Fooling Bentkus's mollifier: bounding $\big|\E\big[ \wt{\calO}(A\bu)\big] - \E\big[\wt{\calO}(A\bz)\big]\big|$ for $\wt{\calO} \in \{ \wt{\calO}_{b^{\outter},\lambda},\wt{\calO}_{b^{\inner},\lambda}\}$, the inner and outer approximators for $\calO_b$ given by \Cref{lemma:bentkus-sandwich}.
\end{enumerate}
\Cref{sec:anticonc} is devoted to the former, and \Cref{sec:fool-bentkus} the latter. In \Cref{sec:put-together} we put these pieces together to prove \Cref{thm:fool-k-tau-regular}.

%!TEX root = main.tex

\section{Boolean anticoncentration within orthant boundaries}
\label{sec:anticonc}

The main result of this section is \Cref{thm:anticonc}, which provides the first step of the two-step program described at the end of \Cref{sec:mollifier}:

\begin{theorem}[Boolean anticoncentration within orthant boundaries]
\label{thm:anticonc}
Assume $A \in \R^{m\times n}$ satisfies the following property: each of its row vectors has a $\tau$-regular subvector of $2$-norm~$1$, where $\tau$ is as set in~\Cref{sec:params}.\footnote{Equivalently, $A$ is $(n, \tau)$-standardized.}
%%%\onote{Er, sorry about this weirdness, but it was weird that $k$ showed up in the hypothesis of the theorem, but not the conclusion.  Any improving suggestion? \blue{Li-Yang: How about saying $(n,\tau)$-standardized instead of $(\infty,\tau)$-standardized in the footnote?  (Actually, I think it'd also be fine if we didn't have a footnote at all.)} \red{Rocco:  I slightly reworded the footnote --- somehow the ``One might say'' and the scare quotes and the $\infty$ seemed weird to me, but the current phrasing just points out that this is equivalent to a previously-defined notion; somehow I'm more comfortable with that.  I agree with Ryan that it's weird to have a $k$ in the hypothesis but not the conclusion, and I agree with Li-Yang that it'd be nice to avoid $\infty$. }}
Then for all $b \in \R^m$ and $\Lambda \geq \tau$,
%\rnote{Was $\Lambda \in (0,1)$ but if $\Lambda = 2^{-2^{2^n}}$ this isn't true, right?  We need $\Lambda \geq \tau$ in order to be able to do the union bound (building up the $\Lambda$-width strip by $\Lambda/\tau$ many disjoint $\tau$-width strips) giving us this bound.}
we have
\[ \Prx_{\bu \sim \bn} \bracks*{A\bu \in \Game_{\pm\Lambda} \calO_{b}}= O\parens*{\Lambda\sqrt{\log m}}.\]
\end{theorem}

En route to proving \Cref{thm:anticonc} we will establish a ``Littlewood--Offord theorem for polytopes,'' \Cref{thm:LO-for-polytopes}, that was stated in \Cref{sec:LO}.  \Cref{thm:LO-for-polytopes} will in fact be obtained as a special case of a more general result about intersections of $m$ arbitrary \emph{unate} functions (namely~\Cref{lem:thin-boundary}).  %\lnote{Should we give a brief description of the statement here? {\red{Rocco:  How about just a forward pointer?}}}

\begin{definition}[Unateness]
A function $F: \bn \to \{0,1\}$ is \emph{unate in direction $\sigma \in \bn$} if the function $G(x_1,\dots,x_n) = F(\sigma_1 x_1,\dots, \sigma_n x_n)$ is a monotone Boolean function, meaning that $G(x) \leq G(x)$ whenever $x_j \leq x_j$ for all $j \in [n].$ We  refer to $\sigma$ as the \emph{orientation} of $F$.
\end{definition}

Our analysis, dealing as it does with intersections of unate functions, is somewhat reminiscent of that of \cite{Kane14intersection}, and indeed we will establish the main result of \cite{Kane14intersection}---an upper bound of $O(\sqrt{n \log m})$ on the average sensitivity of any intersection of $m$ unate functions---in the course of our analysis.

%%%\paragraph{Roadmap for this section: How we'll get to \Cref{thm:anticonc}}
%%%\lnote{For my own benefit, to be deleted}
%%%\begin{itemize}
%%%\item \Cref{sec:anticonc-setup}: Preliminaries
%%%\item \Cref{sec:Kane-reproof}: Reproving Kane's average sensitivity bound
%%%\item \Cref{sec:surface}: Anticoncentration under the assumption that $|A_{ij}|\ge 1$ (\Cref{thm:LO-for-polytopes}, Littlewood--Offord for polytopes)
%%%\item \Cref{sec:semi-thin}: Robust variant where ``at least an $\alpha$ fraction of each row of $A$ has magnitude $\lambda$"
%%%\item \Cref{sec:grand-finale}: Grand finale: anticoncentration for $(k,\tau)$-standardized $A$'s (\Cref{thm:anticonc})
%%%
%%%\end{itemize}
%%%

\subsection{Caps and their boundary edges}
\label{sec:anticonc-setup}

Let $G$ and $H$ be subsets of $\bn$.  We typically think of $G$ as a $G$eneral/arbitrary set and $H$~as a $H$alfspace, though formally $H$ will only need to be unate.  Throughout this section we write $\sigma \in \bn$ to denote the orientation of $H$.

We call the set $G \setminus H$ the \emph{cap},  the set $G\cap H$ the \emph{body}, and the complement of $G$ the \emph{exterior}.
%%   \rnote{The fates are conspiring against us, it's inopportune that ``G cap H'' is how one reads/says the body and not the cap...
%%
%%Should we consider replacing  ``cap'' by ``dome'' or (not entirely joking) ``toque'' or something?  No joking, I literally say ``G cap H'' in my mind when I read ``$G \cap H$'' (I still sometimes move my lips when I read).  Or is this confusing only to me?  \blue{Li-Yang:} Yeah good point, I can fix this at the end.} \onote{I really?  I read it as ``gee intersect aitch''.  I dunno, I feel like `cap' is like `spherical cap'; kinda like it. \red{Rocco:  I'm okay with sticking with ``cap'' because of the association with ``spherical cap''; if you guys are too, let's delete this whole *note discussion.}}
%%%\rnote{Maybe we plant a seed for the later terminology (which I like) and say ``and we refer to the complement of $G$ as the \emph{exterior}'' here?}
Please refer to \Cref{fig:gcap}, where $G$ is the union of the two regions with blue shading and $H$ is the gray-shaded region (depicted as a halfspace in the figure).  The upward arrows in the diagram illustrate some edges of the hypercube.  We have oriented these edges according to $\sigma$: for an edge $\{x,y\}$ in the $j$-th direction in which $x_j = -1$ and $y_j = 1$, the tail of the corresponding arrow represents $x$ if $\sigma_j = -1$, and $y$ if $\sigma_j = 1$.  Note in particular that the edges are oriented ``away" from $H$ (i.e., so that $H$ is antimonotone with respect to the edge orientations).

\myfig{1}{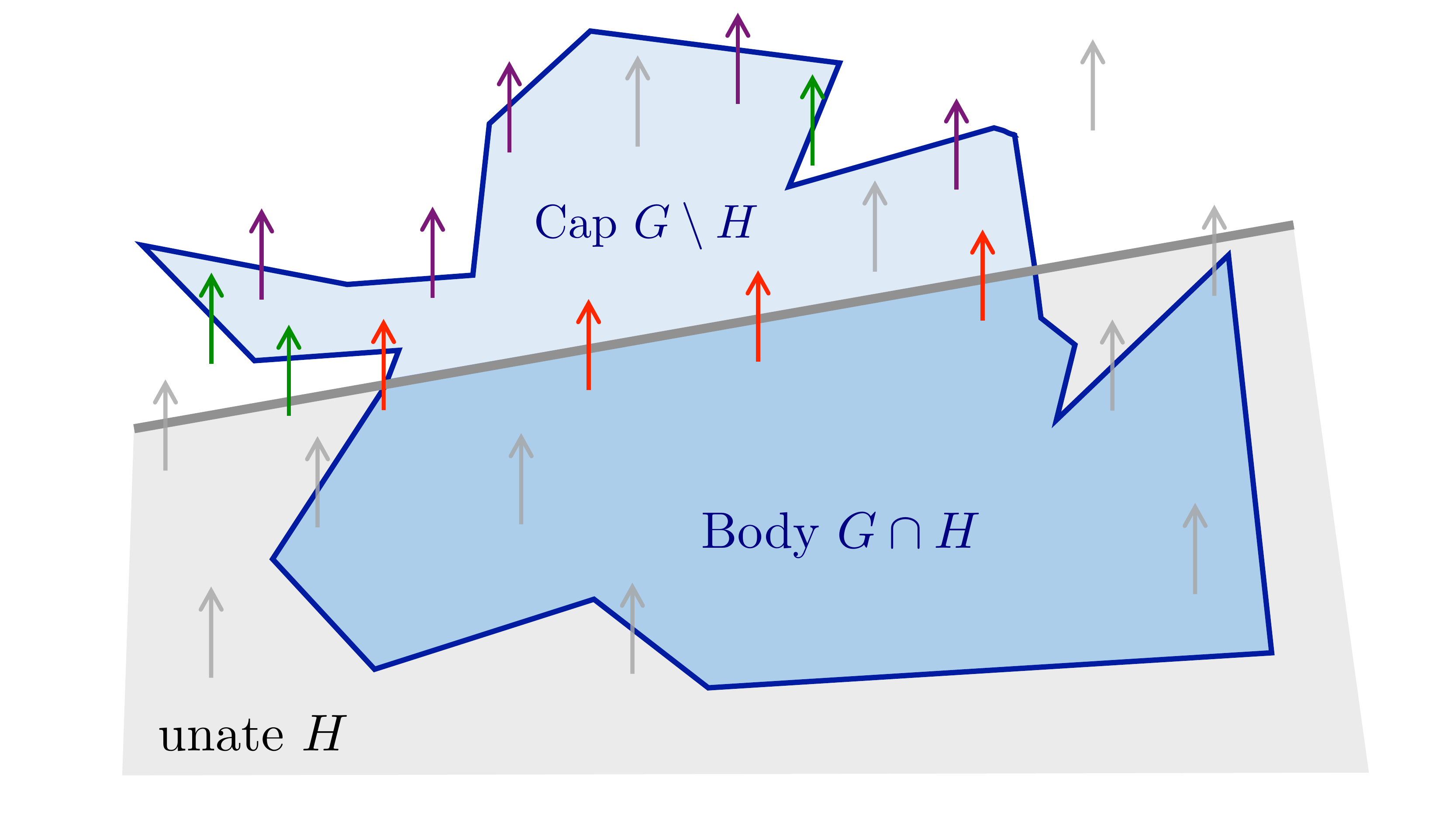}{Illustration of a cap and body}{fig:gcap}

  We will be concerned with the \emph{boundary} edges for the cap $G \setminus H$; these are edges which have one endpoint inside $G \setminus H$ and one endpoint outside it.

  \begin{definition}[Edge boundary] \label{def:AInf}
For a general set $F \subseteq \bits^n$, let $\AInf(F)$ denote the fraction of all $n2^{n-1}$ edges that are boundary edges for~$F$.%\onote{I commented out the Fourier talk here, thinking we'd just save it for \Cref{lem:kane-ineq}} %For Fourier analysis aficionados, this is $\Tinf[F]/n$, where $\Tinf[F] = \sum_{j=1}^n \Tinf_j[F]$ is the total influence (equivalently, average sensitivity) of $F$.\footnote{Here and subsequently we will sometimes blur the distinction between a subset $F \subseteq \bn$ and its $\zo$-valued indicator function where there is no risk of confusion.}\lnote{Should this footnote come earlier, maybe even in the preliminaries?   Do we blur this distinction earlier in the paper?}\onote{Not to be both pompous and annoying, but my book defines Influences for 0/1-functions differently from Influences for $\pm 1$-valued functions.  If $F$ is blurred with a 0/1 indicator, one might have to insert a factor of 2 somewhere?  Perhaps we can drop the ``$\zo$-valued'' adjective in the footnote}
\end{definition}

  We distinguish the three possible types of boundary edges of the cap $G \setminus H$:
%  \lnote{IB has become BC, ID has become EC, and OD has become CE.  (Right now I'm thinking that we can get away with not introducing the terms ``dome" and ``base", and also just referring to gcaps as caps; I may have to retract this later$\dots$)}
\begin{itemize}
\item[$\circ$] {\bf Body$\to$Cap (BC) edges:} the red edges in the diagram.  Formally, these are edges where the tail is in the body $G\cap H$ and the head is in the cap $G\setminus H$.
\item[$\circ$] {\bf Exterior$\to$Cap (EC) edges:} the green edges in the diagram.  Formally, these are edges where the tail is not in $G$, and the head is in the cap $G\setminus H$.
\item[$\circ$] {\bf Cap$\to$Exterior (CE) edges:} the purple edges in the diagram.  Formally, these are edges where the tail is in the cap $G\setminus H$ and the head is not in $G$.
\end{itemize}
\begin{remark}
Note that {\bf there are no Cap$\to$Body (CB) edges}.  Formally, these would be the last possibility for $G \setminus H$ boundary edges, namely ones with tail in the cap $G \setminus H$ and head in the body $G \cap H$. But these cannot exist due to the antimonotonicity of~$H$ vis-a-vis the edges; if the tail is already not in~$H$, then the head cannot be in~$H$.
\end{remark}

Given a cap~$C = G \setminus H$, we write $\BC(G,H)$, $\EC(G,H)$, $\CE(G,H)$ for the fraction of hypercube edges of each of the three above types.  Therefore $\AInf(C) = \BC(G,H) + \EC(G,H) + \CE(G,H)$.
%\gray{Given a cap~$C$, we write $\BC(C)$, $\EC(C)$, $\CE(C)$ for the fraction of hypercube edges of each of the three above types.\onote{Formally this doesn't make sense, as $\BC(C)$ depends not just on the set $C$, but on its decomposition from $G$ and $H$.  Given that, and also given the fact that $\BC(C)$ looks a little weird with the repetitive $C$, perhaps we could change notation to $\BC(G,H)$ and similarly? {\blue{Li-Yang:} Okay will fix later.}} }

We will also be interested in the \emph{directed} edge boundary of caps:

\begin{definition}[Directed edge boundary] \label{def:AOInf}
    For a cap $G \setminus H$, define
    \begin{equation}    \label{eqn:AOinf}
        \AOInf(G,H) = \BC(G,H) + \EC(G,H) - \CE(G,H),
    \end{equation}
    the fraction of inward boundary edges minus the fraction of outward boundary edges.   \end{definition}

It will be very useful for us to have an upper bound on $ \Ainf(G \cap H) - \Ainf(G)$, the change in $\AInf(G)$ when we intersect $G$ with $H$ (note that this quantity can be either positive or negative).  The following fact is immediate from the definitions:

\begin{fact}[Change in boundary size]                                        \label{fact:capInf}
    If $G \setminus H$ is a cap, then
    \begin{equation}    \label{eqn:Ainf}
        \Ainf(G \cap H) - \Ainf(G) = \BC(G,H) -\EC(G,H) -\CE(G,H).
    \end{equation}
\end{fact}

Comparing \Cref{eqn:Ainf,eqn:AOinf}, we plainly have:

\begin{fact}
\label{fact:plainly}
 $\Ainf(G \cap H) - \Ainf(G) \leq \AOInf(G,H)$.\end{fact}

 To get a quantitative bound, we have the following lemma:

\newcommand{\GISO}{U}
\newcommand{\vol}{\mathrm{vol}}
\begin{lemma}                                        \label{lem:kane-ineq}
    For any cap~$C = G \setminus H$,
    \[
        \AOInf(G,H) \leq \frac{\GISO(\vol(C))}{\sqrt{n}},
          %= O\parens*{\frac1{\sqrt{n}}} \cdot  \Phi(\E[C]),
    \]
    where $\vol(C) = |C|/2^n$ and $\GISO$ denotes the function $\GISO(p) = 2p\sqrt{2\ln(1/p)}$.
%where $\Phi$ denotes the function $\Phi(p) = p \sqrt{\ln(e^{1/2}/p)}.$
\end{lemma}
\begin{proof}
    This is a basic fact in analysis of Boolean functions.  Identifying $C$ with its indicator function $C : \{-1,1\}^n \to \{0,1\}$, we have $\vol(C) = \E[{C(\bu)}]$ and
    \[
        \AOInf(G,H) = 2\mathop{\Ex_{\bu \sim \{-1,1\}^n}}_{\bj \sim [n]}[C(\bu)\sigma_{\bj}\bu_{\bj}] = \frac{2}{n} \sum_{j=1}^n \sigma_j \wh{C}(\{j\}),
    \]
    where $\wh{C}(\{j\})$ denotes the degree-1 Fourier coefficient of $C$ corresponding to coordinate $j$.
    It is well known and elementary that for $F : \{-1,1\}^n \to \{0,1\}$ with $\E[F] = p$, one has $\sum_{j=1}^n |\wh{F}(\{j\})| \leq O(p \sqrt{\ln(1/p)})\sqrt{n}$; see, e.g., Kane's paper~\cite[Lemma~6]{Kane14intersection} for the short proof.  For the sake of an asymptotically tight constant, we can use Cauchy--Schwarz and the Fourier ``Level-1 Inequality'' \cite{Tal96,Chang02,IMR14}
%    \onote{haha, after ranting and raving about not using ``Chang'', I appealed to it to get a sharp constant briefly :)  Sorry --- y'all were right! Please double-check my factors of 2, by the way!}\lnote{Charikar was cited here. I changed it to Chang.}
    to get
    \[
        \sum_{j=1}^n \sigma_j \wh{C}(\{j\}) \leq \sqrt{n}\cdot \sqrt{\sum_{j=1}^n \wh{C}(\{j\})^2} \leq \sqrt{n}\cdot  p \sqrt{2\ln(1/p)}. \qedhere
    \]
\end{proof}

\subsubsection{Reproving the main result of \cite{Kane14intersection}}
\label{sec:Kane-reproof}

We can now reprove the main result of \cite{Kane14intersection} (which we will use later):

\begin{theorem}[\cite{Kane14intersection}]
\label{thm:kane-unate}
Let $F$ be the intersection of $m \geq 2$ unate functions over $\bn$.  Then
\begin{equation} \label{eq:kanebound}
    \Ainf(F) \leq \frac{2\sqrt{2\ln m}}{\sqrt{n}}.
\end{equation}
(Equivalently, an intersection of $m \geq 2$ unate functions has average sensitivity at most $2\sqrt{2\ln m}\sqrt{n}$.)
\end{theorem}

\begin{proof} Let $H_1, \dots, H_m$ be unate functions and define associated caps
\begin{equation} C_i = (H_1 \cap \cdots \cap H_{i-1}) \setminus H_i, \label{eq:cap-def}
\end{equation}
with $C_1 = \bits^n \setminus H_1$ (i.e.~$H_0=\bn$). Letting $F = H_1 \cap \cdots \cap H_m$, we have that the complement $F^c = \bn \setminus F$ of $F$ can be expressed as a disjoint union of caps:
\begin{equation}    \label{eq:sqcup}
    F^c = C_1 \sqcup \cdots \sqcup C_m.
\end{equation}
For intuition, we may think of the intersection of $m$ unate sets $F$ as being formed in $m$ stages, starting with $\bn$ and successively intersecting with each $H_i$;  given this interpretation, $C_i$ is the portion of $\bn$ that is removed in the $i$-th stage.  With this notation in hand, we have that
\begin{align*}
\AInf(F) &= \sum_{i=1}^m \AInf((H_1 \cap \cdots \cap H_{i-1}) \cap H_i) - \AInf(H_1 \cap \cdots \cap H_{i-1}) \\
&\le \sum_{i=1}^m \AOInf(H_1 \cap \cdots \cap H_{i-1}, H_i)  \tag*{(\Cref{fact:plainly} with $G = H_1\cap \cdots \cap H_{i-1}$ and $H = H_i$)}  \\
&\le \frac1{\sqrt{n}} \cdot \sum_{i=1}^m \GISO(\vol(C_i)) \tag*{(\Cref{lem:kane-ineq})}.
\end{align*}
Finally,
\[
    \sum_{i=1}^m \GISO(\vol(C_i)) \leq m \cdot \GISO\parens*{\frac{\sum_{i=1}^m \vol(C_i)}{m}} = m \cdot \GISO\parens*{\frac{\vol(F^c)}{m}} \leq m \cdot \GISO(\tfrac{1}{m}) = 2\sqrt{2\ln m},
\]
where we used concavity of $\GISO$, then \Cref{eq:sqcup},  then the fact that $\GISO$ is increasing on~$[0,1/2]$. This completes the proof of~\Cref{thm:kane-unate}.
\end{proof}

\subsection{A Littlewood--Offord theorem for polytopes (\Cref{thm:LO-for-polytopes})} \label{sec:surface}

In this section we prove \Cref{thm:LO-for-polytopes}:

\begin{reptheorem}{thm:LO-for-polytopes}
There is a universal constant $C$ ($C = 5\sqrt{2}$ suffices) such that the following holds.  For all $m \geq 2$, $b \in \R^m$ and $A \in \R^{m\times n}$ with
$|A_{ij}| \geq 1$ for all $i \in [m]$ and $j \in [n]$,
\[ \Prx_{\bu \sim \bn} \bracks*{A\bu \in \Game_{-2}\calO_b} \leq\frac{C\sqrt{\ln m}}{\sqrt{n}}.  \]
\ignore{
\rnote{Somehow now looking at the above, ``$\Game^{-2}\calO_b$'' looks weird to me, like an exponent of $-2$\dots should we consider a subscript instead? \blue{Li-Yang: Okay, I changed the superscripts to subscripts.  Hope I got them all..}}\onote{agreed, kinda weird.  Actually, $\Game$ itself looks weird to me.  Granted, that's kind of an O'DonnellCriticizesO'Donnell comment.}
}
\end{reptheorem}

%%%
%%%\gray{Before proving \Cref{thm:LO-for-polytopes}, we recall the ``point probability'' version of the classic Littlewood--Offord theorem:
%%%
%%%\begin{theorem}[Littlewood--Offord, point probability version]
%%%Let $w \in (\R \setminus \{0\})^n$.  For all $b \in \R$,
%%%\[ \Pr\big[\, w \cdot \bu = b\,\big] \le {n \choose \lfloor n/2\rfloor }\cdot 2^{-n} = O\parens*{\frac1{\sqrt{n}}}. \]
%%%\end{theorem}
%%%
%%%The above result may be viewed as an upper bound on the probability that a random point from $\bn$ lies on the surface of a ``one-facet'' polytopes with no zero-coefficients.
%%%As a direct consequence of \Cref{thm:LO-for-polytopes} we obtain the following natural extension of this theorem to $m$-facet polytopes with no zero-coefficients:
%%%
%%%\begin{corollary}[Littlewood--Offord for polytopes, point probability version]
%%%\label{cor:point}
%%%Let $A \in (\R \setminus \{0\})^{m\times n}$.  For all $b\in \R^m$,
%%%\[\Pr\big[ \, A\bu \in \Game \calO_b\,\big] = O\parens*{\frac{\sqrt{\log m}}{\sqrt{n}}}. \]
%%%\end{corollary}}

We note in passing that that the anticoncentration bound given by \Cref{thm:LO-for-polytopes} is best possible up to constant factors.  Indeed, our matching lower bound applies even to the stricter event of falling \emph{on the surface} of $\calO_b$:

\begin{claim}[Optimality of \Cref{thm:LO-for-polytopes}] \label{claim:optimal}
For $2 \leq m \leq 2^n$, there is a matrix $A \in {\bits^{m \times n}}\ignore{(\R \setminus \{0\})^{m \times n}}$ and a vector $b \in \R^m$ such that
\[
    \Prx_{\bu \sim \bn} \bracks*{A\bu \in \Game \calO_b} = \Omega\parens*{\frac{\sqrt{\ln m}}{\sqrt{n}}}.
\]
\end{claim}

\noindent We prove \Cref{claim:optimal} in Appendix~\ref{ap:optimal}.

\subsubsection{Proof of \Cref{thm:LO-for-polytopes}}
\label{sec:proof-of-LO}
As mentioned at the beginning of this section, we will obtain \Cref{thm:LO-for-polytopes} as a corollary of a more general result about intersections of unate functions.   Let $H_1,\dots,H_m \subseteq \bn$ be unate sets, $m \geq 2$, and further suppose that we have additional unate sets $\ol{H}_1, \dots, \ol{H}_m$ such that $H_i \subseteq \ol{H}_i$ for all~$i$. (For intuition it may be helpful to think of $H_i$ as the ``interior'' of $\ol{H}_i$; see the proof of~\Cref{thm:LO-for-polytopes} using \Cref{lem:thin-boundary} just below for a typical example of sets $H_i$ and $\ol{H}_i$.) We define the following subsets of $\bn$:
\begin{align*}
F &= \overline{H}_1 \cap \cdots \cap \overline{H}_m  \\
F^\circ &= H_1 \cap \cdots \cap H_m \tag{{interior of $F$}}\\
\bdry F &= F \setminus F^\circ  \tag{{boundary of $F$}}\\
F^c &= \bn \setminus F \tag{{exterior of $F$}}\\
\bdry H_i &= \ol{H}_i \setminus H_i \text{ (for each $i \in [m]$).}
\tag{{boundary of $\ol{H}_i$}}
\end{align*}
%We also write $\bdry H_i = \ol{H}_i \setminus H_i$.
\begin{definition}[Thin sets]
    We say that $\bdry H_i$ is \emph{thin} if it does not contain any induced edges of the hypercube.
    %\rnote{This is the case in our standard setup when $\bdry H_i = \{ x : a^i \cdot x = b^i\}$ and $a^i$ has no zeroes.}
\end{definition}

\begin{lemma} \label{lem:thin-boundary}
If $\bdry H_i$ is thin for each $i \in [m]$, then $\vol(\bdry F) \leq \frac{5\sqrt{2\ln m}}{\sqrt{n}}$.
%\ignore{\onote{Oh man, you'll kill me for this, but I put in explicit constants here, and now I ask if you can double-check them :)}}\rnote{Checked Ryan's explicit constants and they look good to me!}
\end{lemma}

\begin{proof}[Proof of \Cref{thm:LO-for-polytopes} assuming \Cref{lem:thin-boundary}]
Fix any $b \in \R^m$ and $A \in \R^{m\times n}$ such that
$|A_{ij}| \geq 1$ for all $i \in [m]$ and $j \in [n]$, and let
\[
\ol{H}_i = \big\{x \in {\bn}\colon A^i  \cdot x \leq b_i\big\},
\quad \quad
H_i = \big\{x \in {\bn}\colon A^i  \cdot x \leq b_i - 2\big\},
\]
so
\begin{align*}
\bdry H_i &= \big\{x \in \bn\colon b_i - 2 < A^i  \cdot x \leq b_i \big\} \quad \text{and}\\
\bdry F &= \big\{x \in \bn \colon Ax \le b\ \& \   A^i \cdot x > b_i-2 \text{ for some $i\in [m]$}\big\} \\
&= {\big\{ x \in \bn \colon  Ax \in \Game_{-2}\calO_b \big\}}.
\end{align*}
Since $|A_{ij}| \geq 1$ for all $i,j$, it follows that each $\bdry H_i$ is thin, and hence \Cref{lem:thin-boundary} directly gives \Cref{thm:LO-for-polytopes}.
\end{proof}

The rest of this section will be devoted to the proof of \Cref{lem:thin-boundary}.   Recalling that $F^\circ$ is called the \emph{interior} of $F$ and $\bdry F$ is called the \emph{boundary} of $F$, we say that an edge in the hypercube is \emph{boundary-to-interior} if it has one endpoint in $\bdry F$ and the other endpoint in $F^\circ$, and we write $\nu_{BI}$ for the fraction of all edges that are of this type.  We similarly define \emph{boundary-to-exterior} edges and $\nu_{BE}$, with~$F^c$.    Note that every boundary-to-interior edge is a boundary edge for $F^\circ = H_1 \cap\cdots \cap H_m$, which is an intersection of $m$ unate sets.  By applying \Cref{thm:kane-unate} to $F^\circ$, we get that
\begin{equation}    \label{eqn:nu_BI}
    \nu_{BI} \leq \frac{2\sqrt{2\ln m}}{\sqrt{n}}.
\end{equation}
Similarly, every boundary-to-exterior edge is a boundary edge for $F = \ol{H}_1 \cap \cdots \cap \ol{H}_m$; applying \Cref{thm:kane-unate} to this intersection yields
\begin{equation}    \label{eqn:nu_BE}
    \nu_{BE} \leq \frac{2\sqrt{2\ln m}}{\sqrt{n}}.
\end{equation}
Next, we bound the fraction of edges that have both endpoints in $\bdry F$ and go between ``two different parts of~$\bdry F$. More precisely, for $x \in \bdry F$, define $i^\star(x)$ to be the least~$i$ for which $x \in \bdry H_i$ (equivalently, the least~$i$ for which $x \not \in H_i$).  We say that an edge ${\{x,y\}}$
%\rnote{Replaced $(x,y)$ with $\violet{\{x,y\}}$ --- since we're counting edges let's be careful about ordered/unordered, and here we mean unordered, right?}
is \emph{boundary-to-boundary$'$} if $x, y \in \bdry F$ but $i^\star(x) \neq i^\star(y)$; we write $\nu_{BB'}$ for the fraction of such edges.

\begin{observation}
\label{obs:thin}
    If every $\bdry H_i$ is thin, then every edge with both endpoints in $\bdry F$ is boundary-to-boundary$'$.  In this case, $\nu_{BI} + \nu_{BE} + \nu_{BB'}$ is exactly the fraction of edges in the cube that touch~$\bdry F$, which in turn is an upper bound on $\vol(\bdry F)$.
\end{observation}

Thus \Cref{lem:thin-boundary} follows from \Cref{eqn:nu_BI,eqn:nu_BE} and the following claim:
\begin{claim}[Boundary-to-boundary$'$ edges]
                                     \label{claim:nu_BB}
    $\displaystyle \nu_{BB'} \leq \frac{\sqrt{2\ln m}}{\sqrt{n}}$.
\end{claim}
\begin{proof}
    We define the caps $C_1, \dots, C_m$ with respect to the $H_i$'s as in \Cref{eq:cap-def} in the proof of \Cref{thm:kane-unate}.
   Subtracting \Cref{eqn:Ainf} from \Cref{eqn:AOinf} for each $C_i$ and summing over $i\in [m]$,
   \begin{align*}
   2\sum_{i=1}^m \EC(H_1 \cap \cdots \cap H_{i-1}, H_i) &= \sum_{i=1}^m \AOInf(H_1 \cap \cdots \cap H_{i-1}, H_i) \\
   & \ \ \ - \bigg(\sum_{i=1}^m \AInf((H_1\cap \cdots \cap H_{i-1})\cap H_i) - \AInf(H_1 \cap \cdots \cap H_{i-1}) \bigg) \\
   &= \sum_{i=1}^m \AOInf(H_1 \cap \cdots \cap H_{i-1}, H_i) - \Ainf(H_1\cap \cdots \cap H_m) \\
   &= \sum_{i=1}^m \AOInf(H_1 \cap \cdots \cap H_{i-1}, H_i) - \AInf({F^\circ}).
   \end{align*}
  Since $\AInf(F^\circ) \ge 0$, it follows that
  \begin{equation}    \label{eq:the-bound}
    \sum_{i=1}^m \EC(H_1 \cap \cdots \cap H_{i-1}, H_i) \leq \frac12 \sum_{i=1}^m \AOinf(H_1 \cap \cdots \cap H_{i-1}, H_i) \le \frac{\sqrt{2\ln m}}{\sqrt{n}},
\end{equation}
where the derivation of the second inequality is exactly as in the proof of \Cref{thm:kane-unate}.   By \Cref{eq:the-bound}, it suffices to show
    \begin{equation}    \label{eqn:claim1}
        \nu_{BB'} \leq \sum_{i=1}^m \EC(H_1 \cap \cdots \cap H_{i-1}, H_i).
    \end{equation}
    Let\ignore{$(x,y)$} {$\{x,y\}$} be a boundary-to-boundary$'$ edge and assume without loss of generality that $i^\star(x) < i^\star(y)$.  We now show that edge\ignore{$(x,y)$} {$\{x,y\}$} contributes to~$\EC(H_1 \cap \cdots \cap H_{i^\star(y)-1}, H_{i^\star(y)})$.  For brevity, write \mbox{$G = H_1 \cap \cdots \cap  H_{i^*(y)-1}$},  $H = H_{i^*(y)}$, and $C = G \setminus H = C_{i^\star(y)}$.  Since $x \in \bdry H_{i^\star(x)} = \overline{H}_{i^\star(x)} \setminus H_{i^\star(x)}$ (in particular, $x \not \in H_{i^\star(x)}$) and $i^\star(x) < i^\star(y)$ we have that $x \not \in G$.  On the other hand, $y \in G \setminus H = C$ by definition of $i^\star(y)$.  Since {$x \notin G$ and $y \in G \setminus H$}, we conclude that indeed ${\{x,y\}} \in \EC(G,H)$ as claimed.
\end{proof}

This completes the proof of \Cref{lem:thin-boundary}, and hence \Cref{thm:LO-for-polytopes}.

\subsection{A robust generalization of the Littlewood--Offord theorem for polytopes}
\label{sec:semi-thin}

In the previous section we proved \Cref{thm:LO-for-polytopes}, which establishes anticoncentration of $A\bu$ under the assumption that all its entries have magnitude at least $1$.  The goal of this section is to prove the following robust generalization of \Cref{thm:LO-for-polytopes}:

\begin{theorem}                                       \label{thm:anticonc1}
    Let $A \in \R^{m \times n}$ have the property that in every row, at least an $\alpha$ fraction of the entries have magnitude at least~$\lambda$.
    Then for any $b \in \R^m$,
%    \onote{Just noticed that I wrote $\{\pm 1\}$ rather than $\{-1,1\}$ below; we might want to go thru doc and fix these $\{\pm 1\}$'s.}
    \[
        \Pr\bracks*{A\bu \in \Game_{-2\lambda}\calO_b} \leq \frac{5 \sqrt{2\ln m}}{\alpha \sqrt{n}}.
    \]
\end{theorem}

Recall that \Cref{thm:LO-for-polytopes} followed as an easy consequence of the fact that
{$\vol(\bdry F) \leq {\frac{5 \sqrt{2\log m}}{\sqrt{n}}}$}
 when all $\bdry H_i$'s are ``thin" (\Cref{lem:thin-boundary}).  We slightly generalize this notion here.
\begin{definition}[Semi-thin]
    For $\alpha \in [0,1]$, say that $\bdry H_i$ is \emph{$\alpha$-semi-thin} if the following holds:  For each $x \in \bdry H_i$, at least an $\alpha$ fraction of its hypercube neighbors are outside $\bdry H_i$.  (Note that ``$1$-semi-thin'' is equivalent to ``thin''.) % Informally, if ``$\alpha = \Omega(1)$'' we may simply say that ``$\bdry H_i$ is semi-thin''.
\end{definition}
\begin{example}
    Suppose $H = \{ x  \in \bn: a \cdot x \leq b_1\}$ and $\overline{H} = \{x \in \bn: a \cdot x \leq b_2\}$ where $b_1 \leq b_2$, so $\bdry H = \{x \in \bn : b_1 < a \cdot x \leq b_2\}$.  If $|a_j| \geq (b_2 - b_1)/2$ for at least an $\alpha$ fraction of the coordinates $j \in [n]$, then $\bdry H$ is $\alpha$-semi-thin.
%     \rnote{I think all the $\leq$'s and $<$'s are correct now? (They seem to synch up with what happens in the ``Proof of \Cref{thm:LO-for-polytopes} assuming \Cref{lem:thin-boundary}''.)}
\end{example}

\Cref{thm:anticonc1} follows as a direct consequence of the following lemma (by the  same reasoning that derives \Cref{thm:LO-for-polytopes} as a corollary of \Cref{lem:thin-boundary}):

\begin{lemma}[Robust version of \Cref{lem:thin-boundary}]\label{lem:anticonc1}
    In the setup of \Cref{sec:proof-of-LO}, suppose each $\bdry H_i$ is $\alpha$-semi-thin.  Then
    \[
        \vol(\bdry F) \leq \frac{5 \sqrt{2\ln m}}{\alpha \sqrt{n}}.
    \]
\end{lemma}

\begin{proof}
Our proof of \Cref{lem:thin-boundary}  (a combination of \Cref{eqn:nu_BI}, \Cref{eqn:nu_BE}, and \Cref{claim:nu_BB}) shows that
\begin{equation}
 \nu_{BI} + \nu_{BE} + \nu_{BB'} \leq \frac{5\sqrt{2\ln m}}{\sqrt{n}}. \label{eq:sum-of-three}
 \end{equation}
However, in our current setting the left-hand side of the above is \emph{not} a bound on $\vol(\partial F)$; \Cref{obs:thin} no longer holds and we now may have edges $(x,y)$ where $i^{\star}(x) = i^{\star}(y)$.  Given an $x \in \partial F$ and $y$ a Hamming neighbor of $x$, we say that $y$ is \emph{$x$-bad} if $y \in \partial F$ and $i^\star(y) = i^\star(x)$; otherwise, we say that $y$ is \emph{$x$-good}.   With this terminology, we can rewrite \Cref{eq:sum-of-three} as
%\onote{double-check no factors of 2 lost? \red{Rocco: I think it's right.  If every edge satisfied the condition of~\Cref{eq:sum-of-three} the LHS of~\Cref{eq:sum-of-three} would be 1, and similarly the LHS of~\Cref{eq:one} would be 1, so I don't think there's any missing or extra factor of 2; let's erase this conversation in a next pass}}
\begin{equation} \Pr\bracks*{\bu \in \partial F \ \&\ \bu^{\oplus \bj} \text{ is $\bu$-good}} \leq \frac{5 \sqrt{2\ln m}}{\sqrt{n}}, \label{eq:one}
\end{equation}
where $\bu \sim \bn$ and $\bj \sim [n]$ are uniformly random, and $\bu^{\oplus \bj}$ denotes $\bu$ with its $\bj$-th coordinate flipped.   By the $\alpha$-semi-thin property, for any $x \in \partial F$, the fraction of $\bj$'s such that $x^{\oplus \bj}$ is $x$-good is at least $\alpha$.  Therefore
\begin{equation} \Pr\bracks*{\bu \in \partial F \ \&\ \bu^{\oplus \bj} \text{ is $\bu$-good}} \ge \Pr[\bu \in \partial F] \cdot \alpha, \label{eq:two}
\end{equation}
and the lemma follows by combining \Cref{eq:one,eq:two}.
\end{proof}

\subsection{Proof of \Cref{thm:anticonc}}
\label{sec:grand-finale}

In this section we prove \Cref{thm:anticonc} using \Cref{lem:anticonc1} established in the previous section.
%%\rnote{I reworded the following sentence
%%
%%``Specifically, we will establish a reduction from the task of bounding the anticoncentration of $A\bu$ where $A$ is $(k,\tau)$-standardized (\Cref{thm:anticonc}) to that of $A\bu$ where at least an $\alpha$ fraction of entries of each row of $A$ has magnitude at least $\tau$ (\Cref{lem:anticonc1}).''
%%
%%because it sounds weird to my ear to use such computational sounding language (reduction, task) for our situation, which is really just using one intermediate result in the proof of another result.}
In more detail, we use a bound on the anticoncentration of $A\bu$ under the assumption that at least an $\alpha$ fraction of entries of each row of $A$ have magnitude at least $\tau$ (given by~\Cref{lem:anticonc1}) to establish a bound on the anticoncentration of $A\bu$ under the assumption that each of $A$'s rows has a $\tau$-regular subvector of 2-norm 1 (\Cref{thm:anticonc}).

%We recall the following definition:
%\begin{definition}
%    For $\tau > 0$, we say a vector  $a \in \R^n$ is \emph{$\tau$-regular} if $\|a\|^4_4 \leq \tau^{2} \cdot \|a\|^2_2$ \onote{Tweaked this slightly to synch it up with HKM - hope it doesn't upset any of the calculations below...}.  In particular, if $\|a\|_\infty \leq \tau \|a\|_2$ then $a$ is $\tau$-regular.
%\end{definition}
%\lnote{I think the following was written using the HKM notion of regularity: $\|a\|^4_4 \leq \tau^{2} \cdot \|a\|^2_2$.  Let's update if necessary.} \lnote{Also, I think we need derandomized versions of them, right?}\onote{I don't see why?}
The following result regarding $\tau$-regular linear forms is fairly standard.
\begin{proposition}                                     \label{prop:good-spread}
    Let $w \in \R^n$ be a $\tau$-regular vector with $\|w\|_2 = 1$. Let $\bpi : [n] \to [B]$ be a random hash function that independently assigns each coordinate in $[n]$ to a uniformly random bucket in $[B]$.\ignore{ uniformly random partition of the coordinates $[n]$ into $B \geq 1$ buckets}  For $b \in [B]$, write  $\bsigma_b^2 = \sum_{j \in \bpi^{-1}(b)} w_j^2$, and say that bucket~$b$ is \emph{good} if $\bsigma_b^2 > \frac{1}{2B}$.  Assume  $B \leq 1/\tau^2$.  Then
    \[
        \Pr\bracks*{\text{at most $\frac{B}{16}$ buckets $b \in [B]$ are good}} \leq \exp\parens*{-\frac{B}{64}}.
    \]
\end{proposition}
\begin{proof}
    Let $\bX_b = \Ind[\bsigma_b^2 > \frac{1}{2B}]$ be the indicator that the $b$-th bucket is good.  Since $\E[\bsigma_b^2] = \frac{1}{B}$ and
    \[
        \E[\bsigma_b^4] = \E\bracks*{\parens*{\sum_{j=1}^n w_j^2 \Ind[\bpi(j) = b]}^2} = \frac{1}{B} \sum_{j=1}^n w_j^4 + \frac{1}{B^2} \sum_{j \neq j'} w_j^2 w_{j'}^2 \leq \frac{\tau^2}{B} + \frac{1}{B^2} \leq \frac{2}{B^2},
    \]
    the Paley--Zygmund inequality implies that $\E[\bX_b] = \Pr[\bsigma_b^2 > \frac12 \E[\bsigma_b^2]] \geq \frac1{8}$.

    The joint random variables $\bsigma_1^2, \dots, \bsigma_B^2$ are of ``balls in bins'' type (where the $j$-th ``ball'' has ``mass'' $w_j^2$), and are therefore negatively associated (see, e.g., \cite[Example~3.1]{DP09}; the fact that the balls have different ``masses'' does not change the argument).  Since $\Ind_{(\frac{1}{2B}, \infty)}$ is a nondecreasing function, it follows that the random variables $\bX_1, \dots, \bX_B$ are also negatively associated.  Thus we may apply the Chernoff bound to $\sum_{k=1}^B \bX_k$, which has mean at least $\frac{B}{8}$.  The result follows.
\end{proof}
Recall the following fact, which can also be easily proven using Paley--Zygmund (see e.g.~Proposition 3.7 of the full version of~\cite{GOWZ10}):
\begin{fact}                                        \label{fact:ant}
For all $w \in \R^n$ and $\bu \sim \bn$, we have $\Pr\bracks*{|w \cdot \bu| \geq \frac12\|w\|_2} \geq \frac{1}{16}$.
\end{fact}
We combine these as follows:
\begin{proposition}                                     \label{prop:mo-good-spread}
    Let $w \in \R^n$ and assume that some subvector~$w'$ of~$w$ is $\tau$-regular with $\|w'\|_2 = 1$.  Let $\bpi : [n] \to [B]$ be as in~\Cref{prop:good-spread}\ignore{a uniformly random partition of the coordinates $[n]$ into $B$ buckets}, where $B \leq 1/\tau^2$.  Let $\bu \sim \bn$, and define $\ol{\bw} \in \R^B$ by $\ol{\bw}_b = \sum_{j \in \bpi^{-1}(b)} w_j \bu_j$.  Call a bucket $b \in [B]$ \emph{big} if $|\ol{\bw}_b| > \frac{1}{2\sqrt{2B}}$.  Then
    \[
        \Pr\bracks*{\text{fewer than $\frac{B}{512}$ buckets are big}} \leq \exp\parens*{-\frac{B}{2048}}.
    \]
\end{proposition}
\begin{proof}
    First apply \Cref{prop:good-spread} to~$w'$ and observe that the presence of additional coordinates from~$w$ cannot harm ``goodness''.  Then apply \Cref{fact:ant} to the good buckets. Each becomes ``big'' independently with probability at least~$\frac{1}{16}$, and the proof follows from another Chernoff bound.
\end{proof}
We take $B = \lfloor 1/\tau^2 \rfloor$ in the above.  This yields the following:
\begin{corollary}                                       \label{cor:most-good-spread}
    Assume $A \in \R^{m\times n}$ satisfies the following property: each of its row vectors has a $\tau$-regular subvector of $2$-norm~$1$. Fix $B = \lfloor 1/\tau^2 \rfloor$ and let $\ol{\bA} \in \R^{m \times B}$ be the matrix obtained from~$A$ by randomly partitioning its columns into~$B$ buckets, and adding them up with uniformly random~$\pm 1$ signs within each bucket.  Say that a row of~$\ol{\bA}$ is \emph{spread} if at least a $\frac{1}{512}$-fraction of its entries exceed~$\frac{\tau}{2\sqrt{2}}$.  Then  except with probability at most $m \cdot \exp(-\Omega(1/\tau^2))$, all of $\ol{\bA}$'s rows are spread.
\end{corollary}

\subsubsection{Proof of \Cref{thm:anticonc}}
We can now prove \Cref{thm:anticonc}, which we restate here for convenience:

\begin{reptheorem}{thm:anticonc}
Assume $A \in \R^{m\times n}$ satisfies the following property: each of its row vectors has a $\tau$-regular subvector of $2$-norm~$1$, where $\tau$ is as set in Section~\ref{sec:params}.   Then for all $b \in \R^m$ and $\Lambda \geq \tau$,
we have
\[ \Prx_{\bu \sim \bn} \bracks*{A\bu \in \Game_{\pm\Lambda} \calO_{b}}= O\parens*{\Lambda\sqrt{\log m}}.\]
\end{reptheorem}

\begin{proof}
  By union-bounding over $2\lceil \Lambda/\tau \rceil$ choices of~$b$, it suffices to prove the following: Whenever $A \in \R^{m \times n}$ has a $\tau$-regular subvector of $2$-norm~$1$ in each row, it holds that $\Pr[A\bu \in \Game_{-\tau}\calO_b] \leq O(\tau \sqrt{\log m})$.  Note that the distribution of $A\bu$ is the same as that of $\ol{\bA}\bu'$, where $\ol{\bA}$ is as in \Cref{cor:most-good-spread}, and $\bu' \sim \bits^B$ is uniform.  Thus applying  \Cref{cor:most-good-spread} and then \Cref{thm:anticonc1} (with $\alpha = \frac{1}{512}$ and $\lambda = \frac{\tau}{2} \geq \frac{\tau}{2\sqrt{2}}$), we conclude that
\[
    \Pr\bracks*{A\bu \in \Game_{-\tau}\calO_b} = O\parens*{\tau\sqrt{\log m}} + m \cdot \exp\parens*{-\Omega(1/\tau^2)}.
\]
By our choice of $\tau$ as set in \Cref{sec:params},  we get the desired overall bound of $O\parens*{\tau\sqrt{\log m}}$ and the proof is complete.
\end{proof}
%\newpage

%!TEX root = main.tex

\section{Fooling Bentkus's mollifier}
\label{sec:fool-bentkus}

The main result of this section is the following theorem, which provides the second step of the two-step program described at the end of~\Cref{sec:mollifier}:

\begin{theorem}[$\mathscr{G}$ fools Bentkus's mollifier]
\label{thm:fool-bentkus}
Let $\mathscr{G}$ be our generator with parameters as given in \Cref{sec:params}, and likewise let $\lambda > 0$ be as set in \Cref{sec:params}.  For all $(k,\tau)$-standardized matrices $A \in \R^{m\times n}$ and all $b \in \R^m$,
\[ \Big| \Ex_{\bu \sim \bn}\big[\wt{\calO}_{b,\lambda}(A\bu) \big] - \Ex_{\bz \sim \mathscr{G}_{\mathrm{MZ}}}\big[ \wt{\calO}_{b,\lambda}(A\bz)\big] \Big| = {O(\delta)}. \]
\end{theorem}

%\rnote{We should give some high level account of the hybrid approach, no?}
At a very high level, in line with the usual Lindeberg approach, \Cref{thm:fool-bentkus} is proved by hybridizing between $\bu$ and $\bz$ via a sequence of intermediate distributions.  In our setting there are $\bucks+1$ such distributions, the first of which is $\bu$ and the last of which is $\bz$, and the $\ell$-th of which may be viewed as ``filling in buckets $\ell,\dots,\bucks$ according to $\bu$ and filling in buckets $1,\dots,\ell-1$ according to $\bz$,'' where the $\bucks$ buckets correspond to the partition of $[n]$ induced by the choice of the random hash function in the Meka--Zuckerman generator.

In \Cref{sec:singleswap} we upper bound the error incurred by taking a single step through this sequence of hybrid distributions. The upper bound given there (see \Cref{lem:single-swap}) has a first component corresponding to the terms of order $0,\dots,\tay-1$ in a {$(d-1)$-st} order Taylor expansion,
%\lnote{Minor point, but are we doing the $d$-th order or $(d-1)$st order expansion?
%\red{Rocco:}  How about we have the error term correspond to degree $d$.  Is that called $d$-th order, or $(d-1)$st order?},
and a second component corresponding to the error term in Taylor's theorem. The first component is upper bounded in \Cref{sec:singleswap}, and the second component is upper bounded in \Cref{sec:hkmreproof}.  \Cref{sec:proof-thm-fool-bentkus} formalizes the hybrid argument and uses the results of these earlier subsections to establish \Cref{thm:fool-bentkus}.

\begin{remark}[Head and tail matrices]\label{rem:H-and-T-matrices}
Recalling the definition of {a} $(k,\tau)$-standardized {matrix} $A$ (\Cref{def:regularity}), for every $i \in [m]$ there is a partition $[n] = \textsc{Head}_i \sqcup \textsc{Tail}_i$ such that $|\head_i| \le k$ and $(A_i)_{\tail_i}$ is $\tau$-regular with $2$-norm $\|(A_i)_{\tail_i}\|_2$ equal to $1$.  Therefore, we may write $A$ as $H + T$ where
\[ H_{ij} = A_{ij} \cdot \Ind[\,j \in \head_i\,] \quad \text{and}\quad  T_{ij} = A_{ij} \cdot \Ind[\,j \in \tail_i\,] \]
 for all $j\in [n]$ and $i\in [m]$.  Note that every row of $H$ is $k$-sparse, and every row of $T$ is $\tau$-regular with $2$-norm $1$.
\end{remark}

\subsection{Single swap in the hybrid argument} \label{sec:singleswap}

\begin{lemma}[Error incurred by a single swap]
\label{lem:single-swap}
Fix \ignore{$d \ge 1$ and }$B \sse [n]$.  Let $H^B, T^B \in \R^{m\times B}$, where every row of $H^B$ is $w$-sparse and every row of $T^B$ has $2$-norm at most 1. Let $\bu, \by$ be random variables over  $\bits^B$, \ignore{$\bu, \by \sim \bits^B$}where $\bu$ is uniform and $\by$ $\delta_{\CNF}$-fools the class of width-$w$ CNFs.
%\rnote{Was ``$(m2^k)$-clause CNFs''}
For all $b \in \R^m$, $\lambda >0$, and all integers $\tay \ge {2}$
%\rnote{Nitpick:  changed it to $d \geq 2$ since we'll take $c=d-1$ in the proposition where we assume $c \geq 1$. And while I'm riding this hobby-horse, should the $O$'s below be $O_d$?},
% \rnote{Note:  what's below is specialized to $s=3$, i.e. going out just to the Taylor expansion where the error term corresponds to degree 3}
%\onote{\gray{Also, I don't know where the $\log^2 m$ in the numerator above $\lambda^3$ is coming from.  Should it be $\log^{1.5} m$?} {\blue{Li-Yang:} Yes you are right, I've fixed this.  (The snapped version that HKM uses would have $\log^2 m$, but we are using the unsnapped version which has better parameters.)}}
\begin{align}
&\big| \E\big[ \wt{\calO}_{b,\lambda}(H^B\bu + T^B\bu)\big] - \E\big[\wt{\calO}_{b,\lambda}(H^B\by + T^B\by)\big]\big|  \label{eq:single-swap} \\
&=  \delta_{\CNF} \cdot m^{d-1} \cdot O_d\parens*{\frac{\sqrt{n}}{\lambda}}^{d-1} +  O_d\parens*{\frac{\sqrt{\log m}}{\lambda}}^d \left(\E\big[\|T^B\bu\|_\infty^\tay\big] + \E\big[\|T^B\by\|_\infty^\tay\big]\right).\nonumber \end{align}
\end{lemma}

As we will see later, \Cref{eq:single-swap} is a useful bound because we can (and will) take $\delta_{\CNF}$ to be very small, and when we apply Lemma~\ref{lem:single-swap} we will be able to ensure that both expectations on the right-hand side of \Cref{eq:single-swap} are small as well.

The main ingredient in the proof of~\Cref{lem:single-swap} is the following claim:

\begin{claim}
\label{claim:fixed-alpha}
For all integers $c \ge 1$ and $\alpha \in \N^m$ such that $|\alpha| = c$,
\begin{equation} \Big| \Ex\big[ \partial_\alpha \wt{\calO}_{b,\lambda}(H^B\bu)\cdot (T^B\bu)^\alpha\big] - \Ex\big[ \partial_\alpha \wt{\calO}_{b,\lambda}(H^B\by)\cdot (T^B\by)^\alpha\big]\Big|  = \delta_{\CNF}\cdot O_c\parens*{\frac{\sqrt{n}}{\lambda}}^c. \label{eq:fixed-alpha}
\end{equation}
\end{claim}

\begin{remark}
Recalling the discussion of Step~1 in \Cref{sec:new-ingredients}, we remark that~\Cref{claim:fixed-alpha} provides the key ingredient of the arguments sketched there.  This claim plays an essential role in enabling us to get a strong bound on the magnitude of the difference of two expectations (which was denoted ``$|\E[\Upsilon(\bv + \bDelta)] - \E[\Upsilon(\bv' + \bDelta')]|$'' in \Cref{sec:new-ingredients} and corresponds precisely to the left-hand side of~\Cref{lem:single-swap} above) through an application of Taylor's theorem around two different points.  As will be seen in~\Cref{sec:proofofclaimfixedalpha}, the proof of~\Cref{claim:fixed-alpha} exploits the product structure of $\wt{\calO}_b$ by using pseudorandom generators for small-width CNF formulas.
\end{remark}

Before proving~\Cref{claim:fixed-alpha}, we observe that~\Cref{lem:single-swap} follows as a consequence:

\begin{proof}[Proof of~\Cref{lem:single-swap} assuming~\Cref{claim:fixed-alpha}]
By the multidimensional Taylor expansion (\Cref{fact:taylor}) applied twice to $\wt{\calO}_{b,\lambda}$, we have
\begin{align}
\text{(\ref{eq:single-swap})}
&\le\Bigg| \sum_{0\le |\alpha| \le \tay-1} \frac1{\alpha!}\Ex\big[ \partial_\alpha \wt{\calO}_{b,\lambda}(H^B\bu)\cdot (T^B\bu)^\alpha\big] - \frac1{\alpha!}\Ex\big[ \partial_\alpha \wt{\calO}_{b,\lambda}(H^B\by)\cdot (T^B\by)^\alpha\big] \Bigg| \nonumber  \\
&\ \ \ + \E\Big[\abs*{\err(H^B\bu, T^B\bu)}\Big] + \E\Big[\abs*{\err(H^B\by, T^B\by)}\Big] \nonumber \\
&\le \sum_{0\le |\alpha|\le \tay-1} \big| \Ex\big[ \partial_\alpha \wt{\calO}_{b,\lambda}(H^B\bu)\cdot (T^B\bu)^\alpha\big] - \Ex\big[ \partial_\alpha \wt{\calO}_{b,\lambda}(H^B\by)\cdot (T^B\by)^\alpha\big]\big|  \label{eq:M^2-summands} \\
&\ \ \ +\sup_{v \in \R^m}\Bigg\{ \sum_{|\alpha|=\tay} |\partial_\alpha \wt{\calO}_{b,\lambda}(v)|\Bigg\} \cdot  \left(\E\big[\|T^B\bu\|_\infty^\tay\big] + \E\big[\|T^B\by\|_\infty^\tay\big]\right). \nonumber
\end{align}
By~\Cref{claim:fixed-alpha}, each of the $O(m^{\tay-1})$ summands of \Cref{eq:M^2-summands} is at most $\delta_\CNF \cdot O(\sqrt{n}/\lambda)^{d-1}$.  This along with the bound on $\wt{\calO}_{b,\lambda}$'s derivatives given by \Cref{thm:bentkus},
\[ \sup_{v \in \R^m}\Bigg\{ \sum_{|\alpha|=d} |\partial_\alpha \wt{\calO}_{b,\lambda}(v)|\Bigg\} = O_d\parens*{\frac{\sqrt{\log m}}{\lambda}}^d \]
yields~\Cref{lem:single-swap}.
\end{proof}
\subsubsection{Proof of~\Cref{claim:fixed-alpha}} \label{sec:proofofclaimfixedalpha}

\begin{definition}
We say that a function $\xi : \bits^B \to \R$ is \emph{Boolean} if its range is contained in $\zo$.   For $\xi_1,\ldots,\xi_m : \bits^B \to \R$, we say that the associated \emph{product function} $\Xi = \prod_{i \in [m]} \xi_i$ is a \emph{Boolean product function} in case all the $\xi_i$'s are Boolean.
\end{definition}

\begin{definition}
We say that $\xi$ is a  \emph{weight-$W$ combination of Boolean functions} if it is expressible as a linear combination $\xi = \sum_\ell c_\ell \xi_\ell$  where each $\xi_\ell$  is a Boolean function and where  $\sum_\ell |c_\ell| \le W$.  Likewise, $\Xi$ is a  \emph{weight-$W$ combination of Boolean product functions} if it is expressible as a linear combination $\Xi = \sum_\ell c_\ell \Xi_\ell$  where each $\Xi_\ell$  is a Boolean product function and where  $\sum_\ell |c_\ell| \le W$.
\end{definition}

The following facts are easy to establish:
\begin{fact}
\label{fact:combo}
\begin{enumerate}
\item A function $\xi : \bits^B \to [0,1]$ is a weight-$1$ combination of Boolean functions.
\item A function $\xi : \bits^B \to [-W,W\,]$ is a weight-$(2W)$ combination of Boolean functions.
\item A weight-$W_1$ combination of weight-$W_2$ combinations of Boolean functions is a weight-$(W_1W_2)$ combination of Boolean functions.
\item
If $\xi_1$ and $\xi_2$ are weight-$W_1$ and weight-$W_2$ combinations of Boolean product functions respectively, then $\xi_1\cdot \xi_2$ is a weight-$(W_1W_2)$ combination of Boolean product functions.
\end{enumerate}
\end{fact}

We are now ready to prove~\Cref{claim:fixed-alpha}.

\begin{proof}[Proof of~\Cref{claim:fixed-alpha}]
We define the function $G_\alpha : \bits^B \to \R$,
%\onote{Is the below accurate?  Derivative of $\wt{\calO}$ is not $\psi$ but rather some mollified half-line, no? \blue{Li-Yang: Yeah I screwed this up, sorry.  I think I fixed it by changing the definition of $\psi$. \red{Rocco:} Note that what was $\psi$ is now $\tInd$.}}
\begin{align} G_\alpha(x) &\coloneqq \partial_{\alpha}\wt{\calO}_{b,\lambda}(H^Bx) \cdot (T^Bx)^\alpha \nonumber \\
&=  \Bigg(\prod_{i \notin S}{\tInd}_{b_i,\lambda}(H^B_i x) \prod_{i \in S} {\tInd}^{(\alpha_i)}_{b_i,\lambda}(H^B_i x)\Bigg)  \cdot \prod_{i \in S} (T^B_i x)^{\alpha_i}, \label{eq:expand}
\end{align}
where $S$ denotes $\supp(\alpha) = \{ i \in [m] \colon \alpha_i > 0\}$.  (\Cref{eq:expand} crucially relies on the product structure of $\wt{\calO}_{b,\lambda} : \R^m \to (0,1)$; recall \Cref{eq:bentkus-product}.)

Note that~\Cref{claim:fixed-alpha} is equivalent to the claim that $\by$ $\delta$-fools $G_\alpha$ for $\delta = \delta_\CNF \cdot O_c(\sqrt{n}/\lambda)^c$.   We analyze the three types of functions in \Cref{eq:expand} in turn:
\begin{itemize}
\item[$\circ$] Recalling the assumptions of~\Cref{lem:single-swap}, by Item 1 of~\Cref{fact:combo}, the function $x \mapsto \tInd_{b_i,\lambda}(H^B_i x)$ is a weight-$1$ combination of Boolean functions. Furthermore, since $|\supp(H^B_i)| \le w$, it is in fact a weight-$1$ combination of Boolean $w$-juntas.
\item[$\circ$] Similarly, by Item 2 of~\Cref{fact:combo}, the function $x \mapsto{\tInd}^{(\alpha_i)}_{b_i,\lambda}(H^B_i x)$ is a weight-$(2\|{\tInd}_{b_i,\lambda}^{(\alpha_i)}\|_\infty)$ combination of Boolean $w$-juntas.
\item[$\circ$] Since $\| T^B_i \|_1  \le \sqrt{B} \cdot \| T^B_i \|_2  \leq \sqrt{B} \leq \sqrt{n}$
%\rnote{Here is where we are using the ``2-norm at most 1'' condition}
and $x_j \in \{-1,1\}$ for all $j\in B$, by Items 2 and 3 of~\Cref{fact:combo} the function $x \mapsto T^B_ix$ is a weight-$(2\sqrt{n})$ combination of Boolean functions.  Furthermore, it is a weight-$(2\sqrt{n})$ combination of Boolean $1$-juntas.
\end{itemize}
Combining the above with Item 4 of~\Cref{fact:combo}, it follows that $G_\alpha : \bits^B \to \R$ is a weight-$W$ combination of  Boolean product functions $\Xi : \bits^B \to \zo$, where
\begin{align*}
W &= \left(
\prod_{i \in S} 2\,\| {\tInd}_{b_i,\lambda}^{(\alpha_i)}\|_\infty
\right) \cdot
\left(
\prod_{i \in S} (2\sqrt{n})^{\alpha_i}
\right)
\\
&=
\left(\prod_{i \in S} O_{\alpha_i}\parens*{\frac1{\lambda^{\alpha_i}}}\right) \cdot
\left(\prod_{i \in S} (2\sqrt{n})^{\alpha_i}\right) \tag*{(\Cref{fact:psi-derivatives})} \\
&= O_c\parens*{\frac{\sqrt{n}}{\lambda}}^c.  \tag*{($|\alpha|= \alpha_1 + \cdots + \alpha_m = c$)}
\end{align*}
%\[ C = \prod_{i \in S} 2\,\| \violet{\tInd}_{b_i,\lambda}^{(\alpha_i)}\|_\infty \prod_{i \in S} (2\sqrt{n})^{\alpha_i} = O_{|\alpha|}\parens*{\parens*{\frac{\sqrt{n}}{\lambda}}^{|\alpha|}};\]
%the second equality is by~\Cref{fact:psi-derivatives} and $|\alpha| \le 2$.
Furthermore, every $\Xi$ in this combination is the product of $m$ Boolean $w$-juntas and $|\alpha|$ Boolean $1$-junta(s).  Since each such $\Xi$  is computable by a width-$w$ CNF, and $\by$ $\delta_\CNF$-fools the class of width-$w$ CNFs, we conclude that $\by$ $\delta$-fools $G_\alpha$ where $\delta = \delta_\CNF \cdot W$.  This completes the proof of~\Cref{claim:fixed-alpha}.
\end{proof}

\subsection{Bounding the error terms} \label{sec:hkmreproof}

We will use the following technical result:
\begin{claim}[Rosenthal's inequality]\label{claim:hkm1}
    Let $\beta \in [0,1]$ and let $\bx_1, \dots, \bx_n$ be independent $\{0,\pm1\}$-valued random variables, each being $0$ with probability $1-\beta$ and $\pm 1$ with probability $\beta/2$ each.  Let $w \in \R^n$ be a $\tau$-regular vector of $2$-norm~$1$.
%    \onote{Ryan says: Let $\beta = 1/L$.  For any particular bucket, if we let $\bx_i$ denote the 0/1 indicator that $i$ is in that bucket, times the $r_{\text{bucket}}$-wise uniform bits for that bucket, then the $\bx_i$'s satisfy the hypothesis of $q$-wise independence mentioned below provided $q \leq r_{\text{bucket}}, r_{\text{hash}}$.}
    Then for any $q \geq 2$,
    \[
        \E[|w \cdot \bx|^q] = O\big(q \tau \cdot (\beta/\tau^2)^{1/q} + \sqrt{q} \sqrt{\beta}\big)^q.
    \]
    Of course, if $q$ is an even integer, then the above continues to hold even if $\bx_1, \dots, \bx_n$ are merely $q$-wise independent.
%    Finally, assuming $\beta/\tau^2 \le \exp(O(q))$, we can more simply write
%    \[
%        \E[|w \cdot \bx|^q] = O(q \tau  + \sqrt{q} \sqrt{\beta})^q.
%    \]
\end{claim}
\begin{proof}
%%\rnote{This had read
%%
%%``
%%This is the Rosenthal inequality, using constants due to Nagaev and Pinelis~\cite{NP78} (see also~\cite[(4)]{PU85}).  We are using that $\sum_j \|w_j \bx_j\|_2^2 = \beta$ and $\sum_j \|w_j \bx_j\|_q^q \leq \beta \sum_j |w_j|^{q-2} \leq \beta \tau^{q-2}$.
%%''
%%
%%but that felt somewhat opaque to me since we hadn't stated the Rosenthal inequality, explained how we were using those bounds, etc.}
    This is an almost immediate consequence of a refinement of an inequality due to Rosenthal~\cite{Rosenthal70}.
    The exact version we use is due to Nagaev and Pinelis \cite{NP78} (see also~\cite[(4)]{PU85}); in our context, it states that
\begin{align*}
     \E[|w \cdot \bx|^q] &\leq 2^{O(q)} \cdot \left(q^q \sum_{j=1}^n \E[|w_j \bx_j|^q]  +
     q^{q/2}\left(\sum_{j=1}^n \E[(w_j \bx_j)^2] \right)^{q/2}\right)\\
     &\leq 2^{O(q)} \cdot \left( q^q \beta \sum_{j=1}^n |w_j|^q + (q\beta)^{q/2} \right).
\end{align*}
Since
$\beta \sum_j |w_j|^q \leq \beta \left(\sum_j w_j^2 \right) \cdot \tau^{q-2} =
\beta \tau^{q-2}$, using $x^q+y^q \leq (x+y)^q$ for positive $x,y$ we get the claimed bound.
   \end{proof}
%\begin{corollary}                                       \label{cor:hkm2}
%    For $q \geq 2$ and even integer, \Cref{lem:hkm1} continues to hold even if $\bx_1, \dots, \bx_n$ are merely $q$-wise independent.
%\end{corollary}

The following lemma will be used to bound the expectations on the right-hand side of~\Cref{eq:single-swap}:

\begin{lemma}
\label{lem:hash}
Let $L, r_\hash,r_\bucket,$ and $\tau$ be as set in \Cref{sec:params}.  Let $\bh : [n]\to [\bucks]$ be an $r_\hash$-wise uniform hash function, and fix a bucket $\ell \in [\bucks]$.   Let $\by \sim \bn$ be an $r_\bucket$-wise uniform random variable.  Let $T \in \R^{m\times n}$ be a $\tau$-regular matrix in which each row has $2$-norm 1.   Then for all integers $d \ge 2$,
\[  \Ex_{\bh,\by}\Big[ \| T^{\bh^{-1}(\ell)} \by_{\bh^{-1}(\ell)}\|^\tay_\infty\Big]  = O_d\parens*{\tau\log m + \sqrt{(\log m)/\bucks}}^d.\]
\end{lemma}

\begin{proof}
Let $q$ be the largest even integer smaller than both $r_\hash$ and $r_\bucket$; note that $q = \Theta(\log(m/\delta))$.
%\ignore{\rnote{Was ``Let $q$ denote the largest even integer smaller than both $r_\hash$ and $r_\bucket$; note that $q = \Theta(\log m)$.'' but it's not clear to me that this largest even integer smaller than $r_\hash,r_\bucket$ will be $\Theta(\log m)$ --- will these $r$-values be exactly $\Theta(\log m)$ or might they be larger? {\blue{Li-Yang:}  I \emph{think} $r_\bucket$ will be exactly $\log(m/\delta)$ (in fact, I believe this proof alone dictates the value of $r_\bucket$), and $r_\hash$ will be at least that large.}}\lnote{\green{CONSTRAINT ON PARAMS:  $r_\hash,r_\bucket \ge \log m$. \red{Rocco:}}  In fact $\geq \log(m/\delta)$, right?}}
For notational brevity we let $\bX$ denote the $\R^m$-valued random variable $\bX \coloneqq T^{\bh^{-1}(\ell)} \by_{\bh^{-1}(\ell)}$.    Since $r_\bucket,r_\hash \ge q$, we can express $\bX$ as $\sum_{j=1}^n \bx_j T^j$ where $\bx_1,\ldots,\bx_n \sim \{-1,0,1\}$ are $q$-wise independent random variables distributed as in \Cref{claim:hkm1}, with $\beta = 1/\bucks$.

Since $q > d$, we have that
    \[
               \E\bracks*{\|\bX\|_\infty^d}
        \leq \E\bracks*{\|\bX\|_q^d}
        \leq \E\bracks*{\|\bX\|_q^{q}}^{d/q}
        = \parens*{\sum_{i=1}^m \E[\bX_i^q]}^{d/q}.
    \]
Applying \Cref{claim:hkm1} to bound each $\E[\bX^q_i]$, we conclude that
\begin{align*}
     \E\bracks*{\|\bX\|_\infty^d}  &= \parens*{m\cdot O\parens*{q\tau\cdot (1/L\tau^2)^{1/q} + \sqrt{q/\bucks}}^q}^{d/q} \nonumber \\
     &= m^{d/q} \cdot O\parens*{q\tau+ \sqrt{q/\bucks}}^{d} \\
     &=  O_d\parens*{\tau\log(m/\delta) + \sqrt{\frac{\log(m/\delta)}{\bucks}}}^d, \nonumber
\end{align*}
where the second inequality uses the fact that $\parens*{\frac1{L\tau^2}}^{1/q} = \parens*{\frac{\delta}{\log m}}^{O(1/q)} = O(1)$. This completes the proof of \Cref{lem:hash}.
%%%\rnote{In response to the question ``I'm a little puzzled as to why HKM can take $r_\bucket$ to be independent of $\delta$ and we cannot'' --- I'm not completely sure but is it because of the quantitative nature of the $q$-dependence arising from the Rosenthal inequality (that's a new ingredient in our analysis which wasn't there in their analysis, right?)}
\ignore{
\gray{
Since our parameter settings satisfy $\beta/\tau^2 = 1/(\bucks\tau^2) \le \exp(O(q)) = \poly(m)$,\rnote{\green{CONSTRAINT ON PARAMS (which should hold very easily): $1/(\bucks \tau^2) \leq \poly(m)$}

\red{ANNOYANCE:}  I guess the constraint $1/(\bucks \tau^2) \leq \poly(m)$ actually
relies on $\delta$ being a non-crazy value --- if, say, $\delta$ were $1//2^{2^m}$ then because of the $\delta$-dependence in $L$ and $\tau$, $1/(\bucks \tau^2) \leq \poly(m)$ would not hold.  So maybe we need to actually take $q$ to be $\Theta(\log(m/\delta))$ or something, instead of $\Theta(\log m)$? . \blue{Li-Yang:} Okay fixed this, though I'm a little puzzled as to why HKM can take $r_\bucket$ to be independent of $\delta$ and we cannot.
} by~\Cref{claim:hkm1}, we can bound each $\E[\bX_i^q]$ by $O\big(q \tau + \sqrt{q/\bucks}\big)^q$.   Thus we have shown that
   \begin{align*}
     \E\bracks*{\|\bX\|_\infty^d}  &\le \parens*{m\cdot O\big(q\tau + \sqrt{q/\bucks}\big)^q}^{d/q} \\
     &=  O_d\parens*{\tau\log m + \sqrt{(\log m)/\bucks}}^d,  \tag*{(\green{$q = \Theta(\log m)$})}
          \end{align*}
and this completes the proof of \Cref{lem:hash}}
}
\end{proof}

%
%\begin{corollary}                                       \label{cor:hkm2}
%     Fix an integer $d \geq 1$.\onote{This will truly be an absolute constant like $1000$, not depending on $m$ or anything else, right?}  Let $T \in \R^{m \times n}$ be a matrix in which each row is $\tau$-regular with $2$-norm~$1$.   Let $\bx_1, \dots, \bx_n$ be $q$-wise independent with distribution as in \Cref{claim:hkm1}, where $q \geq \max(\log m, d)$ is an even integer.  Assume  also $\beta/\tau^2 \leq \poly(m)$.  Let $\bX = \sum_{j=1}^n \bx_j T^j$, where $T^j$ is the $j$-th column of~$T$.  Then
%    \begin{equation}    \label{eqn:hkm2}
%        \E[\|\bX\|_\infty^d] \leq O\parens*{\max\braces*{(\log m) \tau, \sqrt{\log m}\sqrt{\beta} }}^d.
%    \end{equation}
%\end{corollary}
%\begin{proof}
%    There is no harm in assuming $q = \Theta(\log m)$.  Then
%    \[
%               \E\bracks*{\|\bX\|_\infty^d}
%        \leq \E\bracks*{\|\bX\|_q^d}
%        \leq \E\bracks*{\|\bX\|_q^{q}}^{d/q}
%        = \E\bracks*{\sum_{i=1}^m \bX_i^q}^{d/q}.
%    \]
%    By \Cref{lem:hkm1}, we can bound each $\E[\bX_i^q]$ by $O(q \tau + \sqrt{q}\sqrt{\beta})^q$.  Thus our final bound is
%    \[
%        m^{d/q} \cdot O(q \tau + \sqrt{q}\sqrt{\beta})^d,
%    \]
%    which is bounded as in \Cref{eqn:hkm2} using $q = \Theta(\log m)$ and $d=O(1).$
%\end{proof}

\subsection{Proof of~\Cref{thm:fool-bentkus}: the hybrid argument} \label{sec:proof-thm-fool-bentkus}

In this subsection we put together the two main results of the two previous subsections (\Cref{lem:single-swap} and \Cref{lem:hash}) to prove \Cref{thm:fool-bentkus}.

Recalling \Cref{rem:H-and-T-matrices}, we can write $A$ as $H+T$, where every row of $H$ is $k$-sparse and every row of $T$ is $\tau$-regular with $2$-norm $1$.  Let us say that a hash $h : [n] \to [\bucks]$ is \emph{$H$-good} if
\begin{equation} | h^{-1}(\ell) \cap \supp(H_i)| \le w \coloneqq {\frac{2k}{\bucks}} \label{eq:def-of-w}
\ignore{\rnote{Was ``$\frac{k}{\bucks} + O\big(\sqrt{\log(\delta/\bucks m)}\big)$'', but that's $\sqrt{\text{negative number}}$, and it seems to me we kind of need $2k/\bucks$ rather than $k/\bucks+$(something smaller than $k/\bucks$) for the tail bound proposition~\Cref{prop:hash} to go through smoothly.  Note that the factor of 2 doesn't cost us anything meaningful since we would be taking $w \geq k/\bucks$ anyway.}}
%\rnote{\green{CONSTRAINT ON PARAMS:  $w=2k/\bucks=\Theta(\log(m/\delta)\log\log(m/\delta))$}}
\end{equation}
for all buckets $\ell \in [\bucks]$ and rows $i\in [m]$.    Equivalently, for all $\ell \in [\bucks]$, every row of the the submatrix $H^{h^{-1}(\ell)}$ is $w$-sparse.

\begin{proposition}[Even distribution of head variables]
\label{prop:hash}
There is a universal constant $C_1 > 0$ such that the following holds.  If $\bh : [n] \to [\bucks]$ is $r_\hash$-wise uniform where $r_\hash \ge C_1 \log(Lm/\delta)$, then
\[ \Pr\bracks*{\text{$\bh$ is not $H$-good}} \le \delta. \]
\end{proposition}

\begin{proof}
%%\rnote{This proof is really just \cite{SSS95}, which upper bounds $\Pr[|\bX - \mu| > \eps \mu]$ where $\bX$ is $r$-wise independent and $\mu=\E[\bX]$.  There are literally five different bounds depending on the relationship between $\eps,r$ and $\mu$; see
%%\href{https://www.cs.mcgill.ca/~amehra13/Articles/kwise_independent_concentration_summary.pdf}{Theorem~3 of this link}.
%%(Note that their ``$k$'' there is our $r_{\hash}$ and their ``$\mu$'' is our $k/\bucks$.)
%%
%%I think the right regime for us is the first  of their five bounds.  We are taking $\eps=1$ since our bound is $2\mu=2k/\bucks.$  Their ``$k$'', which is our $r_{\hash}$, will be set to be $\log(\bucks m/\delta)$, and this is less than their $\eps^2 \mu e^{-1/3}$, which is (now back to our parameters) $e^{-1/3} k/\bucks = \Theta(\log(m/\delta)\log\log(m/\delta))$.
%%
%%If this looks right, feel free to delete this rnote
%%}
Fix any $\ell \in [\bucks]$ and $i \in [m].$
The quantity $|\bh^{-1}(\ell) \cap \supp(H_i)|$ is a sum of $|\supp(H_i)| \leq k$ many $r_{\hash}$-wise independent $\zo$-valued random variables, each of which takes the value 1 with probability $1/\bucks.$  To bound the probability that $|\bh^{-1}(\ell) \cap \supp(H_i)|$ is larger than $w$, we apply the well-known tail bounds for sums of limited-independence random variables due to Schmidt, Siegel, and Srinivasan~\cite{SSS95}, specifically their Theorem~5(I)(a).  Taking the ``$\delta$'' of their paper to be 1 and observing that their ``$\mu$'' is our $k/\bucks$ and their ``$k$'' is our $r_{\hash}= \Theta(\log(Lm/\delta))$,
%\rnote{\green{CONSTRAINT ON PARAMS:  $r_{\hash} = \Theta(\log(Lm/\delta))$}},
we get that
$\Pr[|\bh^{-1}(\ell) \cap \supp(H_i)| > w] \leq \delta/(\bucks m).$
The proposition follows by a union bound over all $\ell \in [\bucks]$ and $i \in [m].$
\end{proof}

We are now ready to prove \Cref{thm:fool-bentkus}, which we restate here for convenience:

\begin{reptheorem}{thm:fool-bentkus}
Let $\mathscr{G}$ be our generator with parameters as given in \Cref{sec:params}, and likewise {let} $\lambda > 0$ be as set in \Cref{sec:params}.  For all $(k,\tau)$-standardized matrices $A \in \R^{m\times n}$ and {all} $b \in \R^m$,
\[ \Big| \Ex_{\bu \sim \bn}\big[\wt{\calO}_{b,\lambda}(A\bu) \big] - \Ex_{\bz \sim \mathscr{G}_{\mathrm{MZ}}}\big[ \wt{\calO}_{b,\lambda}(A\bz)\big] \Big| = {O(\delta)}. \]
\end{reptheorem}

\begin{proof}
Let $\bh, \by^1,\ldots,\by^\bucks,  \tilde{\by}^1,\ldots,\tilde{\by}^\bucks$, $\breve{\by}$, and $\by^\star$ be the random hash function and random variables associated with our generator $\mathscr{G}$, as defined in \Cref{def:ourgenerator}.  Recall that a draw from $\bz \sim \mathscr{G}$ is $\bz \coloneqq \breve{\by}\oplus\by^\star$.  We will show that in fact $\breve{\by}$ alone satisfies:
 \begin{equation} \Big| \Ex_{\bu \sim \bn}\big[\wt{\calO}_{b,\lambda}(A\bu) \big] - \Ex\big[ \wt{\calO}_{b,\lambda}(A\breve{\by})\big] \Big| = {O(\delta)}. \label{eq:without-xor}
 \end{equation}
 Since $\by^{\star}$ and $\breve{\by}$ are independent, \Cref{thm:fool-bentkus} follows as a consequence of \Cref{eq:without-xor}.

We recall the definition of $\breve{\by}$:
\[  \breve{\by}_{\bh^{-1}(\ell)} = (\by^{\ell} \oplus \tilde{\by}^\ell)_{\bh^{-1}(\ell)}   \qquad \text{for all $\ell \in [\bucks]$}.\]
We observe first that for each $\ell \in [\bucks]$, the random variable $\by^\ell \oplus \tilde{\by}^{\ell} \sim \bn$
\begin{enumerate}
\itemsep -.5pt
\item[(i)] is $r_\bucket$-wise uniform (since $\by^{\ell}$ is); and
\item[(ii)] $\delta_\CNF$-fools the class of width-$w$ CNF formulas (since $\tilde{\by}^{\ell}$ does).
\end{enumerate}
We will use both properties in this proof.  For each hash $h : [n] \to [\bucks]$ and index $\ell \in \{0,1,\ldots,\bucks\}$, we define the hybrid random variable $\bx^{h,\ell} \sim \bn$,
\[
\bx^{h,\ell}_{h^{-1}(c)} =
\begin{cases}
\bu_{h^{-1}(c)} & \text{if $c > \ell$} \\
(\by^{\ell} \oplus \tilde{\by}^\ell)_{h^{-1}(c)} & \text{if $c \le \ell$}.
\end{cases}
\]
Averaging over $\bh$, we get that $\bx^{\bh,0} \equiv \bu$ and $\bx^{\bh,\bucks} \equiv \breve{\by}$, and so we may write
\begin{align*}
\text{LHS of (\ref{eq:without-xor})} &= \big| \Ex\big[\wt{\calO}_{b,\lambda}(A\bu) \big] - \Ex\big[ \wt{\calO}_{b,\lambda}(A\breve{\by})\big] \big|  \\
&=
\big| \Ex\big[\wt{\calO}_{b,\lambda}(A\bx^{\bh,0}) \big] - \Ex\big[ \wt{\calO}_{b,\lambda}(A\bx^{\bh,\bucks})\big] \big| \\
&\le \Ex_{\bh} \Big[ \big| \Ex\big[\wt{\calO}_{b,\lambda}(A\bx^{\bh,0}) \big] - \Ex\big[ \wt{\calO}_{b,\lambda}(A\bx^{\bh,\bucks})\big] \big|  \Big]  \\
&\le \Ex_{\bh} \Big[ \big| \Ex\big[\wt{\calO}_{b,\lambda}(A\bx^{\bh,0}) \big] - \Ex\big[ \wt{\calO}_{b,\lambda}(A\bx^{\bh,\bucks})\big] \big| \cdot \Ind\big[ \, \text{$\bh$ is $H$-good}\,\big] \Big] + \Pr\big[\,\text{$\bh$ is not $H$-good}\,\big] \\
&\le \underbrace{\Ex_{\bh} \Big[ \big| \Ex\big[\wt{\calO}_{b,\lambda}(A\bx^{\bh,0}) \big] - \Ex\big[ \wt{\calO}_{b,\lambda}(A\bx^{\bh,\bucks})\big] \big| \cdot \Ind\big[ \, \text{$\bh$ is $H$-good}\,\big] \Big]}_{\heartsuit} + \delta.
\end{align*}
The penultimate inequality uses the fact that $\wt{\calO}_{b,\lambda}$ is $(0,1)$-valued (and hence the difference in its expectations under any two distributions is at most $1$), and the final inequality is by \Cref{prop:hash} (note that we indeed have {$r_\hash \ge {C_1 \log(Lm/\delta)}$}).

It remains to bound $\heartsuit$ by $O(\delta)$.   Fix a $H$-good hash $h$.  By the triangle inequality,
\begin{equation}
\big| \Ex\big[\wt{\calO}_{b,\lambda}(A\bx^{h,0}) \big] - \Ex\big[ \wt{\calO}_{b,\lambda}(A\bx^{h,\bucks})\big] \big| \le \sum_{\ell=1}^\bucks  \big| \Ex\big[\wt{\calO}_{b,\lambda}(A\bx^{h,\ell-1}) \big] - \E\big[ \wt{\calO}_{b,\lambda}(A\bx^{h,\ell})\big] \big|.\label{eq:triangle}
\end{equation}
Fix $\ell\in [\bucks]$ and consider the corresponding summand
\begin{equation}
 \big| \Ex\big[\wt{\calO}_{b,\lambda}(A\bx^{h,\ell-1}) \big] - \E\big[ \wt{\calO}_{b,\lambda}(A\bx^{h,\ell})\big] \big|. \label{eq:fixed-hash}
 \end{equation}
For notational clarity, let us write $B$ for $h^{-1}(\ell)$ and $\overline{B}$ to denote $[n]\setminus B.$  Furthermore, since these ``adjacent" hybrid random variables $\bx^{h,\ell-1}$ and $\bx^{h,\ell}$ agree on all coordinates outside $B$, we introduce the random variable $\bs \sim \bits^{\overline{B}}$ where $\bs_{h^{-1}(c)} \equiv \bx^{h,\ell-1}_{h^{-1}(c)} \equiv \bx^{h,\ell}_{h^{-1}(c)}$ for all $c \ne \ell$. Note that $\bs, \bu_B$, and $(\by^{\ell} \oplus \tilde{\by}^\ell)_B$ are mutually independent.  We have that
\begin{align*}
 \eqref{eq:fixed-hash} &= \Big| \Ex_{\bs} \Big[ \Ex_{\bu} \big[ \wt{\calO}_{b,\lambda}(A^{\overline{B}}\bs + A^B \bu_{B}) \big] - \Ex_{\by^\ell,\tilde{\by}^{\ell}}\big[  \wt{\calO}_{b,\lambda}(A^{\overline{B}}\bs + A^B (\by^{\ell} \oplus \tilde{\by}^\ell)_{B}) \big]  \Big] \Big| \\
 &\le \Ex_{\bs} \Big[ \big|\Ex_{\bu} \big[ \wt{\calO}_{b,\lambda}(A^{\overline{B}}\bs + A^B \bu_{B}) \big] - \Ex_{\by^\ell,\tilde{\by}^{\ell}}\big[  \wt{\calO}_{b,\lambda}(A^{\overline{B}}\bs + A^B (\by^{\ell} \oplus \tilde{\by}^\ell)_{B}) \big] \big| \Big] \\
 &= \Ex_{\bs} \Big[ \big|\Ex_{\bu} \big[ \wt{\calO}_{b-A^{\overline{B}}\bs,\lambda}(A^B \bu_{B}) \big] - \Ex_{\by^\ell,\tilde{\by}^{\ell}}\big[  \wt{\calO}_{b-A^{\overline{B}}\bs,\lambda}(A^B (\by^{\ell} \oplus \tilde{\by}^\ell)_{B}) \big] \big| \Big] \tag*{(\Cref{fact:shift})} \\
 &= \Ex_{\bs} \Big[ \big|\Ex_{\bu} \big[ \wt{\calO}_{b-A^{\overline{B}}\bs,\lambda}(H^B \bu_{B} + T^B\bu_{B}) \big] - \Ex_{\by^\ell,\tilde{\by}^{\ell}}\big[  \wt{\calO}_{b-A^{\overline{B}}\bs,\lambda}(H^B (\by^{\ell} \oplus \tilde{\by}^\ell)_{B}+ T^B (\by^{\ell} \oplus \tilde{\by}^\ell)_{B}) \big] \big| \Big].
 \end{align*}
Since $h$ is $H$-good, every row of $H^B$ is indeed $w$-sparse, and since every row of $T$ has 2-norm 1, every row of $T^B$ has 2-norm at most 1. Recalling (ii) from above, we may
apply \Cref{lem:single-swap} to each outcome $s$ of $\bs$, and we get that this quantity is at most
 \[
 \delta_{\CNF} \cdot m^{d-1} \cdot  O\parens*{\frac{\sqrt{n}}{\lambda}}^{d-1} +  O\parens*{\frac{\sqrt{\log m}}{\lambda}}^d\parens*{\Ex_{\bu}\big[\|T^B\bu_B\|_\infty^\tay\big] + \Ex_{\by^\ell,\tilde{\by}^{\ell}}\big[\|T^B(\by^{\ell} \oplus \tilde{\by}^\ell)_B\|_\infty^\tay\big]},
  \]
and therefore
\begin{align}
 \text{RHS of \eqref{eq:triangle}} &\le  \bucks\cdot  \delta_{\CNF} \cdot m^{d-1} \cdot O\parens*{\frac{\sqrt{n}}{\lambda}}^{d-1}   \nonumber \\
& \ \  \ \ +  O\parens*{\frac{\sqrt{\log m}}{\lambda}}^d\cdot  \sum_{\ell=1}^\bucks\parens*{\Ex_{\bu}\Big[\|T^{h^{-1}(\ell)}\bu_{h^{-1}(\ell)}\|_\infty^\tay\Big] +\Ex_{\by^\ell,\tilde{\by}^{\ell}}\Big[\|T^{h^{-1}(\ell)}(\by^{\ell} \oplus \tilde{\by}^\ell)_{h^{-1}(\ell)}\|_\infty^\tay\Big]}.\label{eq:almost}
\end{align}
Since \Cref{eq:almost} holds for every $H$-good hash $h$, we have shown that
\begin{align*}
\heartsuit &\le \Ex_{\bh}\big[ \text{(RHS of (\ref{eq:triangle}))} \cdot \Ind[ \,\text{$\bh$ is $H$-good}\,]  \big]  \\
&\le  \Ex_{\bh}\big[ \text{(RHS of (\ref{eq:almost}))} \cdot \Ind[ \,\text{$\bh$ is $H$-good}\,]  \big] \\
&\le \Ex_{\bh}\big[ \text{(RHS of (\ref{eq:almost}))}  \big]  \\
&= \bucks\cdot  \delta_{\CNF} \cdot m^{d-1} \cdot O\parens*{\frac{\sqrt{n}}{\lambda}}^{d-1}   \nonumber \\
& \ \  \ \ +  O\parens*{\frac{\sqrt{\log m}}{\lambda}}^d\cdot \underbrace{\sum_{\ell=1}^\bucks\left(\Ex_{\bh,\bu}\Big[\|T^{\bh^{-1}(\ell)}\bu_{\bh^{-1}(\ell)}\|_\infty^\tay\Big] + \Ex_{\bh,\by^\ell,\tilde{\by}^{\ell}}\Big[\|T^{\bh^{-1}(\ell)}(\by^{\ell} \oplus \tilde{\by}^\ell)_{\bh^{-1}(\ell)}\|_\infty^\tay\Big]\right)}_{\diamondsuit}.
\end{align*}
Applying \Cref{lem:hash} to bound each of the $2\bucks$ many summands of $\diamondsuit$, we have that
\begin{align}
\heartsuit &\le \bucks\cdot  \delta_{\CNF} \cdot m^{d-1} \cdot O\parens*{\frac{\sqrt{n}}{\lambda}}^{d-1} +  O\parens*{\frac{\sqrt{\log m}}{\lambda}}^d\cdot  2\bucks \cdot O\parens*{\tau\log(m/\delta) + \sqrt{\frac{\log(m/\delta)}{\bucks}}}^d \nonumber \\
&= \bucks\cdot  \delta_{\CNF} \cdot m^{d-1} \cdot O\parens*{\frac{\sqrt{n}}{\lambda}}^{d-1} + L \cdot O\parens*{\frac{\tau\sqrt{\log(m)} \log(m/\delta)}{\lambda} + \frac{\sqrt{\log(m)\log(m/\delta)}}{\lambda \sqrt{\bucks}}}^d. \label{eq:params}
 \end{align}
By our choice of parameters as set in \Cref{sec:params},
\begin{align*}
(\ref{eq:params}) &=   O(\delta)  + \frac{(\log m)^5}{\delta^{2+\eps}} \cdot O\parens*{\delta^\eps \cdot \frac{\log(m) (\log(m/\delta))^{1.5 + \eps}}{(\log m)^{2.5+\eps}} + \delta^{\eps/2} \cdot \frac{\log(m)\log(m/\delta)}{(\log m)^{2.5}}}^d.
\end{align*}
Taking $d$ to be sufficiently large relative to $\eps$,\ignore{\rnote{Nitpick:  Even if $C_3$ were large relative to $d$, since $\eps$ is an absolute constant, wouldn't we still get $O(\delta)$ out (where now there'd be a large $\eps$-dependent constant hidden by the big-Oh, but whatever)? I guess the question is or not whether we think of our big-Oh bounding the seed length as being able to hide (large) constant factors that depend on $\eps.$ \blue{Li-Yang: Hmm, good point.  Actually, I guess our $O(\delta)$ will have to hide constant factors that depend on $d$ anyway because of the second summand. I removed the $C_3$.  I also removed all the $d$ subscripts in this section, since this is in the context of the parameter settings in \Cref{sec:params}, where $d$ is an absolute constant. }}} the above expression can be bounded by $O(\delta)$.  This establishes \Cref{eq:without-xor}, and the proof of~\Cref{thm:fool-bentkus} is complete.
\end{proof}

%%%%\ignore{
%%%%\rnote{\green{CONSTRAINT ON PARAMS:  The first expression forces
%%%%\[
%%%%\delta_{\CNF} \leq {\frac {\delta}{\bucks \cdot O(m\sqrt{n}/\lambda)^{d-1}}}.
%%%%\]
%%%%To analyze the second expression, it's helpful to recall that $L \geq 1/\tau^2$ so $\tau \geq 1/\sqrt{L}$; this means that the ``$\tau \log m$'' part dominates the ``$\sqrt{(\log m)/L}$'' part, so we can view the ``$O($stuff$)^d$'' in the second expression as being just $O(\tau \log m)^d$.  And thus the second expression forces
%%%%\[
%%%%L \cdot O\left(
%%%%{\frac {\tau (\log m)^{3/2}}{\lambda}}
%%%%\right)^d \leq \delta
%%%%\]
%%%%which results in the constraint
%%%%\[
%%%%\overbrace{\tau \leq {\frac {\lambda (\delta/L)^{1/d}}{(\log m)^{3/2}}}}^{\text{call this constraint $A$}}.
%%%%\]
%%%%}}
%%%%}

%!TEX root = main.tex

\section{Proof of \Cref{thm:fool-k-tau-regular}}
\label{sec:put-together}

Having completed both steps of the two-step program described at the end of~\Cref{sec:mollifier} we are finally ready to prove \Cref{thm:fool-k-tau-regular}, which we restate here for convenience:

\begin{reptheorem}{thm:fool-k-tau-regular}
Let $\mathscr{G}$ be our generator with parameters as set in \Cref{sec:params}.   For all $m$-facet $(k,\tau)$-standardized polytopes $Ax \leq b$,
\[ \Big| \Prx_{\bu \sim \bn}\big[A\bu \in \calO_{b} \big] - \Prx_{\bz \sim \mathscr{G}}\big[ A\bz \in \calO_{b}\big] \Big| = O(\delta). \]
\end{reptheorem}

\begin{proof}
Let $\lambda \in (0,1)$ be as set in \Cref{sec:params}.  By \Cref{lemma:bentkus-sandwich},
% \violet{with $\eps = \delta$}\rnote{If we're going to take $\eps=\delta$, guess we may as well just switch all the $\eps$'s in Section~6 to $\delta$'s?  I'll do this in a later pass}, 
 there are $b^{\inner},b^{\outter} \in \R^m$ such that $\wt{\calO}_{b^{\inner},\lambda},\wt{\calO}_{b^{\outter},\lambda}$ are $(\Lambda,\delta)$-inner and -outer approximators for $\calO_b$ respectively,
where $\Lambda = \Theta(\lambda \sqrt{\log(m/\delta)}).$  Next, we apply \Cref{lem:soft-to-hard} with $\bv$ and $\tilde{\bv}$ being $A\bu$ and $A \bz$ respectively, using \Cref{thm:fool-bentkus} to show that \Cref{eq:fool-Upsilon} is satisfied for both $\wt{\calO}_{b^{\inner},\lambda}$ and $\wt{\calO}_{b^{\outter},\lambda}$ with $\gamma = O(\delta)$.  We conclude that:
\begin{align}
&\abs*{\Prx_{\bu \sim \bn}\bracks*{A\bu \in \calO_{b}} - \Prx_{\bz \sim \mathscr{G}}\bracks*{A\bz \in \calO_{b}} }\nonumber \\
&= O(\delta) + \Pr\bracks*{A\bu \in \Game_{\pm\Lambda}\calO_b}  \tag*{(\Cref{lem:soft-to-hard} and \Cref{thm:fool-bentkus})}\\
&= O(\delta) + O\big(\Lambda\sqrt{\log m})\big) \tag*{(\Cref{thm:anticonc}; note that $\Lambda \ge \tau$ is indeed satisfied)}\\
&= O(\delta) + O\big(\lambda\sqrt{\log(m/\delta)\log m}\big) \nonumber \\
&= O(\delta). \label{eq:our-lambda} 
\end{align}
This completes the proof of \Cref{thm:fool-k-tau-regular}.
\end{proof}

\section*{Acknowledgements}
 
 R.O.~is supported by NSF grants CCF-1618679 and CCF-1717606. R.S. is supported by NSF grant CCF-1563155.  L.-Y.T. is supported by NSF grant CCF-1563122; part of this work was performed while the author was at TTI-Chicago.

This material is based upon work supported by the National Science Foundation under grant numbers listed above. Any opinions, findings and conclusions or recommendations expressed in this material are those of the authors and do not necessarily reflect the views of the National Science Foundation.

\bibliography{allrefs}{}
\bibliographystyle{alpha}

\appendix

%!TEX root = main.tex

\section{Proof of \Cref{claim:optimal}} \label{ap:optimal}
%\lnote{Small typographical changes to standardize with rest of paper: $f \to F$, $s \to \sigma$}
We recall~\Cref{claim:optimal}:

\begin{claim}[\Cref{claim:optimal} restated]
%\onote{In what sense is it slightly sharpened?  Looks identical to me\dots}
For $\violet{2} \leq m \leq 2^n$, there is a matrix $A \in \{-1,1\}^{m \times n}$ and a vector $b \in \Z^m$ such that
\[
\Pr\bracks*{ A\bu \in \Game \calO_b}  = \Omega\parens*{\frac{\sqrt{\ln m}}{\sqrt{n}}}.
\]
\end{claim}

\begin{proof}

The proof is a simple probabilistic existence argument that follows the approach used to prove Theorem~2 in~\cite{Kane14intersection}. For a polytope ${\cal K} = \Ind[ Ax \leq b]$ we define % $\Inside({\cal K}), \Surface({\cal K})$ and $\Outside({\cal K})$ to denote the set of those points $x \in \{-1,1\}^n$ that lie inside, on the surface of, or outside of ${\cal K}$ respectively, i.e.,
\begin{align*}
\Inside ({\cal K}) &= \{x \in \bits^n: Ax \in \calO_b \setminus \Game \calO_b, &\text{i.e.~}A_i x < b_i \text{~for all~}i \in [m]\},\\
\Surface ({\cal K}) &= \{x \in \bits^n: Ax \in  \Game \calO_b, &\text{i.e.~}Ax \leq b \text{~and~}A_i x = b_i \text{~for some~}i \in [m]\}.
%\Outside ({\cal K}) &= \{x \in \bits^n: Ax \notin \calO_b, &\text{i.e.~}A_i x > b_i \text{~for some~}i \in [m]\}.%
\end{align*}

Given $\violet{2} \leq m \leq 2^n$, if $m < 10$ then the one-facet polytope $\Ind[x_1 + \cdots + x_n \leq 0]$ does the job (more formally, we take $A$ to be the $m \times n$ all-1's matrix and $b$ to be the zero vector in $\R^m$).  It is also clear that proving our result for $m \leq 2^{n/10}$ also proves it for $2^{n/10} \leq m \leq 2^n$. So we henceforth assume that $10 \leq m \leq 2^{n/10}$.
 Let $k \geq n/2$ be an integer to be chosen later, and let $F: \{-1,1\}^n \to \zo$ denote the halfspace $F(x)= \Ind[x_1 + \cdots + x_n \leq t]$. Now define the following quantities:
\begin{align*}
p_{\mathrm{I}} &\coloneqq |\Inside(F)|/2^n = {n \choose < k}/2^n,\qquad 
p_{\mathrm{S}} \coloneqq |\Surface(F)|/2^n  = {n \choose k}/2^n.
%p_{\mathrm{O}} &\coloneqq |\Outside(F)|/2^n= {n \choose > k}/2^n.
\end{align*}
%\onote{Deleted some unused notation here}
 %For $\sigma \in \{-1,1\}^n$, let $F_\sigma(x)=F(\sigma_1 x_1,\dots,
%\sigma_n x_n)$.  Each $F_\sigma$ is clearly a halfspace with %We note that for each $\sigma \in \{-1,1\}^n$ the function $F_\sigma$ is a halfspace which, after negating variables, is identical to $F$ (and in particular has
%the same number of inside, outside, and surface points as $F$.%).

Let $\overline{\bsigma}=(\bsigma^1,\dots,\bsigma^{m})$ where each $\bsigma^i$ is an independent uniform string in $\bits^n$.  Let $\bA \in \bits^{m \times n}$ be the matrix whose $i$-th row is $\bsigma^i$, and let $b$ be the vector $(k,\dots,k) \in \Z^{m}.$ It is easy to see that in order to prove our result it suffices to show that there is a fixed outcome $A$ of $\bA$ such that
\begin{equation} \label{eq:optimal-goal}
\Pr\bracks*{ A\bu \in \Game \calO_{b}} = \Omega\parens*{\frac{\sqrt{\log m}}{\sqrt{n}}},
\end{equation}
and this is what we show below.  Towards this end, for each $i \in [m]$ let us define the matrix $\bA^{\setminus i} \in \bits^{(m-1) \times n}$ obtained by removing the $i$-th row of $\bA$, and further define $b' = (k,\dots,k) \in \Z^{m-1}.$

For each fixed $z \in \{-1,1\}^n$ and each $i \in [m]$ we have
\[
\Prx_{\overline{\bsigma}} \left[ z \in \Inside(\Ind[\bA^{\setminus i} x \leq b']) \right]= p_{\mathrm{I}}^{m-1}
\]
and
\[
\Prx_{\overline{\bsigma}} \left[ z \in \Surface(\Ind[\bsigma^i \cdot x \leq k])\right]= p_{\mathrm{S}}.
\]
Since $\bsigma^i$ and $(\bsigma^{i'})_{i' \in [m] \setminus \{i\}}$ are independent for each $i \in [m]$, it follows that
\[
\Prx_{\overline{\bsigma}}
\left[
z \in \Inside(\Ind[\bA^{\setminus i} x \leq b']) \, \, \& \, \,
z \in \Surface(\Ind[\bsigma^i \cdot x \leq k])
\right
]= p_{\mathrm{S}} \cdot p_\mathrm{I}^{m-1},
\]
and since the events
\[
z \in \Inside(\Ind[\bA^{\setminus i} x \leq b']) \, \, \& \, \,
z \in \Surface(\Ind[\bsigma^i \cdot x \leq k])
\]
and
\[
z \in \Inside(\Ind[\bA^{\setminus i'} x \leq b']) \, \, \& \, \,
z \in \Surface(\Ind[\bsigma^{i'} \cdot x \leq k])
\]
are mutually exclusive for $i \neq i' \in [m]$, by a union bound we have that
\[
\Prx_{\overline{\bsigma}} \big[ \, z \in \Surface(\Ind[\bA x \leq b']) \big] = \Prx_{\overline{\bsigma}}\bracks*{ \bA z \in \Game \calO_{b'}}
= m \cdot p_{\mathrm{S}} \cdot p_{\mathrm{I}}^{m-1}.
\]
It follows that there is an outcome of $\overline{\bsigma}$ such that the resulting matrix $A \in \R^{m \times n}$ has at least an $m \cdot p_{\mathrm{S}} \cdot p_{\mathrm{I}}^{m-1}$ fraction of all  points in $\{-1,1\}^n$ satisfying $A z \in \Game \calO_{\violet{b}}$; i.e.,
\begin{equation} \label{eq:goal}
\Pr\bracks*{ A\bu \in \Game \calO_b} \geq m \cdot p_{\mathrm{S}} \cdot p_{\mathrm{I}}^{m-1}.
\end{equation}

It remains only to argue that for any $10 \leq m \leq 2^{n/10}$, there is a value $k$ such that, for $p_{\mathrm{I}}={n \choose <k}$ and $p_{\mathrm{S}} = {n \choose k}$, we have
\[
m \cdot p_{\mathrm{S}} \cdot p_{\mathrm{I}}^{m-1} = \Omega(\sqrt{\log m}/\sqrt{n}).
\]
Towards this end we choose $k$ to be the largest integer such that ${n \choose < k} / 2^n \leq 1-{\frac 1 m}$.  Recalling that $10 \leq m \leq 2^{n/10}$, we have that $n/2 \leq k \leq 0.99n$, and hence ${n \choose k}$ and ${n \choose k-1}$ are within an absolute constant multiplicative factor of each other.  It follows that

\[
p_{\mathrm{I}} = {n \choose < k}/2^n = 1-\Theta(1/m),
\]
\ignore{\rnote{Note that if $m$ were, say, $2^{n-1}$ or $2^n/(n^{1.5})$, this would not be true. If $m=2^{n-1}$ then $1-1/m = 1-2/2^n,$ and the value of $k$ would be $n-1$ and ${n \choose  < n-1} = 1 - (n+1)/2^n,$ so it's false since $(n+1)/2^n$ is not $\Theta(1/m).$   Similarly, if $m=2^n/(n^{1.5})$ then the value of $k$ would be $n-2$ and ${n \choose < n-2} = 1 - ({n \choose 2} + n + 1)/2^n$ so it's false since $({n \choose 2} + n + 1)/2^n$ is not $\Theta(1/m)$.}} which implies that
\[
p_{\mathrm{I}}^{m-1}=\Omega(1).
\]
Writing $k = n/2 + (\sqrt{n}/2) t$, we have the elementary binomial tail lower bound $1-p_{\mathrm{I}} \geq \exp(-O(t^2))$ (see, e.g.,~\cite[inequality~(4.2)]{LedouxTalagrand})); hence $t \geq \Omega(\sqrt{\ln m})$.  The desired bound
\[
    p_{\mathrm{S}} = {n \choose k}/2^n = \Omega(t/(m \sqrt{n}))
\]
now follows from asymptotically tight estimates (up to universal constants for all $0 \leq t \leq \sqrt{n}$) for the Mills ratio of the binomial distribution; see~\cite{McK89}.
\end{proof}

\ignore{
%\red{See Daniel's Lemma 7; roughly speaking we should have $m=\Theta(m)$,
%\[
%\text{his~}\eps = \Theta(1/m),
%\]
%\[
%p_S = {L \choose k}/2^L = \Omega(\eps \sqrt{\log(1/\eps)}/\sqrt{L}) = \Omega((1/m) \cdot \sqrt{\log m}/\sqrt{L}),
%\]
%\[
%p_I = {L \choose < k}/2^L = 1-\Theta(1/m) \text{~or so, so~}p_I^{m-1}=\Omega(1) \text{~and~}
%m p_S p_I^{m-1} = \Omega(\sqrt{\log m}/\sqrt{L}).
%\]
%}
}

\end{document}